\documentclass[11pt]{article}

\usepackage[T1]{fontenc}
\usepackage{libertinus}
\usepackage{libertinust1math}
\usepackage{microtype}
\microtypesetup{expansion=false}
\usepackage[margin=1.06in]{geometry}
\usepackage{amsmath,amssymb,amsthm,mathtools}
\usepackage{graphicx}
\usepackage{float}
\usepackage{array}
\usepackage{xcolor}
\usepackage[colorlinks=true,
            linkcolor=blue!45!black,
            citecolor=blue!45!black,
            urlcolor=blue!45!black,
            pdftitle={On Trace Averages and Sharp Asymptotics for Powers of Haar-Unitary Matrices}]{hyperref}
\usepackage[nameinlink,noabbrev]{cleveref}

\theoremstyle{plain}
\newtheorem{theorem}{Theorem}
\newtheorem{proposition}{Proposition}
\newtheorem{lemma}{Lemma}
\newtheorem{corollary}{Corollary}

\theoremstyle{remark}
\newtheorem{remark}{Remark}

\DeclareMathOperator{\Tr}{Tr}
\DeclareMathOperator{\wt}{wt}
\newcommand{\U}{\mathrm U}
\newcommand{\F}{\mathcal F}
\newcommand{\TV}{\mathrm{TV}}
\newcommand{\dd}{\,\mathrm d}
\newcommand{\e}{\mathrm e}
\newcommand{\ii}{\mathrm i}
\newcommand{\C}{\mathbb C}
\newcommand{\R}{\mathbb R}
\newcommand{\N}{\mathbb N}
\newcommand{\1}{\mathbf 1}
\newcommand{\norm}[1]{\lVert #1\rVert}

\newcommand{\Ch}{\mathrm{Ch}}
\newcommand{\starD}{\mathbin{\star_D}}

\usepackage{titling}
\title{On Trace Averages and Sharp Asymptotics for Powers of Haar-Unitary Matrices}
\author{
  Ayesha Irfan\\[0.1em]
  \small\href{mailto:ayesha.irfan@bristol.ac.uk}
             {\nolinkurl{ayesha.irfan@bristol.ac.uk}}
}
\date{}

\begin{document}

\maketitle
\vspace{-3.65em}

\begin{abstract}
For a Haar-distributed matrix $U \in \U(N)$, Rains's power-map theorem represents the spectrum of $U^m$, $m \in \mathbb{N}$, as a superposition
of independent lower-dimensional CUE spectra.  We determine how this exact
decomposition appears in trace moments and in sharp Gaussian and
Poisson limit laws. For averages of mixed trace products, a
Charlier--Ewens convolution separates the stable Wick term from its full
finite-$N$ correction. A Fredholm factorisation gives sharp
total-variation asymptotics for
the normalised trace in four regimes, from fixed powers to $m\geq N$.
For a fixed trace power $p$, coherent summation of the leading Laguerre correction yields the
exact factor $\e^{2/p}$. At the mean-spacing scale, we obtain a sharp Poisson
limit for the unfolded spectrum of $U^m$ on a fixed interval as $N, m \rightarrow \infty$; its first-order correction, obtained via the second Charlier polynomial, quantifies the residual under-dispersion of local eigenvalue counts.
\end{abstract}

\setcounter{tocdepth}{2}
\vspace{-1em}
\tableofcontents
\vspace{-1.5em}

\ifdefined\PedagogicalVersion\PedPreface\fi

\section{Introduction}
\label{sec:introduction}

\subsection{Power maps and the emergence of independence}

Let $U_N$ be an $N \times N$ Haar-distributed unitary matrix; it has eigenvalues $e^{i\theta_1}, \ldots, e^{i\theta_N}$ that lie on the unit circle with $\theta_1, \ldots, \theta_N \in [0, 2\pi)$. Each eigenvalue has a uniform marginal distribution, but jointly the eigenvalues exhibit strong repulsion, which produces a configuration far more evenly spaced than independent uniform points. The power map
\begin{equation*}
 U_N\longmapsto U_N^m
\end{equation*}
acts on the eigenangles by $\theta\mapsto m\theta\pmod{2\pi}$ with $m \in \mathbb{N}$.  It folds the
circle over itself $m$ times, so that $m$ distinct points share a common image. After folding, one might expect a complicated rearrangement of CUE correlations. Instead, Rains showed that the spectrum decomposes cleanly into independent unitary spectra \cite{kn:rains03}. In particular, once $m \geq N$, the eigenvalues of $U_N^m$ are distributed exactly as $N$
independent uniform points on the unit circle \cite{kn:rains97}. 

Let
\begin{equation}\label{eq:euclidiean-division}
N=mq+r,\qquad 0\leq r<m,
\end{equation}
be the Euclidean division of $N$ by $m$. Rains proved the distributional identity
\begin{equation}\label{eq:rains-decomposition}
\operatorname{spec}(U_N^m)\ \overset d=
\bigsqcup_{a=1}^{r}\operatorname{spec}(V_{q+1,a})
\ \sqcup\!
\bigsqcup_{a=r+1}^{m}\operatorname{spec}(V_{q,a}),
\end{equation}
where the matrices on the right-hand side are mutually independent and Haar distributed in the indicated unitary groups. Thus, in law, the spectrum of $U_N^m$ is the superposition of $r$ independent CUE spectra of size $q+1$ and $(m-r)$ of size $q$. We refer to the spectrum of each such independent matrix as a block. After zero-dimensional blocks are discarded, there are $\min(m,N)$ such blocks, each containing at most $\lceil N/m\rceil$ eigenvalues. 

The decomposition identifies exactly which correlations the power map removes. Eigenvalues belonging to different blocks are independent, whereas those within the same block retain the characteristic repulsion of a CUE spectrum. When $m=o(N)$, the blocks remain large and each block retains CUE local statistics. When $m\asymp N$, their dimensions remain bounded while their number grows proportionally to $N$. At the endpoint $m\geq N$, precisely $N$ one-dimensional blocks remain and exact independence follows.
\begin{figure}[H]
    \centering
\includegraphics[width=13.7cm]{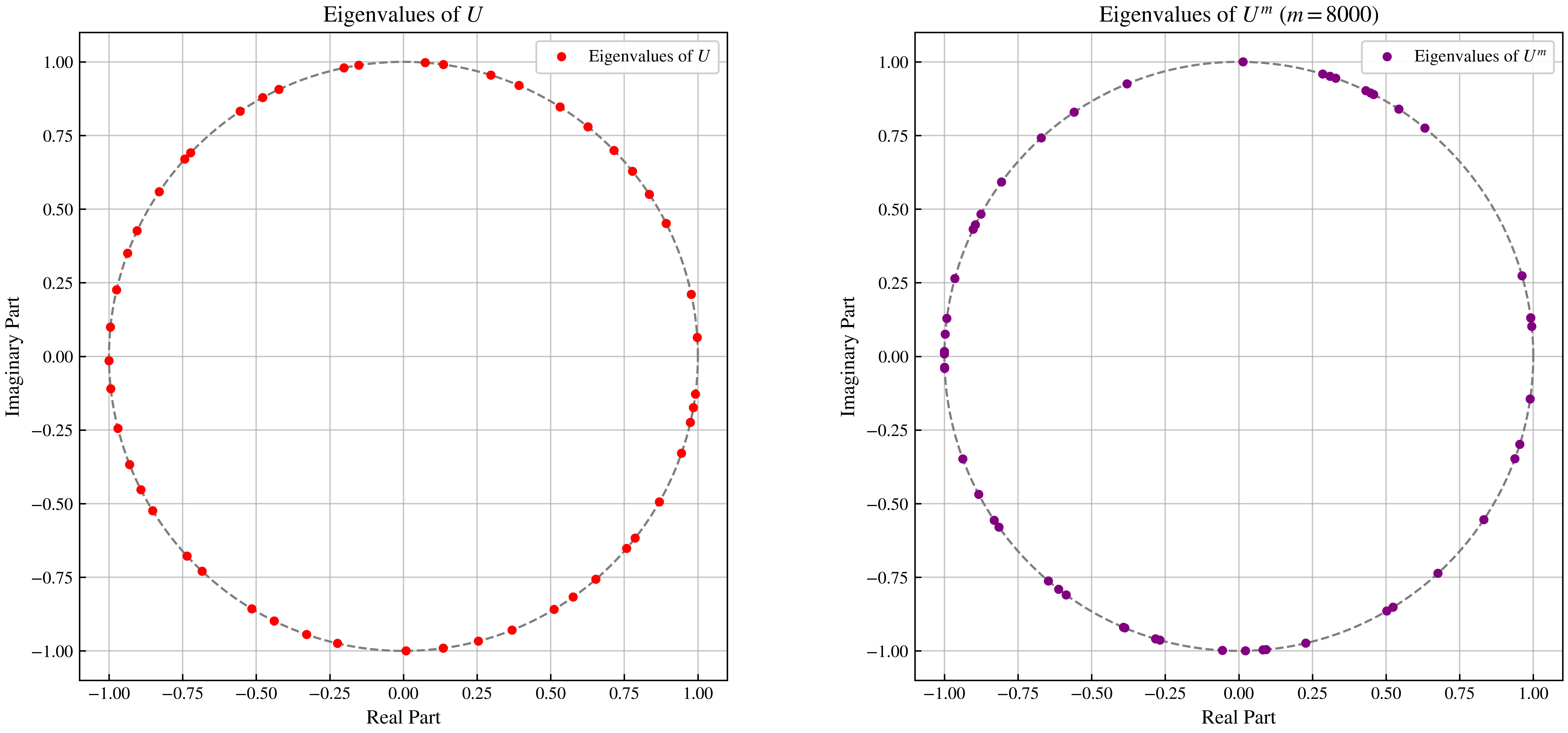}
    \caption{A comparison of the  eigenvalue distribution for $U$ and $U^{m}$ with $m \geq N$ and $U \in U(N)$.}
    \label{fig:galaxy}
\end{figure}
\newpage
\subsection{Trace moments: exact identities and finite-$N$ corrections}
\label{subsec:trace-moments-intro}

To convert this spectral picture into a global scalar statistic, consider the empirical spectral measure and its Fourier transform. For eigenvalues
$\e^{\ii\theta_1},\ldots,\e^{\ii\theta_N}$, set
\begin{equation}\label{eq:intro-trace-fourier}
 \mu_{U_N}:=\frac1N\sum_{a=1}^N\delta_{\theta_a},
 \qquad
 \hat\mu_{U_N}(m)
 =\frac1N\sum_{a=1}^N\e^{\ii m\theta_a}
 =\frac1N\Tr(U_N^m).
\end{equation}
Thus $N^{-1}\Tr(U_N^m)$ is the $m$th Fourier coefficient of $\mu_{U_N}$; it averages the unit complex numbers
$\e^{\ii m\theta_a}$ and its magnitude measures the cancellation in their
complex-vector sum. The trace in \eqref{eq:intro-trace-fourier} belongs to a classical theory of
exact CUE averages. In 1962, Dyson's analysis of the CUE kernel
\cite{kn:dyson62iii} gave the covariance identity
\begin{equation}\label{eq:intro-trace-covariance}
 \mathbb E\!\left[\Tr(U_N^j)\overline{\Tr(U_N^k)}\right]
 =\mathbf 1_{\{j=k\}}( \, j\wedge N),
 \qquad j,k\geq1,
\end{equation}
where $a\wedge b:=\min(a,b)$. 

Later, representation theory provided a natural language for arbitrary mixed trace averages.
A product of traces is a power-sum symmetric function of the eigenvalues, and
Frobenius's formula decomposes it into Schur characters \cite{kn:macdonald1995} indexed by partitions of length at most $N$. As these
characters are orthonormal under Haar measure on $\U(N)$, the unitary integral
can be evaluated exactly as a finite character sum. Explicitly, let
$\mathbf p=(p_1,\ldots,p_{k_1})$ and
$\mathbf q=(q_1,\ldots,q_{k_2})$ be positive integer vectors and let
\begin{equation}\label{eq:intro-exponent-types}
 \mu=(1^{a_1}2^{a_2}\cdots),\qquad
 \nu=(1^{b_1}2^{b_2}\cdots)
\end{equation}
be their exponent types; thus $a_j=\#\{i:p_i=j\}$ and
$b_j=\#\{i:q_i=j\}$. Write
\begin{equation}\label{eq:intro-CN-definition}
 C_N(\mathbf p,\mathbf q)
 :=\mathbb E\!\left[\prod_{i=1}^{k_1}\Tr(U_N^{p_i})
 \prod_{j=1}^{k_2}\overline{\Tr(U_N^{q_j})}\right],
\end{equation}
and note that multiplication of $U_N$ by a scalar phase shows immediately that $ C_N(\mathbf p,\mathbf q)=0$ unless $\sum_{i=1}^{k_1}p_i
 =\sum_{j=1}^{k_2}q_j=:D$. In this balanced case, Frobenius's
character formula and Schur orthogonality give the exact finite character
sum
\begin{equation}\label{eq:intro-character-trace-sum}
 C_N(\mathbf p,\mathbf q)
 =\sum_{\substack{\lambda\vdash D \\\ell(\lambda)\leq N}}
 \chi^\lambda_\mu\chi^\lambda_\nu.
\end{equation}
Here, $\chi^\lambda_\mu$ is the irreducible character of the symmetric group $S_D$ labelled by the partition $\lambda$, evaluated on any permutation of cycle type $\mu$. While explicit, the form \eqref{eq:intro-character-trace-sum} can be difficult to work with for large $D$ due to the algebraic complexity of the characters and the growing number of partitions. 

A separate approach to Haar integration originates in Weingarten’s 1978 formula, which operates directly at the level of matrix entries \cite{kn:weingarten78}. It expresses integrals of products of unitary matrix entries and their complex conjugates as sums over permutations; this approach was later developed into a systematic integration calculus by Collins and \'{S}niady \cite{kn:collinssniady06}, providing new expressions for \eqref{eq:intro-CN-definition}. In 1994, Diaconis and Shahshahani \cite{kn:diasha} identified the decisive simplification of the character sum above, with the range later corrected in \cite{kn:diaconisevans01}. If
$\mathbf a=(a_j)_{j\geq1}$ and $\mathbf b=(b_j)_{j\geq1}$ are finitely
supported and
$\sum_{j\geq1}j(a_j+b_j)\leq2N$, then
\begin{equation}\label{eq:intro-DS-history}
  C_N(\mathbf p,\mathbf q)
 =\1_{\{\mathbf a=\mathbf b\}}
 \prod_{j\geq1}j^{a_j}a_j! . 
\end{equation}
The condition $\sum_{j\geq1}j(a_j+b_j)\leq2N$ implies every partition $\lambda \vdash D$ automatically satisfies $\ell(\lambda)\leq N$. The length restriction in
\eqref{eq:intro-character-trace-sum} therefore discards no contributions, and 
the full character orthogonality relation gives
$\sum_{\lambda\vdash D}\chi^\lambda_\mu\chi^\lambda_\nu
= \mathbf1_{\{\mu=\nu\}} \prod_{j\geq1}j^{a_j}a_j! \,$ recovering \eqref{eq:intro-DS-history}. We refer to $D \leq N$ as the stable range, i.e. the range in which the simplified form \eqref{eq:intro-DS-history} holds; here, the non-vanishing condition $\mathbf{a} = \mathbf{b}$ implies that $\mathbf{p}$ and $\mathbf{q}$ coincide up to permutation and $k_1=k_2:=k$. Outside this range, $k_1$ and $k_2$ may differ in contributing terms and $C_{N}(\mathbf{p}, \mathbf{q})$ takes a correspondingly more complex form. The identity \eqref{eq:intro-DS-history} has a natural Gaussian interpretation, which is made transparent by Wick's rule. Let
$(G_j)_{j\geq1}$ be independent centred complex Gaussians satisfying
\begin{equation}\label{eq:intro-gaussian-trace-family}
 \mathbb E[G_j\overline{G_k}]=j\,\mathbf1_{\{j=k\}}, \qquad \mathbb E[G_j^2]=0.
\end{equation}
Then, by Wick's rule, one has
\begin{equation}\label{eq:intro-DS-as-wick}
 \mathbb E\!\left[
   \prod_{i=1}^{k}G_{p_i}
   \prod_{j=1}^{k}\overline{G_{q_j}}
 \right]
 =
 \sum_{\sigma\in S_k}
 \prod_{i=1}^{k}
 p_i\,\mathbf 1_{\{p_i=q_{\sigma(i)}\}}
 =
 \mathbf 1_{\{\mathbf a=\mathbf b\}}
 \prod_{j\geq1}j^{a_j}a_j!,
\end{equation}
and thus \eqref{eq:intro-DS-history} states that
$C_N(\mathbf p,\mathbf q)$ agrees exactly with the corresponding Gaussian
moment throughout the stable range. This range must not, however, be confused with exact Gaussianity of the traces themselves.  For fixed $N$, we have 
$|\Tr(U_N^m)|\leq N$, whereas a non-degenerate complex Gaussian has
unbounded support. The assertion instead concerns a growing initial collection of
moments. In particular, for fixed $m$,
\begin{equation}\label{eq:intro-single-power-stable-moments}
 \mathbb E|\Tr(U_N^m)|^{2k}=m^k k!
 \qquad\text{whenever}\qquad
 k\leq\left\lfloor\frac Nm\right\rfloor.
\end{equation}

Cumulants provide a complementary description of the transition beyond the stable range. For non-zero integers $r_1,\ldots,r_s$, let $\kappa_s^{(N)}(r_1,\ldots,r_s)$ denote the joint cumulant of
$\Tr(U_N^{r_1}),\ldots,\Tr(U_N^{r_s})$, with
$\Tr(U_N^{-r}):=\overline{\Tr(U_N^{r})}$ for $r \geq 1$; as with the moments,
it vanishes unless $\sum_a r_a=0$. Any cumulants of order at least three vanish identically for a jointly Gaussian family, so their survival measures a departure from Gaussian behaviour. Soshnikov expressed $\kappa_s^{(N)}(r_1,\ldots,r_s)$ as an alternating sum over ordered set partitions \cite{kn:soshnikov00}. When
\begin{equation*}
 \sum_a|r_a|\leq2N,
\end{equation*}
every such higher cumulant vanishes, leaving only the covariance and hence Wick's rule. Beyond this range, Soshnikov and Wu analysed the surviving terms from higher cumulants for powers of order $N$ \cite{kn:soshnikovwu23}. Recovering a full mixed moment from the cumulants, however, still requires summing over set partitions via the
moment--cumulant relation.

Other approaches capture trace averages through different structures.
Pastur and Vasilchuk \cite{kn:pasturvasilchuk04} derived integration-by-parts recurrences for trace
moments in the classical compact groups, while Adler and van Moerbeke \cite{kn:adlervanmoerbeke01}
related generating functions for mixed trace moments to the Toeplitz lattice, a reduction of the $2$-Toda hierarchy. Jiang and Matsumoto \cite{kn:jiangmatsumoto15} obtained related
power-sum moment estimates with fixed powers for the circular $\beta$-ensembles, including an exact finite-$N$ expansion in Jack-polynomial coefficients. We separate the Wick term from the finite-$N$ correction as follows:
\begin{equation}\label{eq:intro-wick-correction}
 C_N(\mathbf p,\mathbf q)
 =
 \mathbf1_{\{\mathbf a=\mathbf b\}}
 \prod_{j\geq1}j^{a_j}a_j!
 +\mathcal R_N(\mathbf p,\mathbf q),
\end{equation}
where $\mathcal R_N(\mathbf p,\mathbf q)=0$ for $D\leq N$. Here, the Wick term is zero if $k_1 \neq k_2$. Theorem~\ref{thm:charlier-correction} resolves $\mathcal R_N$ into
finitely many Charlier--Ewens correction bands; the bands are indexed by $\ell=1,2,\ldots$, and only those satisfying
\begin{equation}\label{eq:intro-exterior-cutoff}
 \ell(N+\ell)\leq D
\end{equation}
can contribute, which identifies both the onset and the finite depth of
the correction. Theorem~\ref{thm:displacement} evaluates $C_N(\mathbf p,\mathbf q)$ by a finite sum of determinants that bypasses the partition sum.

\subsection{Global trace fluctuations and total variation}
\label{subsec:global-TV}

As Haar invariance gives
$\mathbb E\Tr(U_N^m)=0$, \eqref{eq:intro-trace-covariance} yields
\begin{equation}\label{eq:intro-rains-trace-sum}
 X_{N,m}:=\frac{\Tr(U_N^m)}{\sqrt{m\wedge N}},
 \qquad \mathbb E|X_{N,m}|^2=1.
\end{equation}
This is the natural unit-variance normalisation of the $m$th trace.

Although Rains's power-map theorem \eqref{eq:rains-decomposition} was used above to motivate the
spectral transition, it also has an immediate consequence for traces.
Taking traces in \eqref{eq:rains-decomposition} gives
\begin{equation}\label{eq:intro-rains-trace-decomposition}
 X_{N,m}\ \overset d=\ \frac{1}{\sqrt{m\wedge N}}
 \left(
   \sum_{a=1}^{r}\Tr(V_{q+1,a})+
   \sum_{a=r+1}^{m}\Tr(V_{q,a})
 \right).
\end{equation}
Equation~\eqref{eq:intro-rains-trace-decomposition} reveals several routes to
Gaussian behaviour. If $m=p$ is fixed, each block trace converges to a complex Gaussian and the first
non-Gaussian moment of $X_{N,p}$ is pushed to an order proportional to $N$. If $m\asymp N$ and $q=\lfloor N/m\rfloor$ is fixed,
$X_{N,m}$ is a normalised sum of increasingly many independent traces of bounded-size
Haar matrices, so a central limit theorem for triangular arrays applies. If both
$m$ and $q$ diverge, the two mechanisms coexist. Finally, for $m\geq N$, it reduces to the classical central limit theorem for a sum of $N$ independent uniform phases. Let $Z$ have
density
\begin{equation}\label{eq:gaussian-density}
 g(x)=\pi^{-1}\e^{-|x|^2},\qquad x\in\C\simeq\R^2.
\end{equation}
In every case $X_{N,m}$ tends to
$Z$ as $N\to\infty$, but the leading total-variation correction remembers
which mechanism produced the Gaussian limit. Duits and Johansson \cite{kn:duitsjohansson10} proved that this Gaussian
limit holds for arbitrary positive integer sequences $m=m(N)$, with joint
convergence to independent standard complex Gaussians for any fixed
number of distinct, separately normalised powers. 

This leads to our first question: how close is the distribution of $X_{N,m}$ to the standard complex Gaussian distribution? We measure the error in total variation,
\begin{equation}\label{eq:intro-TV-definition}
 d_{\TV}(\mu,\nu):=\sup_A|\mu(A)-\nu(A)|
 =\frac12\int|f_\mu-f_\nu|\dd x,
\end{equation}
where the second equality holds if $\mu$ and $\nu$ have probability densities $f_\mu$ and $f_\nu$. For a random variable or vector $Y$, we write $\mathcal L(Y)$ for its probability distribution, or law, and use the shorthand
\begin{equation}
d_{\TV}(Y,W):=d_{\TV}\left(\mathcal L(Y),\mathcal L(W)\right).
\end{equation}
This distance controls the
error for every measurable event; for a comprehensive background on Gaussian approximation of traces in total variation, see Meckes \cite{kn:meckes19}. Johansson \cite{kn:johansson97} proved that, for any fixed coefficients
$\mathbf c=(c_1,\ldots,c_k)\in\C^k$, there exist
$C_{\mathbf c},\delta_{\mathbf c}>0$ such that
\begin{equation}
d_{\TV} \left(
\Re\sum_{j=1}^k c_j\Tr(U_N^j),
\Re\sum_{j=1}^k c_jG_j
\right)
\leq C_{\mathbf c}N^{-\delta_{\mathbf c}N}
\end{equation}
for all sufficiently large $N$. Thus fixed real linear combinations
of the real and imaginary parts of the traces admit super-exponentially
accurate Gaussian approximations. For growing powers, Fulman \cite{kn:fulman12} used Stein’s method and heat-kernel techniques to obtain an $O(\tfrac m N)$ Kolmogorov bound for a normalised real component. Döbler and Stolz \cite{kn:doblerstolz11} extended this approach to multivariate Wasserstein bounds for vectors of traces of matrix powers.

Furthermore, Johansson and Lambert \cite{kn:johanssonlambert21} obtained total variation bounds between $(\Tr(U_N^j))_{j=1}^k$ and the independent Gaussian vector $(G_j)_{j=1}^k$. For fixed $k$, their bounds give super-exponential convergence of the joint trace distribution to the independent Gaussian law, with the factor $\Gamma(N/k+1)^{-1}$ multiplied by explicit dimension-dependent factors. Their control of the dependence on $k$ also allows the number of traces to grow with the matrix size: the joint total-variation error still tends to zero when $k=\lfloor N^\alpha\rfloor$, for every fixed $0<\alpha<1/2$. Courteaut, Johansson and Lambert \cite{kn:courteautjohanssonlambert24} subsequently studied a single trace power with $m=m(N)$. Combining \eqref{eq:intro-rains-trace-decomposition} with Fredholm determinant estimates, they obtained, for $1\leq m\leq N$ and all sufficiently large $N$,
\begin{equation}
d_{\TV}\left(X_{N,m}, Z\right)
\lesssim
\frac{(q+1)^{1/4}\sqrt{\log N}}
{m^q\Gamma(q+2)},
\qquad q=\left\lfloor\frac Nm\right\rfloor,
\end{equation}
with an absolute implicit constant. This bound describes the transition from factorial decay for fixed powers to polynomial decay when the block dimensions remain bounded. In the independent-phase regime $m\geq N$, equation 1.2 in \cite{kn:courteautjohanssonlambert24} states
\begin{equation}
d_{\TV}(X_{N,m},Z)\leq\frac{C}{N},
\qquad m\geq N,
\end{equation}
with an absolute constant $C>0$, where the order $N^{-1}$ is sharp. For $m=1$ and all sufficiently large $N$, they also proved the two-sided bounds
\begin{equation}
\label{eq:intro-p1-fixed-case}
\frac{c }{\Gamma(N+2)\sqrt{N+1}}
\leq d_{\TV}(X_{N,1},Z)
\leq \frac{\tilde{C} N^{1/4}\sqrt{\log N}}{\Gamma(N+2)},
\end{equation}
where $c,\tilde{C} >0$ are absolute constants. Both bounds contain the factorial factor $\Gamma(N+2)^{-1}$, but do not determine the remaining dependence on $N$ or the leading total-variation constant. For the $L^2$ distance between the trace and Gaussian densities when $m=1$, they obtained the full leading asymptotic, including its constant, with decay of order $\Gamma(N+2)^{-1}N^{-1/4}$.

Fix $p\in\N$ and write $N=pq+r$, with $0\leq r<p$. By building on the methods of Courteaut, Johansson and Lambert, our fixed-power theorem determines the complete leading term:
\begin{equation}\label{eq:intro-fixed-power-theorem}
d_{\TV}(X_{N,p},Z)
\sim c_* \frac{(p-r) \, \e^{2/p}}
{p^{q+1}\Gamma(q+2)(q+1)^{1/4}},
\qquad
c_* :=\frac{2^{3/4}\Gamma(3/4)}{\pi^{3/2}},
\end{equation}
as $N\to\infty$, which is consistent with \eqref{eq:intro-p1-fixed-case} when $p=1$. The first altered moment gives the correct factorial scale, but not the
leading constant: the factor $\e^{2/p}$ is produced by coherent summation
over the complete first correction band. The four global theorems in
Section~\ref{sec:main-results} give sharp asymptotics in the regimes above.

\subsection{Local spectral statistics and Poisson approximation}

The first question concerned the global statistic $\Tr(U^m)$ and asked how close
it is to a complex Gaussian. The second unfolds the eigenvalues of $U^m$ at
mean spacing and asks whether the resulting point process is close to a
Poisson process on a fixed interval. Both questions measure how the power map changes spectral statistics, but their limiting mechanisms are different.

For the local problem, we write the eigenvalues of $U_N^m$ as
$\e^{\ii\theta_1},\ldots,\e^{\ii\theta_N}$ and unfold them to mean spacing
one:
\begin{equation}\label{eq:intro-unfolded-process}
 \Xi_{N,m}:=\sum_{j=1}^N\delta_{N\theta_j/(2\pi)}
 \quad\text{on }\mathbb T_N=\mathbb R/N\mathbb Z.
\end{equation}
For a fixed interval $B$ of length $L$ and a unit-rate Poisson process
$\Pi_B$ on $B$, we study
\begin{equation}\label{eq:intro-local-TV}
 d_{\TV}\bigl(\Xi_{N,m}\restriction _B, \Pi_B \bigr).
\end{equation}

Here $\Xi_{N,m}\restriction_B$ denotes the point process obtained by
retaining only those unfolded eigenvalues whose positions lie in $B$. After unfolding, a block of size $d$ in \eqref{eq:rains-decomposition}
has constant intensity $d/N$, and hence contributes $Ld/N$ eigenvalues
to $B$ on average. The bound $d\leq\lceil N/m\rceil$ implies that these
intensities vanish uniformly as $N,m\to\infty$. At this scale, repulsion
within each CUE block suppresses the probability that two or more
eigenvalues fall in $B$; a block is therefore primarily a zero-or-one source, and their superposition approaches a unit-rate Poisson
process on $B$, even though the
points within one block retain determinantal repulsion.

Decreusefond and Vasseur give quantitative Poisson approximation bounds
for superpositions of independent repulsive point processes
\cite[Theorem~5.7]{kn:decreusefondvasseur18}. Applied to Rains's blocks,
their bound yields
\begin{equation}
 d_{\TV}\bigl(\Xi_{N,m}\restriction _B,\Pi_B\bigr)
 =O_L\bigl((m\wedge N)^{-1}\bigr).
\end{equation}
Our theorem identifies the exact leading term within this scale, including its constant and its dependence on the block dimensions. If $N=mq+r$, the block intensity fractions are $(q+1)/N$ for $r$ blocks and
$q/N$ for the remaining $m-r$. Their quadratic mass is
\begin{equation}\label{eq:intro-beta}
 \beta_{N,m}=\frac{r(q+1)^2+(m-r)q^2}{N^2}.
\end{equation}
Theorem~\ref{thm:local-poisson} proves, uniformly whenever $N,m\to\infty$,
\begin{equation}\label{eq:intro-poisson-summary}
 d_{\TV}\bigl(\Xi_{N,m}\restriction _B,\Pi_B\bigr)
 \sim A(L)\beta_{N,m},\qquad
 A(L):=\frac14\,\mathbb E\bigl|(Z_L-L)^2-Z_L\bigr|,
\end{equation}
where $Z_L$ is Poisson with mean $L$. This determines the leading
total-variation error for the full eigenvalue configuration in $B$, with no
restriction on the relative growth rates of $N$ and $m$. The local eigenvalue count
$\Xi_{N,m}(B)$ has the same leading term and is slightly under-dispersed
relative to $Z_L$.

\subsection*{Brief note on results}

The main contributions fall into three categories.  The first concerns
exact averages of traces beyond the stable range; the other two determine
sharp probability distances at the global and local scales.
\begin{center}
\begingroup
\small
\renewcommand{\arraystretch}{1.18}
\begin{tabular}{@{}
>{\raggedright\arraybackslash}p{0.20\textwidth}
>{\raggedright\arraybackslash}p{0.29\textwidth}
>{\raggedright\arraybackslash}p{0.41\textwidth}@{}}
\hline
\textbf{Problem} & \textbf{Principal representation} & \textbf{Consequence}\\
\hline
Mixed trace moments
& Charlier--Ewens convolution and displacement determinant
& Exact evaluation for arbitrary mixed exponents and an explicit separation
  of the Wick term from $\mathcal R_N$ and its correction bands.\\[0.35em]
Global Gaussian law
& Fredholm factorisation and Fourier--Laguerre transfer
& Sharp total-variation asymptotics in four power regimes, including the
  coherent factor $\e^{2/p}$ for fixed powers.\\[0.35em]
Local Poisson law
& Rains blocks and a Poisson--Charlier expansion
& A sharp process-level and count-level asymptotic, uniform whenever
  $N,m\to\infty$.\\
\hline
\end{tabular}
\endgroup
\end{center}

Section~\ref{sec:main-results} states the main probability results and
the exact mixed trace identities. Section~\ref{sec:exact-identities-tools}
proves the identities using Frobenius's formula, Bernstein operators,
and Heine--Szeg\H{o} \cite{kn:macdonald1995,kn:bumdia02}. It also develops the
Bessel--Fredholm representations and Fourier--Laguerre tools used in the probability results. The
Fredholm formulas use the Borodin--Okounkov--Geronimo--Case identity
\cite{kn:borodinokounkov00, kn:basorwidom00}, while the determinant
estimates come from Courteaut, Johansson and Lambert
\cite{kn:courteautjohanssonlambert24}. Sections~\ref{sec:fixed-power}--\ref{sec:high-power} prove the four global
Gaussian asymptotics, with the fixed-power analysis also using
Szeg\H{o}'s Laguerre asymptotics \cite{kn:szego75}.
Section~\ref{sec:local-poisson} proves the local Poisson theorem by
combining Rains's decomposition \cite{kn:rains03}, a Poisson--Charlier
expansion, and local CUE kernel estimates controlling the point
configuration.
\section{Main results}
\label{sec:main-results}

We identify $\C$ with $\R^2$ and retain the notation $Z$, $g$, and
$d_{\TV}$ introduced in Section~\ref{sec:introduction}. Throughout this
section, $U_N$ is Haar distributed on $\U(N)$. For $N,m\in\N$, write
\begin{equation}
\label{eq:quotient-remainder}
 N=mq+r,\qquad 0\leq r<m,\qquad
 \nu=q+1,\qquad s=m-r.
\end{equation}
In the fixed-power regime we write $m=p$. We use
$L_k=L_k^{(0)}$ for the ordinary Laguerre polynomial. Subscripts in our
asymptotic notation indicate dependence on fixed parameters: constants
in $O_L(\cdot)$ may depend on the interval length $L$, and constants in
$O_q(\cdot)$ may depend on the fixed quotient $q$.

\subsection*{Gaussian and Poisson approximation}

\begin{theorem}[Fixed powers]\label{thm:fixed-power}
Fix $p\in\N$ and take $m=p$ in the notation above. Then, as $N\to\infty$,
\begin{align}
 d_{\TV}\!\left(\frac{\Tr(U_N^p)}{\sqrt p},Z\right)
 &=
 \frac{(p-r)\e^{2/p}}{2p^\nu\Gamma(\nu+1)}
 \int_0^\infty \e^{-u}|L_\nu(u)|\dd u
 +o\!\left(\frac1{p^\nu\Gamma(\nu+1)\nu^{1/4}}\right)
 \label{eq:fixed-laguerre-form}\\
 &\sim
 \frac{2^{3/4}\Gamma(3/4)}{\pi^{3/2}}\,
 \frac{(p-r)\e^{2/p}}
 {p^\nu\Gamma(\nu+1)\nu^{1/4}}.
 \label{eq:fixed-explicit-form}
\end{align}
The error terms are uniform in the finitely many residues
$r\in\{0,\ldots,p-1\}$.
\end{theorem}
The first altered radial moment is \(\mathbb E|X_{N,p}|^{2\nu}\), but it does not by itself give the leading total-variation constant. Every fixed shift
$L_{\nu+j}$ in the first correction band has the same leading oscillatory
profile. Relative to the coefficient of $L_\nu$, their coefficients
converge to $(2/p)^j/j!$, and hence add up to $\e^{2/p}$. 

\begin{theorem}[Independent-phase regime]\label{thm:high-power}
Let $m=m_N\geq N$. Then, as $N\to\infty$,
\begin{equation}\label{eq:high-power-main}
 d_{\TV}\!\left(\frac{\Tr(U_N^m)}{\sqrt N},Z\right)
 =
 \frac{\e^{-2}\bigl(\sqrt2\cosh(\sqrt2)-\sinh(\sqrt2)\bigr)}{N}
 +O(N^{-2}).
\end{equation}
The estimate is uniform over all choices of $m_N\geq N$.
\end{theorem}

\begin{theorem}[Joint growth of power and block dimension]
\label{thm:joint-growth}
Let $m=m_N<N$ and suppose that $m\to\infty$ and $q\to\infty$.
Then, uniformly in the remainder $r$,
\begin{align}
 d_{\TV}\!\left(\frac{\Tr(U_N^m)}{\sqrt m},Z\right)
 &=
 \frac{m-r}{2m^\nu\Gamma(\nu+1)}
 \int_0^\infty \e^{-u}|L_\nu(u)|\dd u
 +o\!\left(
 \frac{m-r}{m^\nu\Gamma(\nu+1)\nu^{1/4}}
 \right)
 \label{eq:joint-laguerre-form}\\
 &\sim
 \frac{2^{3/4}\Gamma(3/4)}{\pi^{3/2}}
 \frac{m-r}{m^{q+1}\Gamma(q+2)(q+1)^{1/4}}.
 \label{eq:joint-explicit-form}
\end{align}
No relation between the rates at which $m$ and $q$ diverge is required.
\end{theorem}

The assumptions of \cref{thm:joint-growth} are equivalently $m\to\infty$
and $m=o(N)$. Both the fixed-power and joint-growth regimes therefore have
diverging block dimensions. Their leading terms have the same form,
except for the factor $\e^{2/p}$ when the power is fixed, consistently
with $\e^{2/m}\to1$ when $m\to\infty$.

\begin{theorem}[Fixed block dimension]
\label{thm:fixed-quotient}
Suppose that $q\geq1$ is fixed and $m=m_N\to\infty$. Set
\begin{equation}\label{eq:fixed-quotient-AB}
 A_{m,r}=\frac{s}{m^\nu\Gamma(\nu+1)},
 \qquad
 B_{m,r}=\frac{r}{m^{\nu+1}\Gamma(\nu+2)}.
\end{equation}
Then, uniformly in $0\leq r<m$,
\ifdefined\PedagogicalVersion\begingroup\small\fi
\begin{align}
 d_{\TV}\!\left(\frac{\Tr(U_N^m)}{\sqrt m},Z\right)
 &=
 \frac12\int_0^\infty \e^{-u}
 \left|
   (-1)^q A_{m,r}L_\nu(u)
   +(-1)^{q+1}B_{m,r}L_{\nu+1}(u)
 \right|\dd u
 +o_q\!\left(\frac{s+1}{m^\nu}\right).
 \label{eq:fixed-quotient-two-mode}
\end{align}
\ifdefined\PedagogicalVersion\endgroup\fi
In particular, if $r/m\to\theta\in[0,1)$, then
\begin{equation}\label{eq:fixed-quotient-interior}
 d_{\TV}\!\left(\frac{\Tr(U_N^m)}{\sqrt m},Z\right)
 \sim
 \frac{1-\theta}{2m^q\Gamma(q+2)}
 \int_0^\infty\e^{-u}|L_{q+1}(u)|\dd u.
\end{equation}
If instead $s\to c\in\N$, then
\begin{align}
 d_{\TV}\!\left(\frac{\Tr(U_N^m)}{\sqrt m},Z\right)
 &\sim
 \frac1{2m^{q+1}}
 \int_0^\infty\e^{-u}
 \left|
   \frac{c}{\Gamma(q+2)}L_{q+1}(u)
   -\frac1{\Gamma(q+3)}L_{q+2}(u)
 \right|\dd u.
 \label{eq:fixed-quotient-endpoint}
\end{align}
\end{theorem}

\begin{theorem}[Local Poisson approximation]
\label{thm:local-poisson}
For $j\geq2$, set
\begin{equation}\label{eq:main-local-gamma}
 \gamma_{j;N,m}
 =
 \frac{r(q+1)^j+(m-r)q^j}{N^j},
 \qquad
 \varepsilon_{N,m}
 =
 \frac{q+\mathbf1_{\{r>0\}}}{N}.
\end{equation}
Thus $\gamma_{2;N,m}=\beta_{N,m}$ from \eqref{eq:intro-beta}.
Fix $L>0$ and an interval $B$ of length $L$ on the circle
$\mathbb T_N=\mathbb R/N\mathbb Z$ of circumference $N$.
Let $\Xi_{N,m}$ and $\Pi_B$ be as in
\eqref{eq:intro-unfolded-process}--\eqref{eq:intro-local-TV}, let
$Z_L\sim\operatorname{Pois}(L)$, and use the constant $A(L)$ defined in
\eqref{eq:intro-poisson-summary}.
There is a constant $C_L<\infty$ such that, whenever
$N>L$ and $L\varepsilon_{N,m}\leq\frac14$,
\begin{align}
 &\left|
 d_{\TV}\!\left(
   \Xi_{N,m}\!\restriction_B,\Pi_B
 \right)
 -A(L)\beta_{N,m}
 \right|
 \notag\\
 &\quad\quad\quad+
 \left|
 d_{\TV}\!\left(
   \Xi_{N,m}(B),Z_L
 \right)
 -A(L)\beta_{N,m}
 \right|
 \leq
 C_L\left(
   \gamma_{3;N,m}+\beta_{N,m}^2+\gamma_{4;N,m}
 \right).
 \label{eq:main-local-poisson}
\end{align}
Here we regard $\Xi_{N,m}\!\restriction_B$ and $\Pi_B$ as random elements
of the configuration space $\mathcal N(B)$ of finite simple counting
measures on $B$. Consequently, if $N\to\infty$ and $m=m_N\to\infty$,
with no restriction on their relative rates, each of the two
total-variation distances in \cref{eq:main-local-poisson} is asymptotic
to $A(L)\beta_{N,m}$.
\end{theorem}

The remainder $r$ counts the Rains blocks of size $q+1$, while $s$
counts those of size $q$. For fixed $q$, the $L_{\nu+1}$ term is lower
order when $r/m$ stays bounded away from one. When $s$ remains bounded,
the two adjacent Laguerre modes have comparable size; this is why
\cref{thm:fixed-quotient} retains both. The local theorem depends on the same block dimensions through
$\gamma_{j;N,m}$. Note that Courteaut, Johansson and Lambert~\cite{kn:courteautjohanssonlambert24}
prove quantitative bounds which interpolate between the fixed-power
and high-power regimes. In particular, their upper bound for fixed $p$
is of order
$q^{1/4}\sqrt{\log N}/(p^q\Gamma(q+2))$, whereas
\cref{thm:fixed-power} identifies the leading term, which is smaller
by a factor of order $q^{-1/2}(\log N)^{-1/2}$. For $m_N\geq N$,
their $O(N^{-1})$ estimate is sharpened by \cref{thm:high-power}
to an explicit constant and an $O(N^{-2})$ remainder.

\subsection*{Structure of the leading corrections}

We prove each global theorem at the density level. If $f_{N,m}$ denotes
the density of $X_{N,m}$ from \eqref{eq:intro-rains-trace-sum}, the common target is an expansion
\begin{equation}\label{eq:leading-density-template}
 f_{N,m}=g+\Lambda_{N,m}+r_{N,m},
 \qquad \norm{r_{N,m}}_1=o(\norm{\Lambda_{N,m}}_1).
\end{equation}
The definition of total variation and the reverse triangle inequality
then give
\begin{equation}\label{eq:TV-from-signed-correction}
 d_{\TV}(X_{N,m},Z)
 =\frac12\norm{f_{N,m}-g}_1
 =\frac12\norm{\Lambda_{N,m}}_1
 +o(\norm{\Lambda_{N,m}}_1).
\end{equation}
The correction $\Lambda_{N,m}$ has integral zero, so its positive and
negative parts have equal mass. By
\eqref{eq:TV-from-signed-correction}, the leading total-variation error
is half its $L^1$ norm. The contributions to $\Lambda_{N,m}$ must
therefore be summed before taking absolute values: a termwise
triangle inequality gives an upper bound that need not preserve the
leading constant.

In the fixed-power and joint-growth theorems, the first-band sum is
approximated by a multiple of $g(x)L_\nu(|x|^2)$ with relative
$L^1$ error $o(1)$. Radial integration gives
\begin{equation}
 \norm{g(\cdot)L_\nu(|\cdot|^2)}_1
 =
 \int_0^\infty \e^{-u}|L_\nu(u)|\dd u.
\end{equation}
By \cref{lem:laguerre-L1}, as $\nu\to\infty$,
\begin{equation}
 \int_0^\infty \e^{-u}|L_\nu(u)|\dd u
 \sim
 \frac{2^{7/4}\Gamma(3/4)}{\pi^{3/2}}\nu^{-1/4}.
\end{equation}
The exact norm gives the integral expressions in the theorems, while
its large-degree asymptotic gives their explicit equivalents. For fixed $q$, the correction contains the two adjacent modes
$L_{q+1}$ and $L_{q+2}$. If $r/m\to\theta\in[0,1)$, the coefficient
of $L_{q+2}$ is lower order, giving
\eqref{eq:fixed-quotient-interior}. If $s=m-r$ remains bounded, both
coefficients are of order $m^{-(q+1)}$. Their signed combination
must then be retained inside the absolute value, as in
\eqref{eq:fixed-quotient-endpoint}.
In the independent-phase regime, the leading correction has fixed
degree and equals $-g(x)L_2(|x|^2)/(2N)$.

\subsection*{Exact mixed trace moments}
\label{sec:mixed-moments}

\ifdefined\PedagogicalVersion\PedMixedMoments\fi

We use two exact formulas for the mixed trace moments
$C_N(\mathbf p,\mathbf q)$ from \eqref{eq:intro-CN-definition}. The
Charlier--Ewens formula reorganises the character truncation into
the stable Wick term and finite-$N$ correction bands. The displacement
formula expresses the mixed moments as a finite sum of
determinants; it follows by differentiating
the Heine--Szeg\H{o} Toeplitz determinant for the exponential trace
generating function. For analytic estimates, the Borodin--Okounkov
formula rewrites the normalised generating function as a Fredholm
determinant, which we use to bound the remainder in the Gaussian
approximation.

We recall the notation in \eqref{eq:intro-CN-definition} for vectors of
possibly unequal lengths. Haar invariance under multiplication by a scalar phase implies that
$C_N(\mathbf p,\mathbf q)=0$ unless
\begin{equation}
 \sum_{\alpha=1}^{k_1}p_\alpha
    =\sum_{\beta=1}^{k_2}q_\beta := D.
\end{equation}
Let $\mu,\nu\vdash D$ be the partitions having parts
$p_1,\ldots,p_{k_1}$ and $q_1,\ldots,q_{k_2}$, respectively.  Write
\begin{equation}
\label{eq:z-wick-partitions-def}
 \mu=(1^{a_1}2^{a_2}\cdots),
 \qquad
 \nu=(1^{b_1}2^{b_2}\cdots),
 \qquad
 z_\mu=\prod_{j\geq1}j^{a_j}a_j!.
\end{equation}
Let $P_n(x;\alpha)$ denote the monic Charlier polynomial
\begin{equation}\label{eq:charlier-generating}
 \sum_{n=0}^\infty P_n(x;\alpha)\frac{t^n}{n!}
 =\e^{-\alpha t}(1+t)^x.
\end{equation}
Our normalisation follows the Poisson-orthogonal convention in
\cite{kn:chihara1978} throughout.
For a finitely supported multi-index
$\mathbf c=(c_1,c_2,\ldots)$, set
\begin{equation}
 \wt(\mathbf c)=\sum_{j\geq1}jc_j,
 \qquad
 \Phi_{\mathbf c}(\mu)
 =
 \prod_{j\geq1}
 \frac{P_{c_j}(a_j;1/j)}
 {\sqrt{j^{-c_j}c_j!}}.
\end{equation}
Define
\begin{equation}\label{eq:charlier-kernel}
 K_{N,D}^{\Ch}(\mu,\nu)
 =
 \sum_{\wt(\mathbf c)<D-N}
 \Phi_{\mathbf c}(\mu)\Phi_{\mathbf c}(\nu).
\end{equation}
For kernels on partitions of $D$, their Ewens convolution is
\begin{equation}\label{eq:ewens-convolution}
 (A\starD B)(\mu,\nu)
 =
 \sum_{\gamma\vdash D}
 \frac{A(\mu,\gamma)B(\gamma,\nu)}{z_\gamma},
\end{equation}
whose identity is
$I_D(\mu,\nu)=z_\mu\1_{\{\mu=\nu\}}$.  Finally, define
\begin{align}
 \binom{K}{0}_{\star_D}&=I_D,\notag\\
 \binom{K}{\ell}_{\star_D}
 &=
 \frac1{\ell!}
 K\starD(K-I_D)\starD\cdots\starD(K-(\ell-1)I_D),
 \qquad \ell\geq1.                                  \label{eq:convolution-binomial}
\end{align}

Under the Ewens measure
$z_\gamma^{-1}$ on partitions of $D$, the cycle multiplicities have the
joint law of independent Poisson variables with means $1/j$, $1\leq j\leq D$,
conditioned on their total weighted size being $D$ \cite[p.~520, Eq.~(3)]{kn:arratiabarbourtavare92}. Charlier polynomials
are orthogonal for the corresponding unconditioned Poisson laws. The spectral resolution \eqref{eq:charlier-spectral-resolution},
proved below, identifies the character eigenvalues of the truncated
Charlier kernel with the number of shifted diagonal points of a partition.

\begin{theorem}[Charlier--Ewens correction]
\label{thm:charlier-correction}
\label{thm:length-exceedance}
For $\lambda\vdash D$, set
\begin{equation}\label{eq:exceedance-number}
 d_N(\lambda)=\#\{i\geq1:\lambda_i\geq N+i\}, \qquad  s(N,D)=
 \max\{\ell\geq1:\ell(N+\ell)\leq D\},
\end{equation}
with $s(N,D)=0$ if the set is empty. For $ K_{N,D}^{\Ch}(\mu,\nu)$ as defined in \eqref{eq:charlier-kernel} and $z_\mu$ in \eqref{eq:z-wick-partitions-def}, we have 
\begin{equation}\label{eq:wick-plus-charlier}
 C_N(\mathbf p,\mathbf q)
 =
 z_\mu\1_{\{\mu=\nu\}}+\mathcal R_N(\mu,\nu),
\end{equation}
where
\begin{equation}\label{eq:charlier-correction}
 \mathcal R_N(\mu,\nu)
 =
 (-1)^{k_1+k_2+1}
 \sum_{\ell=1}^{s(N,D)}
 (-1)^{\ell+1}
 \binom{K_{N,D}^{\Ch}}{\ell}_{\star_D}(\mu,\nu).
\end{equation}
Thus the Diaconis--Shahshahani term is the identity kernel and the
$\ell$th summand is the $\ell$th correction band.
\end{theorem}

\paragraph{Stable range and the first correction band.}
The finite upper limit in \eqref{eq:charlier-correction} immediately gives
two useful simplifications. If $D\leq N$, then $s(N,D)=0$ and
\begin{equation}\label{eq:DS-from-charlier}
 C_N(\mathbf p,\mathbf q)=z_\mu\1_{\{\mu=\nu\}}.
\end{equation}
The right side is exactly the Wick term in
\eqref{eq:intro-wick-correction}. If
$\mu=(1^{a_1}2^{a_2}\cdots)$, a Wick matching must pair equal exponents,
so it exists only when $a_j=b_j$ for every $j$. In that case, there are
$\prod_j a_j!$ choices, each with covariance product $\prod_j j^{a_j}$,
giving $z_\mu$. In partition notation, the condition $a_j=b_j$ for every
$j$ is precisely $\mu=\nu$. In the next range,
\begin{equation}\label{eq:first-band-range}
 N+1\leq D<2(N+2),
\end{equation}
one has $s(N,D)=1$, so only the $\ell=1$ term contributes and the full correction is linear:
\begin{equation}\label{eq:first-band-linear-correction}
 \mathcal R_N(\mu,\nu)=
 (-1)^{k_1+k_2+1}K_{N,D}^{\Ch}(\mu,\nu).
\end{equation}
The cutoff $s(N,D)$ has a geometric interpretation.
By \eqref{eq:kernel-binomial}, the $\ell$th correction band can
contribute only if some partition $\lambda\vdash D$ satisfies
$d_N(\lambda)\geq\ell$. This requires
$\lambda_\ell\geq N+\ell$. Since partition rows are non-increasing
in length, the first $\ell$ rows must each contain at least
$N+\ell$ boxes. Thus the Young diagram contains a rectangle
with $\ell$ rows and $N+\ell$ columns, giving
\begin{equation}
 D=|\lambda|\geq\ell(N+\ell).
\end{equation}
The minimum is attained by the rectangular partition
$\lambda=((N+\ell)^\ell)$. Consequently, $\ell(N+\ell)$ is
the smallest degree at which the $\ell$th band can appear.

\paragraph{A computational formula.} For direct computation, the following finite determinant is often faster than
the Charlier--Ewens formula: it replaces the character sum by a finite sum of determinants of permutation matrices. It follows
directly from the Heine--Szeg\H{o} determinant by multilinearity; to our
knowledge, no previous source writes it in this explicit form.

\begin{theorem}[Displacement determinant]\label{thm:displacement}
For $\mathbf r\in\{0,\ldots,N-1\}^{k_1}$ and
$\mathbf s\in\{0,\ldots,N-1\}^{k_2}$, define
\begin{equation}\label{eq:displacement}
 h_a(\mathbf r,\mathbf s)
 =
 \sum_{\alpha=1}^{k_1}p_\alpha\1_{\{r_\alpha=a\}}
 -
 \sum_{\beta=1}^{k_2}q_\beta\1_{\{s_\beta=a\}},
 \qquad 0\leq a\leq N-1.
\end{equation}
Then
\begin{equation}\label{eq:displacement-determinant}
 C_N(\mathbf p,\mathbf q)
 =
 \sum_{\mathbf r,\mathbf s}
 \det\!\left[
   \1_{\{b=a-h_a(\mathbf r,\mathbf s)\}}
 \right]_{a,b=0}^{N-1}.
\end{equation}
In particular, every summand is $0$, $1$, or $-1$.
\end{theorem}

Each determinant in \eqref{eq:displacement-determinant} can also be
written in terms of a normalised Vandermonde product;  this form is displayed in the corollary below. 

\begin{corollary} \label{cor:displacement-product}
Let $J_N = \{0, \ldots, N-1\}$ and define
\begin{equation}
 \delta_N(t)=\1_{\{t\in J_N\}}.
\end{equation}
For $\mathbf r, \mathbf s$, and $h_a=h_a(\mathbf r,\mathbf s)$ as above, we have
\begin{equation}\label{eq:displacement-product}
 C_N(\mathbf p,\mathbf q)
 =
 \frac{1}{G(N+1)}
  \sum_{\mathbf r,\mathbf s}
 \left(\prod_{a=0}^{N-1}\delta_N(a-h_a)\right)
 \prod_{0\leq a<b\leq N-1}
 \bigl[(b-a)-(h_b-h_a)\bigr],
\end{equation}
where $G$ is the Barnes $G$-function as defined in \eqref{eq:displacement-normalisation} and empty products are understood to equal one.
\end{corollary}

\section{Exact trace identities and analytic tools}
\label{sec:fourier-laguerre}
\label{sec:exact-identities-tools}

We first prove the exact mixed trace identities stated in Section~\ref{sec:mixed-moments}. We then develop their Bessel and Fredholm specialisations and the analytic estimates used in the global total-variation asymptotics. Two complementary tools connect the exact generating functions to total
variation.  Fourier--Laguerre transfer identifies the leading signed density
correction, while the Bessel Fredholm estimate controls the determinant
remainder uniformly.

In the Fourier
arguments, we repeatedly use
\begin{equation}\label{eq:recurring-radial-notation}
 a=\frac{|t|^2}{4},\qquad \rho=\frac{|t|}{\sqrt m},\qquad \nu=q+1,
\end{equation}
where $m$ denotes the trace power (and $m=p$ in the fixed-power regime).

\subsection*{Exact identity proofs}
\label{sec:proof-charlier-correction}

\begin{proof}[Proof of \cref{thm:charlier-correction}]
We first prove that the Charlier kernel has the
spectral resolution
\begin{equation}\label{eq:charlier-spectral-resolution}
 K_{N,D}^{\Ch}(\mu,\nu)
 =
 \sum_{\lambda\vdash D}
 d_N(\lambda)\chi^\lambda_\mu\chi^\lambda_\nu,
\end{equation}
by a coefficient
argument using Bernstein operators. This avoids choosing or ordering
rim-hook removals: the inclusion--exclusion in each Charlier factor is
performed first, by a translation operator.

Let $\Lambda$ be the ring of symmetric functions with its Hall inner
product. We use $p_j$ to denote
the $j$th power-sum symmetric function. Thus
\begin{equation}
 e_\gamma:=\frac{p_\gamma}{\sqrt{z_\gamma}},
 \qquad \gamma\vdash w,\quad w\geq0,
\end{equation}
is an orthonormal power-sum basis, and the Schur functions are another
orthonormal basis. Frobenius's formula \cite[Chapter~I]{kn:macdonald1995} gives
\begin{equation}\label{eq:frobenius-normalized}
 s_\lambda
 =
 \sum_{\gamma\vdash D}
 \frac{\chi^\lambda_\gamma}{\sqrt{z_\gamma}}e_\gamma.
\end{equation}
Denote by $\Pi_w$ the orthogonal projection onto the homogeneous
component $\Lambda_w$; this means that $\Pi_w$ retains only the terms of
total degree $w$, with $\deg p_j=j$. The local Charlier transform follows directly from
\cref{eq:charlier-generating}: set the polynomial argument equal to
$a$, multiply by $x^a/a!$, sum over $a$, and compare the coefficient
of $t^c$. This gives
\begin{equation}\label{eq:local-charlier-transform}
 \sum_{a=0}^\infty P_c(a;1/j)\frac{x^a}{a!}
 =
 \e^x(x-1/j)^c.
\end{equation}
For a multi-index $\mathbf c$, let
$\rho(\mathbf c)=(1^{c_1}2^{c_2}\cdots)$ and
$w=\wt(\mathbf c)$. Substituting $x=p_j/j$ in
\cref{eq:local-charlier-transform}, multiplying over $j$, and then
taking homogeneous degree $D$ yields the exact Frobenius series
\begin{align}
 F_{\mathbf c}^{(D)}
 &:=
 \sum_{\gamma\vdash D}
 \Phi_{\mathbf c}(\gamma)\frac{p_\gamma}{z_\gamma}
 \notag\\
 &=
 \Pi_D\left[
 \exp\!\left(\sum_{j\geq1}\frac{p_j}{j}\right)
 \prod_{j\geq1}
 \frac{(p_j-1)^{c_j}}{\sqrt{j^{c_j}c_j!}}
 \right].
 \label{eq:charlier-frobenius-series}
\end{align}
All Charlier inclusion--exclusion terms are contained in the
translations $p_j\mapsto p_j-1$ in this identity.

Write $p_j^\perp=j\,\partial/\partial p_j$ and introduce the Bernstein
vertex operator \cite[Chapter~I, Section~5]{kn:macdonald1995}
\begin{equation}\label{eq:bernstein-operator}
 \mathcal B(z)
 =
 \exp\!\left(\sum_{j\geq1}\frac{p_jz^j}{j}\right)
 \exp\!\left(-\sum_{j\geq1}\frac{p_j^\perp z^{-j}}{j}\right)
 =
 \sum_{n\in\mathbb Z}\mathcal B_nz^n.
\end{equation}
At $z=1$ the second exponential is
$\exp(-\sum_j\partial/\partial p_j)$, hence it performs precisely the
translation $p_j\mapsto p_j-1$. Since
\begin{equation}
 e_{\rho(\mathbf c)}
 =
 \prod_j\frac{p_j^{c_j}}{\sqrt{j^{c_j}c_j!}},
\end{equation}
\cref{eq:charlier-frobenius-series} becomes
\begin{equation}\label{eq:charlier-as-bernstein}
 F_{\mathbf c}^{(D)}
 =
 \Pi_D\mathcal B(1)e_{\rho(\mathbf c)}
 =
 \mathcal B_{D-w}e_{\rho(\mathbf c)}.
\end{equation}
The last equality is simply homogeneous-degree bookkeeping; in
particular, the formal specialisation $z=1$ causes no convergence issue.

We now sum the rank-one operators associated with the left side of
\cref{eq:charlier-spectral-resolution}. Put $T=D-N$. Completeness of
the normalised power sums in every $\Lambda_w$ and
\cref{eq:charlier-as-bernstein} give
\begin{align}
 \mathcal A_{N,D}
 &:=
 \sum_{\substack{\mathbf c\\\wt(\mathbf c)<T}}
 |F_{\mathbf c}^{(D)}\rangle
 \langle F_{\mathbf c}^{(D)}|
 \notag\\
 &=
 \sum_{w=0}^{T-1}
 \mathcal B_{D-w}\Pi_w\mathcal B_{D-w}^{\,*}.
 \label{eq:bernstein-rank-sum}
\end{align}
Here and below, empty sums are zero.

It remains to diagonalise this finite operator. The defining identity
for Bernstein modes \cite[Chapter~I, \S5, Example~29]{kn:macdonald1995} is
\begin{equation}
 \mathcal B_n s_\eta=s_{(n,\eta)},
\end{equation}
where the right side is a generalised Schur function. Exchanging adjacent
rows in the Jacobi--Trudi determinant
\cite[Chapter~I, Eq.~(3.4)]{kn:macdonald1995} gives
\begin{equation}
 s_{(\ldots,a,b,\ldots)}
 =
 -s_{(\ldots,b-1,a+1,\ldots)}.
\end{equation}
This straightening relation gives the following. If $\lambda\vdash D$
and $i\geq1$ is a row of $\lambda$, set
\begin{equation}\label{eq:bernstein-removal-data}
 n_i=\lambda_i-i+1,\qquad
 \eta^{(i)}
 =
 (\lambda_1+1,\ldots,\lambda_{i-1}+1,
  \lambda_{i+1},\lambda_{i+2},\ldots).
\end{equation}
Then $\eta^{(i)}\vdash D-n_i$ and
\begin{equation}\label{eq:bernstein-straightening}
 \mathcal B_{n_i}s_{\eta^{(i)}}
 =
 (-1)^{i-1}s_\lambda.
\end{equation}
Conversely, for $1\leq n\leq D$, straightening $(n,\eta)$ to $\lambda$
is possible only for the unique row $i$ with $n=n_i$; the integers
$n_i$ are strictly decreasing. Taking adjoints in
\cref{eq:bernstein-straightening}, and using orthonormality of Schur
functions, therefore gives
\begin{equation}\label{eq:bernstein-number-projection}
 \mathcal B_n\Pi_{D-n}\mathcal B_n^{\,*}s_\lambda
 =
 \mathbf1_{\{n=\lambda_i-i+1\text{ for some }i\}}s_\lambda,
 \qquad 1\leq n\leq D.
\end{equation}
In \cref{eq:bernstein-rank-sum}, $n=D-w$ runs through
$N+1,\ldots,D$. Hence
\begin{align}
 \mathcal A_{N,D}s_\lambda
 &=
 \#\{i:\lambda_i-i+1>N\}\,s_\lambda
 \notag\\
 &=
 \#\{i:\lambda_i\geq N+i\}\,s_\lambda
 =
 d_N(\lambda)s_\lambda.
 \label{eq:exceedance-diagonalization}
\end{align}
Finally, the $(e_\mu,e_\nu)$ matrix entry of
\cref{eq:bernstein-rank-sum} is
$K_{N,D}^{\Ch}(\mu,\nu)/\sqrt{z_\mu z_\nu}$, whereas
\cref{eq:frobenius-normalized,eq:exceedance-diagonalization} show that
it is
\begin{equation}
 \frac1{\sqrt{z_\mu z_\nu}}
 \sum_{\lambda\vdash D}
 d_N(\lambda)\chi^\lambda_\mu\chi^\lambda_\nu.
\end{equation}
Multiplication by $\sqrt{z_\mu z_\nu}$ proves
\cref{eq:charlier-spectral-resolution}.

We now apply this spectral resolution to the finite-$N$ moment.
The Frobenius characteristic formula and Schur orthogonality give the
standard character truncation:
\begin{equation}\label{eq:character-truncation}
 C_N(\mathbf{p},\mathbf{q})
 =
 \sum_{\substack{\lambda\vdash D\\\ell(\lambda)\leq N}}
 \chi^\lambda_\mu\chi^\lambda_\nu,
\end{equation}
where the dependence on $\mathbf{p}$ and $\mathbf{q}$ is carried through the cycle types $\mu$ and $\nu$ from \eqref{eq:z-wick-partitions-def}. The full character orthogonality relation is
\begin{equation}
 \sum_{\lambda\vdash D}
 \chi^\lambda_\mu\chi^\lambda_\nu
 =
 z_\mu\1_{\{\mu=\nu\}},
\end{equation}
so $\mathcal R_N$ is minus the contribution from $\ell(\lambda)>N$.

Conjugate every excluded partition. By
\begin{equation}
 \chi^{\lambda'}_\mu
 =
 (-1)^{D-k_1}\chi^\lambda_\mu,
\end{equation}
this contributes the common sign $(-1)^{k_1+k_2}$. The excluded
condition becomes $d_N(\lambda)\geq1$.

The spectral resolution also diagonalises the convolutional
binomials. Character orthogonality and
\cref{eq:charlier-spectral-resolution} imply, for $j,\ell\geq1$,
\begin{align}
 (K_{N,D}^{\Ch})^{\star_D j}(\mu,\nu)
 &=
 \sum_{\lambda\vdash D}
 d_N(\lambda)^j\chi^\lambda_\mu\chi^\lambda_\nu,
 \label{eq:kernel-power}\\
 \binom{K_{N,D}^{\Ch}}{\ell}_{\star_D}(\mu,\nu)
 &=
 \sum_{\lambda\vdash D}
 \binom{d_N(\lambda)}{\ell}
 \chi^\lambda_\mu\chi^\lambda_\nu.
 \label{eq:kernel-binomial}
\end{align}
Thus Ewens convolution is an ordinary matrix
multiplication after diagonalisation by the character table, and the
falling factorial in \cref{eq:convolution-binomial} becomes the ordinary
binomial coefficient in \cref{eq:kernel-binomial}. Finally,
\begin{equation}
 \sum_{\ell\geq1}(-1)^{\ell+1}\binom d\ell
 =
 \1_{\{d\geq1\}},
 \qquad d\in\mathbb Z_{\geq0}.
\end{equation}
Moreover, $d_N(\lambda)\geq\ell$ forces
$\lambda_\ell\geq N+\ell$ and hence
\begin{equation}
 D=|\lambda|\geq\ell(N+\ell).
\end{equation}
Therefore the sum stops at $s(N,D)$. Substituting
\cref{eq:kernel-binomial} and the last binomial identity into the
excluded character contribution in \cref{eq:character-truncation}, including its initial
minus sign, proves \cref{eq:charlier-correction}.
\end{proof}

\begin{proof}[Proof of \cref{thm:displacement}]
Introduce variables $t_1,\ldots,t_{k_1}$ and
$u_1,\ldots,u_{k_2}$, and set
\begin{equation}
 \omega(z)=
 \exp\!\left(
   \sum_{\alpha=1}^{k_1}t_\alpha z^{p_\alpha}
   +\sum_{\beta=1}^{k_2}u_\beta z^{-q_\beta}
 \right),
 \qquad
 \omega_k=[z^{-k}]\omega(z).
\end{equation}
The Heine--Szeg\H{o} identity
\cite[Theorem~1]{kn:bumdia02} gives
\begin{equation}
 \mathbb E_{\U(N)}
 \exp\!\left(
   \sum_{\alpha=1}^{k_1}t_\alpha\Tr(U^{p_\alpha})
   +\sum_{\beta=1}^{k_2}u_\beta
      \overline{\Tr(U^{q_\beta})}
 \right)
 =
 \det[\omega_{b-a}]_{a,b=0}^{N-1}.
\end{equation}
Write
\begin{equation}
 \mathcal D=
 \prod_{\alpha=1}^{k_1}\partial_{t_\alpha}
 \prod_{\beta=1}^{k_2}\partial_{u_\beta},
\end{equation}
and let evaluation at $0$ mean that all $t_\alpha$ and $u_\beta$
are set to zero. Since the traces are bounded at fixed $N$,
differentiation may be passed under the expectation. Thus
\begin{equation}
 C_N(\mathbf p,\mathbf q)
 =
 \left.\mathcal D\det[\omega_{b-a}]_{a,b=0}^{N-1}\right|_0.
\end{equation}
To evaluate this derivative, we expand the determinant:
\begin{equation}
 \det[\omega_{b-a}]_{a,b=0}^{N-1}
 =
 \sum_{\sigma\in S_N}\operatorname{sgn}(\sigma)
 \prod_{a=0}^{N-1}\omega_{\sigma(a)-a}.    
\end{equation}
Leibniz's rule assigns each labelled derivative to one of the $N$
factors. Denote these assignments by $r_\alpha=a$ when
$\partial_{t_\alpha}$ acts on the factor belonging to row $a$, and
$s_\beta=a$ when $\partial_{u_\beta}$ acts on that factor.
For a fixed assignment, put
\begin{equation}
 \mathcal D_a=
 \prod_{\substack{1\leq\alpha\leq k_1\\r_\alpha=a}}
       \partial_{t_\alpha}
 \prod_{\substack{1\leq\beta\leq k_2\\s_\beta=a}}
       \partial_{u_\beta},
\end{equation}
with an empty product interpreted as the identity operator.
Summing over assignments and recombining the permutation expansion gives
\begin{equation}
 C_N(\mathbf p,\mathbf q)
 =
 \sum_{\substack{\mathbf r\in J_N^{k_1}\\
                  \mathbf s\in J_N^{k_2}}}
 \det\!\left[
   \left.\mathcal D_a\omega_{b-a}\right|_0
 \right]_{a,b=0}^{N-1}.
\end{equation}
There are no additional factorials: every variable is differentiated
once, and every assignment occurs once in Leibniz's rule, even when
some trace exponents coincide.

For each row, differentiating the exponential symbol gives
\begin{equation}
 \left.\mathcal D_a\omega(z)\right|_0
 =
 \prod_{\alpha:r_\alpha=a}z^{p_\alpha}
 \prod_{\beta:s_\beta=a}z^{-q_\beta}
 =
 z^{h_a(\mathbf r,\mathbf s)}.
\end{equation}
Consequently,
\begin{equation}
 \left.\mathcal D_a\omega_{b-a}\right|_0
 =
 [z^{a-b}]z^{h_a(\mathbf r,\mathbf s)}
 =
 \1_{\{b=a-h_a(\mathbf r,\mathbf s)\}},
\end{equation}
which proves \cref{eq:displacement-determinant}. Each row of the resulting matrix is either zero or a
standard basis vector. Its determinant vanishes if a displaced index
lies outside $J_N$ or if two displaced indices coincide. Otherwise,
the matrix is a permutation matrix and its determinant is $1$ or $-1$.
\end{proof}

\begin{proof}[Proof of \cref{cor:displacement-product}]
Fix an assignment $(\mathbf r,\mathbf s)$ and write $x_a=a-h_a$.
If some $x_a\notin J_N$, the corresponding determinant has a zero
row and the cutoff product in \eqref{eq:displacement-product} vanishes.
If two displaced indices coincide, the determinant has repeated rows
and the Vandermonde product vanishes.

Otherwise, $x_a=\sigma(a)$ for a permutation $\sigma$ of $J_N$.
The determinant is then $\operatorname{sgn}(\sigma)$, while
\begin{equation}
 \prod_{0\leq a<b\leq N-1}(x_b-x_a)
 =
 \operatorname{sgn}(\sigma)
 \prod_{0\leq a<b\leq N-1}(b-a).
\end{equation}
Using \eqref{eq:displacement-normalisation} together with
$x_b-x_a=(b-a)-(h_b-h_a)$ gives, in every case,
\begin{equation}
 \det[\1_{\{b=a-h_a\}}]_{a,b=0}^{N-1}
 =
 \frac{1}{G(N+1)}
 \left(\prod_{a=0}^{N-1}\delta_N(a-h_a)\right)
 \prod_{0\leq a<b\leq N-1}
 \bigl[(b-a)-(h_b-h_a)\bigr].
\end{equation}
Here $G$ is the Barnes $G$-function, whose positive-integer values satisfy
\begin{equation}\label{eq:displacement-normalisation}
 G(N+1)
 =
 \prod_{j=1}^{N-1}j!
 =
 \prod_{0\leq a<b\leq N-1}(b-a).
\end{equation}
Summing over assignments and applying
\cref{eq:displacement-determinant} proves the result.
\end{proof}

\subsection{Bessel and Fredholm representations}
\label{sec:bessel-fredholm-representations}

Let
\begin{equation}
 M_{N,p}(k)=\mathbb E_{\U(N)}|\Tr(U^p)|^{2k},
\end{equation}
and define the entire Bessel determinant
\begin{equation}\label{eq:H-definition}
 \mathsf H_d(z)
 =
 \det[I_{a-b}(2\sqrt z)]_{a,b=0}^{d-1},
 \qquad \mathsf H_0(z)=1.
\end{equation}
Here, $I_n$ denotes the modified Bessel function of the first kind of order $n$. Although individual entries use $\sqrt z$, the determinant is entire. This Bessel determinant is the one-time Toeplitz specialisation appearing in
Gessel's symmetric-function identity \cite[Theorem~16]{kn:gessel90} and Forrester and Witte showed that it is a
Painlev\'e~$\mathrm{III}'$ $\tau$-function \cite{kn:forwit04}.

\begin{corollary}\label{cor:H-factorization}
With $N=pq+r$ as in \cref{eq:euclidiean-division},
\begin{equation}\label{eq:H-product}
 \sum_{k=0}^\infty M_{N,p}(k)\frac{z^k}{(k!)^2}
 =
 \mathsf H_{q+1}(z)^r\mathsf H_q(z)^{p-r}.
\end{equation}
Consequently, for
$X_{N,p}=p^{-1/2}\Tr(U_N^p)$ and $a=|t|^2/4$,
\begin{equation}\label{eq:characteristic-H-product}
 \mathbb E\e^{\ii t\cdot X_{N,p}}
 =
 \mathsf H_{q+1}(-a/p)^r\mathsf H_q(-a/p)^{p-r}.
\end{equation}
\end{corollary}

\begin{proof}
In \cref{eq:displacement-determinant}, take $k$ positive and $k$ negative
copies of the same exponent $p$.  If $A_i$ and $B_i$ are the numbers of
positive and negative labels assigned to row $i$, division by $(k!)^2$ and
summation in $k$ give
\begin{equation}
 \sum_{\mathbf A,\mathbf B\geq0}
 \prod_{i=0}^{N-1}\frac{z^{(A_i+B_i)/2}}{A_i!B_i!}
 \det[\1_{\{j=i-p(A_i-B_i)\}}]_{i,j}.
\end{equation}
In the above, each positive label contributes $p$ to the displacement and each
negative label contributes $-p$, so $h_i=p(A_i-B_i)$. The determinant already forces $\sum_iA_i=\sum_iB_i$.  By multilinearity,
the $(i,j)$ entry after summation is zero unless $i\equiv j\pmod p$, and
otherwise is
\begin{equation}
 \sum_{A-B=(i-j)/p}\frac{z^{(A+B)/2}}{A!B!}
 =
 I_{(i-j)/p}(2\sqrt z).
\end{equation}
Reorder the indices by residue class modulo $p$.  There are $r$ classes of
size $q+1$ and $p-r$ classes of size $q$, which proves
\cref{eq:H-product}. Expanding the exponential and using rotational invariance to eliminate
terms with unequal powers of the trace and its conjugate gives
\begin{equation}
 \mathbb E\e^{\ii t\cdot X_{N,p}}
 =
 \sum_{k=0}^\infty M_{N,p}(k)
 \frac{(-a/p)^k}{(k!)^2}.
\end{equation}
Thus substituting $z=-a/p$ into \cref{eq:H-product} proves
\cref{eq:characteristic-H-product}.
\end{proof}

\paragraph{The first altered single-power moment.}
Corollary~\ref{cor:H-factorization} makes the first defect completely
explicit. By character truncation, $\mathsf H_d(z)$ agrees with $\e^z$ through degree $d$;
at degree $d+1$ it excludes only the partition $(1^{d+1})$, whose character
at the identity is one.  Consequently the first logarithmic correction is
\begin{equation}\label{eq:first-H-defect}
 \log\mathsf H_d(z)=z-\frac{z^{d+1}}{((d+1)!)^2}
 +O(z^{d+2}).
\end{equation}
Substituting this expansion into \eqref{eq:H-product}, with $\nu=q+1$, gives
\begin{equation}\label{eq:first-moment-defect-explicit}
 \mathbb E|\Tr(U_N^p)|^{2k}=k!p^k\quad(0\leq k\leq q),
 \qquad
 \mathbb E|\Tr(U_N^p)|^{2\nu}=\nu!p^\nu-(p-r).
\end{equation}
The $p-r$ blocks of size $q$ each contribute the coefficient
$-1/(\nu!)^2$ in \eqref{eq:first-H-defect}; the $r$ larger blocks have no
defect at that degree.  Multiplication by $(\nu!)^2$ converts the generating
coefficient back into the moment.  Thus the first deviation is an integer
count of the smaller residue classes.  For $p=1$ this reduces to the familiar
identity
$\mathbb E|\Tr(U_N)|^{2(N+1)}=(N+1)!-1$.

The scalar Bessel product is the one-time specialisation of a multitime
factorisation; we state explicitly the extension because it links the displacement
determinant directly to the multitime generating function below.

\begin{corollary}[Multitime factorisation and Fredholm correction]
\label{cor:multitime-factorization}
Let $g\geq1$ be an integer, write $N=gQ+R$ with $0\leq R<g$, and suppose that
$\mathbf t=(t_k)_{k\geq1}$ and $\mathbf s=(s_k)_{k\geq1}$ satisfy
\begin{equation}\label{eq:multitime-summability}
 \sum_{k\geq1}k(|t_k|+|s_k|)<\infty.
\end{equation}
Define
\begin{equation}\label{eq:multitime-H}
 \mathcal H_d(\mathbf t,\mathbf s)=
 \mathbb E_{V\in\U(d)}\exp\!\left\{
 \sum_{k\geq1}t_k\Tr(V^k)+\sum_{k\geq1}s_k\Tr(V^{-k})\right\},
 \qquad \mathcal H_0=1.
\end{equation}
Then
\begin{align}
 \mathcal Z_{N,g}(\mathbf t,\mathbf s)
 &:={\mathbb E}_{U\in\U(N)}\exp\!\left\{
 \sum_{k\geq1}t_k\Tr(U^{gk})+
 \sum_{k\geq1}s_k\Tr(U^{-gk})\right\} \notag\\
 &=\mathcal H_{Q+1}(\mathbf t,\mathbf s)^R
   \mathcal H_Q(\mathbf t,\mathbf s)^{g-R}.      \label{eq:multitime-factorization}
\end{align}
Put
\begin{align}
 \Phi(z)&=\exp\!\left(\sum_{k\geq1}t_kz^k+
                      \sum_{k\geq1}s_kz^{-k}\right),
 &\mathcal I_n&=[z^{-n}]\Phi(z),                 \label{eq:multitime-symbol}\\
 \mathfrak s(\mathbf t,\mathbf s)&=\sum_{k\geq1}kt_ks_k,
 &\mathcal F_d(\mathbf t,\mathbf s)&=
   \e^{-\mathfrak s(\mathbf t,\mathbf s)}
   \mathcal H_d(\mathbf t,\mathbf s).           \label{eq:normalized-H}
\end{align}
The one-block factor and its normalised correction factor are
\begin{equation}\label{eq:multitime-toeplitz-fredholm}
 \mathcal H_d=\det[\mathcal I_{b-a}]_{a,b=0}^{d-1},
 \qquad \mathcal F_d=\det(I-K_d).
\end{equation}
Here, with Laurent coefficients denoted by $[z^n]$, set
\begin{equation}\label{eq:BO-kernel-definition}
 b(z)=\exp\!\left(\sum_{k\geq1}(s_kz^{-k}-t_kz^k) \right),
 \qquad c(z)=b(z)^{-1},
\end{equation}
let $\operatorname{Hank}(a)=([z^{i+j+1}]a(z))_{i,j\geq0}$ be the Hankel operator,
write $\tilde c(z)=c(z^{-1})$, and let $Q_d$ project
$\ell^2(\mathbb Z_{\geq0})$ onto $\ell^2(\{d,d+1,\ldots\})$.  Then
\begin{equation}\label{eq:BO-block-kernel}
 K_d=Q_d\operatorname{Hank}(b)\operatorname{Hank}(\tilde c)Q_d.
\end{equation}
In particular, on the direct sum of $g$ tail spaces,
\begin{align}
 \e^{-g\mathfrak s}\mathcal Z_{N,g}
 &=\det(I-\mathbf K_{N,g}),                       \label{eq:global-fredholm-factor}\\
 \mathbf K_{N,g}
 &=K_{Q+1}^{\oplus R}\oplus K_Q^{\oplus(g-R)},  \label{eq:global-block-kernel}\\
 \det(I-\mathbf K_{N,g})
 &=1-\Tr\mathbf K_{N,g}+
   \sum_{\ell\geq2}(-1)^\ell
   \Tr(\wedge^\ell\mathbf K_{N,g}),             \label{eq:global-exterior-expansion}\\
 \Tr\mathbf K_{N,g}
 &=R\Tr K_{Q+1}+(g-R)\Tr K_Q.                   \label{eq:global-first-trace}
\end{align}
All identities also hold coefficientwise in the completed formal power-series
ring with $\deg t_k=\deg s_k=k$.
\end{corollary}

\begin{proof}
The symbol for $\mathcal Z_{N,g}$ is $\Phi(z^g)$.  Its Fourier coefficients
vanish outside $g\mathbb Z$, so reordering the Toeplitz indices by residue
class modulo $g$ produces $R$ blocks of size $Q+1$ and $g-R$ blocks of size
$Q$.  This proves \eqref{eq:multitime-factorization}.  Heine--Szeg\H{o} gives
the Toeplitz determinant in \eqref{eq:multitime-toeplitz-fredholm}.

The Borodin--Okounkov--Geronimo--Case identity, in the product-of-Hankel
normalisation of \cite{kn:borodinokounkov00,kn:basorwidom00}, gives
\begin{equation}
 \det[\mathcal I_{b-a}]_{a,b=0}^{d-1}
 =\exp\!\left(\sum_{k\geq1}kt_ks_k\right)
  \det\!\left(I-Q_d\operatorname{Hank}(b)
                   \operatorname{Hank}(\tilde c)Q_d\right).
\end{equation}
Condition \eqref{eq:multitime-summability} makes the Hankel product trace
class and justifies all rearrangements.  Taking the product over residue
classes gives \eqref{eq:global-fredholm-factor} and
\eqref{eq:global-block-kernel}. The Fredholm exterior-power expansion gives
\eqref{eq:global-exterior-expansion}, and additivity of the trace on
direct sums gives \eqref{eq:global-first-trace}.

To obtain the formal identities, fix $D$ and set $t_k=s_k=0$ for
$k>D$. This truncation leaves all coefficients of weighted degree
at most $D$ unchanged. At these degrees, only finitely many terms
contribute to the Fredholm expansions, so the analytic identities
for the truncated times give equality of the corresponding Taylor
coefficients. Since $D$ is arbitrary, the formal identities follow.
\end{proof}

The factor $\e^{g\mathfrak s}$ is the stable contribution, while
$\mathbf K_{N,g}$ encodes the finite-$N$ correction. Its first trace
retains the two possible block dimensions $Q,Q+1$. The asymptotic
proofs use this trace to identify the leading correction and the
full determinant to control the remainder.

\subsection{Bessel Fredholm estimates}
\label{sec:bessel-estimate}

\ifdefined\PedagogicalVersion\PedBesselFredholm\fi

The fixed-power and joint-growth arguments use the same special
linearisation of the Borodin--Okounkov determinant.  We state it once in the
normalisation of \cref{eq:bessel-fredholm-kernel}. The
signed trace itself controls the quadratic error below, which does not follow
from a generic trace-class estimate.  On $\ell^2(\N)$ let
\begin{equation}\label{eq:bessel-fredholm-kernel}
 \mathcal J(\rho)=\bigl[I_{i+j-1}(\rho)\bigr]_{i,j\geq1},
 \qquad
 \mathcal L=\operatorname{diag}((-1)^i)_{i\geq1},
 \qquad
 \mathcal K(\rho)=-\mathcal J(\rho)\mathcal L
                    \mathcal J(\rho)\mathcal L,
\end{equation}
and let $Q_d$ project onto the coordinates $d+1,d+2,\ldots$.

For $d\geq1$ and $\rho\geq0$, put
\begin{equation}\label{eq:common-TD-definition}
 T_d(\rho)=\Tr\!\left(Q_d\mathcal K(\rho)Q_d\right),
 \qquad
 D_d(\rho)=\det\!\left(I-Q_d\mathcal K(\rho)Q_d\right).
\end{equation}

\begin{lemma}[Bessel determinant linearisation]
\label{lem:bessel-determinant-linearization}
Let $c_0$ be the positive solution of
$c_0\e^{1+c_0^2}=1$, and put
\begin{equation}\label{eq:CJL-trace-norm-constant}
 C_{\mathrm{tr}}=\frac1{2\pi(1-c_0^2)^2}
 =0.2004129914\ldots<0.201.
\end{equation}
There is an absolute constant $C_1<\infty$ such that, for every $d\geq1$
and $0\leq\rho\leq2c_0(d+1)$,
\begin{align}
 |T_d(\rho)|
 &\leq
 \frac{2\exp\!\left(\rho^2/[2(d+1)]\right)}
 {\left(1-(\rho/[2(d+1)])^4\right)^2}
 \frac{(\rho/2)^{2d+2}}{\Gamma(d+2)^2},
                                                        \label{eq:common-trace-factorial}\\
 |T_d(\rho)|&\leq\frac{C_{\mathrm{tr}}}{d+1},
                                                        \label{eq:common-trace-small}\\
 \sum_{j\geq2}\frac1j
 \left|\Tr\!\left((Q_d\mathcal K(\rho)Q_d)^j\right)\right|
 &\leq C_1|T_d(\rho)|^2.                       \label{eq:common-higher-traces}
\end{align}
Consequently, with
\begin{equation}\label{eq:common-log-remainder}
 \mathfrak r_d(\rho)=\sum_{j\geq2}\frac1j
 \Tr\!\left((Q_d\mathcal K(\rho)Q_d)^j\right),
\end{equation}
one has
\begin{align}
 D_d(\rho)&=\exp\{-T_d(\rho)-\mathfrak r_d(\rho)\},
 &|\mathfrak r_d(\rho)|&\leq C_1|T_d(\rho)|^2, \label{eq:common-logdet}\\
 D_d(\rho)&=1-T_d(\rho)+E_d(\rho),
 &|E_d(\rho)|&\leq C_1|T_d(\rho)|^2.            \label{eq:common-linearization}
\end{align}
The constants are uniform for the two dimensions $d=q,q+1$.
\end{lemma}

\begin{proof}
In the notation of Courteaut--Johansson--Lambert~\cite{kn:courteautjohanssonlambert24},
their $n$ is our $d$, their $r=|\zeta|$ is our $\rho$, and their $Q_n$ has
kernel $\operatorname{span}(e_1,\ldots,e_n)$, so it is exactly our tail
projection $Q_d$.  Their operator has the opposite overall sign from
\cref{eq:bessel-fredholm-kernel}; this changes odd traces by a sign but
does not affect any estimate in absolute value. With this translation, \cref{eq:common-trace-factorial} is their
Lemma~2.6.  Their Lemma~2.5 gives the sharper trace-norm estimate
\begin{equation}\label{eq:common-trace-norm}
 \norm{Q_d\mathcal K(\rho)Q_d}_1
 \leq\frac{C_{\mathrm{tr}}}{d+1}
 \leq\frac{0.201}{2}<0.101<1
\end{equation}
throughout the stated range, and hence also
\cref{eq:common-trace-small}.  Their Lemma~2.8 proves in the same range
the special signed-trace estimate \cref{eq:common-higher-traces}.

Next, we derive the logarithmic and linear approximations
to $D_d$. Since $\norm{Q_d\mathcal KQ_d}_1<1$, the trace series for
the logarithm of the Fredholm determinant converges absolutely
\cite[Eq.~(3.3)]{kn:bornemann10}. Separating the first term and
exponentiating gives
\begin{equation}
 D_d
 =\exp\!\left\{-T_d-
   \sum_{j\geq2}\frac1j\Tr((Q_d\mathcal KQ_d)^j)\right\}
 =\e^{-T_d-\mathfrak r_d},
\end{equation}
which proves \cref{eq:common-logdet}.

Set $z=T_d+\mathfrak r_d$. By
\cref{eq:common-higher-traces,eq:common-trace-small},
\begin{equation}
 |\mathfrak r_d|\leq C_1|T_d|^2,
 \qquad
 |z|\leq |T_d|+C_1|T_d|^2\leq C|T_d|.
\end{equation}
In particular, $\e^{|z|}$ is uniformly bounded. Since
$z-T_d=\mathfrak r_d$, we have the exact decomposition
\begin{equation}
 D_d-1+T_d=(\e^{-z}-1+z)-\mathfrak r_d.
\end{equation}
Applying
$|\e^{-z}-1+z|\leq\frac12|z|^2\e^{|z|}$ therefore yields
\begin{equation}
 |D_d-1+T_d|
 \leq \frac12|z|^2\e^{|z|}+|\mathfrak r_d|
 \leq C|T_d|^2.
\end{equation}
Enlarging $C_1$ proves \cref{eq:common-linearization}.
This quadratic estimate also appears in the proof of
\cite[Corollary~2.9]{kn:courteautjohanssonlambert24}.
\end{proof}

\subsection{Fourier--Laguerre transfer}
\label{sec:fourier-laguerre-transfer}

For $f\in L^1(\R^2)$, let $\F f$ denote its Fourier transform,
with the convention
\begin{equation}\label{eq:fourier-convention}
 (\F f)(t)=\int_{\R^2}\e^{\ii t\cdot x}f(x)\dd x,
 \qquad
 (\F^{-1}h)(x)=\frac1{(2\pi)^2}
 \int_{\R^2}\e^{-\ii t\cdot x}h(t)\dd t.
\end{equation}
Then $(\F g)(t)=\e^{-|t|^2/4}$. We also use the Plancherel extensions of these operators to $L^2(\R^2)$. The ordinary Laguerre polynomials are
\begin{equation}
 L_k(u)=\sum_{j=0}^k(-1)^j\binom{k}{j}\frac{u^j}{j!}.
\end{equation}

\begin{lemma}[Laguerre modes]\label{lem:laguerre-fourier}
For every $k\geq0$,
\begin{equation}\label{eq:general-laguerre-fourier}
 \F\!\left(L_k(|\cdot|^2)g\right)(t)
 =
 \e^{-|t|^2/4}\frac{(|t|^2/4)^k}{k!}.
\end{equation}
 Moreover, if $R\in L^2(\R^2)$ and its distributional Laplacian
$\Delta R$ belongs to $L^2(\R^2)$, and if $f=\F^{-1}R$, then
\begin{equation}\label{eq:weighted-fourier-transfer}
 \norm{f}_1
 \leq C\bigl(\norm{R}_2+\norm{\Delta R}_2\bigr).
\end{equation}
\end{lemma}

\begin{proof}
The identity
$\Delta^kg=(-4)^kk!L_k(|x|^2)g$ follows by induction from the Laguerre
differential equation.  Since
$\F(\Delta^kg)=(-1)^k|t|^{2k}\e^{-|t|^2/4}$, with the signs cancelled by
the chosen Laguerre convention, one obtains
\cref{eq:general-laguerre-fourier}.  For the second assertion,
\begin{equation}
 \norm{f}_1
 \leq\norm{(1+|x|^2)^{-1}}_2
      \norm{(1+|x|^2)f}_2
 \leq C\norm{R-\Delta R}_2,
\end{equation}
by Cauchy--Schwarz and Plancherel.
\end{proof}

We shall repeatedly use the following radial identity. If
$a=|t|^2/4$ and $F=F(a)$ has two locally absolutely continuous
derivatives on $[0,\infty)$, then, distributionally (and pointwise
wherever the derivatives are classical),
\begin{equation}\label{eq:radial-laplacian-a}
 \Delta_tF(a)=F'(a)+aF''(a).
\end{equation}
For a rotationally invariant random variable $X$ with a finite moment $\mathbb E|X|^{2k}$, define its $k$th Laguerre coefficient
\begin{equation}\label{eq:beta-definition}
 \beta_X(k)=\mathbb E L_k(|X|^2)
 =
 \sum_{j=0}^k(-1)^j\binom{k}{j}
 \frac{\mathbb E|X|^{2j}}{j!}.
\end{equation}
Whenever the resulting Bessel/moment series is absolutely convergent,
expanding the Bessel function $J_0$ of the first kind of order zero and using binomial inversion gives the characteristic-function
identity
\begin{equation}\label{eq:laguerre-characteristic-series}
 \mathbb E\e^{\ii t\cdot X}
 =
 \e^{-a}\sum_{k=0}^\infty\beta_X(k)\frac{a^k}{k!},
 \qquad a=\frac{|t|^2}{4}.
\end{equation}
All trace variables used below are bounded at fixed matrix dimension, so
termwise expectation here follows directly from dominated convergence.
If the characteristic function is integrable, Fourier inversion first gives
the existence of a density.  Whenever, in addition, the corresponding
Laguerre series converges in $L^1$, that density is
\begin{equation}\label{eq:density-laguerre-series}
 f_X(x)=g(x)\sum_{k=0}^\infty\beta_X(k)L_k(|x|^2).
\end{equation}

\section{Fixed powers: sharp Gaussian approximation}
\label{sec:fixed-power}

\ifdefined\PedagogicalVersion\PedFixedPower\fi

We now prove \cref{thm:fixed-power}, retaining the notation
\eqref{eq:quotient-remainder}; in particular, $q\to\infty$.
The proof has four stages. The Charlier--Ewens formula identifies the
first correction band and its onset. The one-block Bessel determinant $\mathsf H_d$ defined in
\eqref{eq:H-definition} computes the balanced cumulants with
indices $d+1,\ldots,2d+1$ for a block of dimension $d$. We call these ``near-edge cumulants'' as
they lie beyond the Gaussian moment-matching range $k\leq d$;
they determine the coefficients of the leading Laguerre correction.
The normalised Fredholm determinant controls the remainder after the
first trace is removed, including the unsmoothed case $p=1$.
Finally, Fourier--Laguerre transfer and spatial truncation give the
density expansion \eqref{eq:leading-density-template}.

Here ``first band'' means the summand indexed by $\ell=1$ in
\eqref{eq:charlier-correction}. For fixed $p$, we retain its $q$-block
Laguerre modes with indices $\nu,\nu+1,\ldots,2\nu-1$ where $\nu = q+1$.
By \eqref{eq:first-H-defect}, the non-Gaussian part of the logarithm
begins at degree $\nu$, so exponentiating introduces no products of
two such terms below degree $2\nu$. The Laguerre coefficients in this
range therefore depend linearly on the near-edge cumulants.
The estimates below show that the higher $q$-block modes and the
$q+1$-block contribution are negligible on the leading
total-variation scale. Although degree $\nu$ gives the first non-Gaussian coefficient,
the higher coefficients in this range also affect the leading constant.
Every fixed shift $\nu+j$ is eventually retained as $q\to\infty$,
and its Laguerre mode contributes at the same asymptotic order as
$L_\nu$. The coefficient ratios converge to $(2/p)^j/j!$,
whose coherent summation produces the factor $\e^{2/p}$ in
\eqref{eq:band-summation}. Thus the first defective moment alone
does not determine the leading total-variation constant.

\subsection{Near-edge cumulants and the first correction band}

For $d\geq1$, define the balanced cumulants $\kappa_k^{(d)}$
by the Taylor expansion at $z=0$:
\begin{equation}\label{eq:balanced-cumulants}
 \log\mathsf H_d(z)
 =
 z+\sum_{k\geq2}\kappa_k^{(d)}\frac{z^k}{(k!)^2}.
\end{equation}

\begin{lemma}[Near-edge cumulants]\label{lem:near-edge-cumulants}
Let $d\geq1$. For $1\leq s\leq d+1$,
\begin{equation}\label{eq:near-edge-cumulant-formula}
 \kappa_{d+s}^{(d)}
 =
 \frac{(-1)^s}{(s-1)!}(2d+s)_{s-1},
\end{equation}
where $(x)_j=x(x+1)\cdots(x+j-1)$.  All
$\kappa_k^{(d)}$ with $2\leq k\leq d$ vanish.
\end{lemma}

\begin{proof}
Apply \cref{thm:charlier-correction} to $\mu=\nu=(1^k)$.  For $k\leq d$
the correction sum is empty, giving the Gaussian moments. If
$k=d+s\leq2d+1$, the logarithm of the moment series cannot contain a
product of two non-Gaussian terms, because every such product has degree at
least $2d+2$. Hence only the $\ell = 1$ correction band is needed in this degree range. Combining
\cref{eq:trace-bessel-square,eq:bessel-square-series} below, or equivalently
performing the one-coordinate Charlier extraction in
\cref{eq:charlier-kernel}, gives
\begin{equation}\label{eq:cumulant-finite-sum}
 \kappa_{d+s}^{(d)}
 =
 \sum_{h=1}^{s}(-1)^{s+h+1}h\,
 \frac{(2d+s+h+1)_{s-h}}{(s-h)!}.
\end{equation}
Put $u=s-h$ and $n=2d+2s$; then
\begin{equation}
 \frac{(2d+2s-u+1)_u}{u!}=\binom nu
\end{equation}
and therefore
\begin{equation}
 \kappa_{d+s}^{(d)}
 =\sum_{u=0}^{s-1}(-1)^{u+1}(s-u)\binom nu.
\end{equation}
The two partial-binomial identities
\begin{align*}
 \sum_{u=0}^{s-1}(-1)^u\binom nu
 =(-1)^{s-1}\binom{n-1}{s-1}, \qquad 
 \sum_{u=0}^{s-1}(-1)^uu\binom nu
 =(-1)^{s-1}n\binom{n-2}{s-2}
\end{align*}
(the second follows from $u\binom nu=n\binom{n-1}{u-1}$)
yield
\begin{equation}
 \kappa_{d+s}^{(d)}
 =(-1)^s\binom{2d+2s-2}{s-1}
 =\frac{(-1)^s}{(s-1)!}(2d+s)_{s-1},
\end{equation}
where $(x)_0=1$ and
$(x)_j=x(x+1)\cdots(x+j-1)$ for $j\geq1$. Above, we use the convention $\binom{a}{-1}=0$.
\end{proof}
For $X_{N,p}=p^{-1/2}\Tr(U_N^p)$ write
\begin{equation}
 \beta_{N,p}(k)=\mathbb E L_k(|X_{N,p}|^2),
 \qquad
 A_{N,p}=\frac{p-r}{p^\nu\Gamma(\nu+1)}.
\end{equation}
For $0\leq j\leq q$, put
\begin{equation}\label{eq:cqj}
 c_{q,j}
 =
 \frac{(q+1)!}{(q+j+1)!}\,
 \frac{(2q+j+1)_j}{j!}.
\end{equation}
Then $c_{q,0}=1$, and for every fixed $j$,
\begin{equation}\label{eq:cqj-limit}
 c_{q,j}\longrightarrow\frac{2^j}{j!}.
\end{equation}
Uniformly for $0\leq j\leq q$,
$0\leq c_{q,j}\leq4^j/j!$.

We now use the near-edge cumulants to construct a candidate for the
leading density correction. Retaining the Laguerre modes with indices
$\nu,\ldots,2\nu-1$ contributed by the blocks of dimension $q$, define
the signed comparison function
\begin{equation}\label{eq:one-residue-density}
 \Lambda_{N,p}(x)
 =
 (-1)^qA_{N,p}g(x)
 \sum_{j=0}^q c_{q,j}p^{-j}L_{\nu+j}(|x|^2).
\end{equation}
The following lemma identifies these coefficients and bounds the
additional contribution from the blocks of dimension $q+1$.
The Fredholm estimates in the next subsection will control the
remaining terms in the density approximation.

\begin{lemma}[Near-edge Laguerre coefficients]
\label{lem:edge-laguerre-coefficients}
For $0\leq j\leq q$,
\begin{equation}\label{eq:edge-beta}
 \beta_{N,p}(\nu+j)
 =
 (-1)^q A_{N,p}\,p^{-j}c_{q,j}
 +\varepsilon_{q,j},
\end{equation}
where
\begin{equation}\label{eq:large-block-error}
 \sum_{j=0}^q|\varepsilon_{q,j}|
 \leq \frac{C_pA_{N,p}}{\nu}.
\end{equation}
More precisely, with
\begin{equation}\label{eq:adjacent-block-coefficients}
 A_+=\frac{r}{p^{\nu+1}(\nu+1)!},\qquad
 c^+_{q,j}
 =\frac{(q+2)!}{(q+j+2)!}\,
   \frac{(2q+j+3)_j}{j!},
\end{equation}
one has $\varepsilon_{q,0}=0$ and
\begin{equation}\label{eq:adjacent-block-error-exact}
 \varepsilon_{q,j}
 =(-1)^{q+1}A_+p^{-(j-1)}c^+_{q,j-1},
 \qquad 1\leq j\leq q.
\end{equation}
\end{lemma}

\begin{proof}
By \cref{cor:H-factorization},
\begin{align}
 \log\!\left[
   \e^{-z}\mathsf H_{q+1}(z/p)^r
   \mathsf H_q(z/p)^{p-r}
 \right]
 &=
 r\sum_{k\geq q+2}\kappa_k^{(q+1)}
      \frac{(z/p)^k}{(k!)^2}\notag\\
 &\quad +(p-r)\sum_{k\geq q+1}\kappa_k^{(q)}
      \frac{(z/p)^k}{(k!)^2}.
 \label{eq:log-fixed-factor}
\end{align}
Below degree $2\nu$, exponentiating does not produce a product of two
non-Gaussian coefficients. Substitute
\cref{eq:near-edge-cumulant-formula} in the second sum, replace $z$ by
$-a$, and compare with \cref{eq:laguerre-characteristic-series}.
This gives the first term in \cref{eq:edge-beta}; the factorial
quotient is exactly \cref{eq:cqj}.

The blocks of dimension $q+1$ first contribute a non-Gaussian
coefficient at degree $\nu+1$. Applying
\cref{eq:near-edge-cumulant-formula} with $d=q+1$ shows that their
contribution at degree $\nu+1+j$, $0\leq j\leq q-1$, is
\begin{equation}
 (-1)^{q+1}A_+p^{-j}c^+_{q,j}.
\end{equation}
This proves \cref{eq:adjacent-block-error-exact}.
Finally, for $0\leq j\leq q-1$, pairing the $j$ numerator and
denominator factors in \cref{eq:adjacent-block-coefficients} gives
\begin{equation}\label{eq:adjacent-block-ratio-bound}
 0\leq c^+_{q,j}\leq\frac{4^j}{j!}.
\end{equation}
Since
\begin{equation}
\label{eq:second-reference}
 \frac{A_+}{A_{N,p}}
 =\frac{r}{(p-r)p(\nu+1)}
 \leq\frac{C_p}{\nu},
\end{equation}
we obtain
\begin{equation}
 \sum_{j=0}^q|\varepsilon_{q,j}|
 \leq A_+\sum_{j\geq0}\frac{(4/p)^j}{j!}
 \leq\frac{C_pA_{N,p}}{\nu},
\end{equation}
which is \cref{eq:large-block-error}.
\end{proof}

\subsection{Fredholm remainder estimates}

\ifdefined\PedagogicalVersion\PedGlobalRemainder\fi

The Charlier formula identifies the near-edge moment corrections. We obtain absolute control of the entire remainder from
the determinant, not by taking absolute values inside the Ewens
convolution.  Specialising the normalised Fredholm identity in
\cref{cor:multitime-factorization} to one non-zero pair of times gives
\begin{equation}\label{eq:BO-H}
 \mathsf H_d(-\rho^2/4)
 =\e^{-\rho^2/4}
   \det\!\bigl(I-Q_d\mathcal K(\rho)Q_d\bigr).
\end{equation}
In particular,
\begin{equation}\label{eq:first-fredholm-trace}
 \Tr(Q_d\mathcal K(\rho)Q_d)
 =\sum_{i>d}\sum_{j\geq1}(-1)^{i+j+1}I_{i+j-1}(\rho)^2.
\end{equation}

\begin{proposition}[Fixed-power Fredholm approximation]
\label{prop:fixed-power-approximation}
Fix $p\geq1$.  For all sufficiently large $N$, $X_{N,p}$ has a density
$f_{N,p}$, and
\begin{equation}\label{eq:determinant-L2-remainder}
 \norm{f_{N,p}-g-\Lambda_{N,p}}_2
 \leq C_p A_{N,p}\nu^{-5/4}.
\end{equation}
Moreover,
\begin{equation}\label{eq:fixed-L1-remainder}
 \norm{f_{N,p}-g-\Lambda_{N,p}}_1
 =o\!\left(A_{N,p}\nu^{-1/4}\right).
\end{equation}
Both estimates include the unsmoothed case $p=1$.
\end{proposition}

\begin{proof}
Put $a=|t|^2/4$ and $\rho=|t|/\sqrt p$.  We first define the characteristic
function, without assuming a density:
\begin{equation}\label{eq:fixed-fredholm-product}
 \varphi_{N,p}(t)
 :=\mathbb E\exp\!\left(\ii t\cdot
       \frac{\Tr(U_N^p)}{\sqrt p}\right)
 =\e^{-a}D_{q+1}(\rho)^rD_q(\rho)^{p-r}.
\end{equation}
This follows from
\cref{eq:characteristic-H-product,eq:BO-H,cor:multitime-factorization}; in
particular, we derived its only block reduction from the displacement
determinant.  Put $\lambda_{N,p}=\F\Lambda_{N,p}$ and
\begin{equation}\label{eq:frequency-remainder-fixed}
 \mathcal E_{N,p}(t)
 =\varphi_{N,p}(t)-\e^{-a}-\lambda_{N,p}(t).
\end{equation}

We begin on $0\leq\rho\leq2c_0\nu$.  By
\cref{lem:bessel-determinant-linearization}, with
\begin{equation}
 L(\rho)=(p-r)T_q(\rho)+rT_{q+1}(\rho),
\end{equation}
the finite product satisfies
\begin{equation}\label{eq:fixed-product-linearization}
 D_q(\rho)^{p-r}D_{q+1}(\rho)^r
 =1-L(\rho)+\mathcal R_p^{\mathrm{det}}(\rho),
 \qquad
 |\mathcal R_p^{\mathrm{det}}|
 \leq C_p\bigl(|T_q|^2+|T_{q+1}|^2\bigr).
\end{equation}
Insert \cref{eq:common-linearization}, expand the fixed number $p$
of factors, and use \cref{eq:common-trace-small}; every term not displayed
contains either one $E_d$ or at least two traces.

We next identify and sum the omitted part of $L$.  Grouping
\cref{eq:first-fredholm-trace} by the Bessel index gives
\begin{equation}\label{eq:trace-bessel-square}
 T_d(\rho)=\sum_{h=1}^\infty h(-1)^{d+h}I_{d+h}(\rho)^2.
\end{equation}
For $x=\rho^2/4$, Vandermonde's identity gives
\begin{equation}\label{eq:bessel-square-series}
 I_n(2\sqrt x)^2
 =\sum_{j=0}^\infty
 \frac{(2n+j+1)_j}{j!(n+j)!^2}x^{n+j}.           
\end{equation}
The coefficient of $x^{n+j}$ equals
\begin{equation}
 \sum_{v=0}^j
 \frac1{v!(j-v)!(n+v)!(n+j-v)!}
 =\frac{(2n+j+1)_j}{j!(n+j)!^2}.
\end{equation}
After the scaling $x=a/p$, the absolute coefficient of
$a^{d+h+j}$ in the $h$th term is
\begin{equation}\label{eq:scaled-block-coefficient}
 B_{d,h,j}(p)
 =\frac{h}{p^{d+h+j}}
   \frac{(2(d+h)+j+1)_j}{j!(d+h+j)!^2}.
\end{equation}
Then, after multiplication by $\e^{-a}$ and Fourier inversion,
the corresponding absolute Laguerre coefficient from the
blocks of dimension $q$ is
\begin{equation}\label{eq:fixed-density-bhj}
 b_{h,j}
 =(p-r)\frac{h}{p^{q+h+j}}
   \frac{(2(q+h)+j+1)_j}{j!(q+h+j)!}.
\end{equation}
Thus $b_{1,0}=A_{N,p}$ and, after division, we obtain
\begin{align}
 \frac{b_{h,j+1}}{b_{h,j}}
 &\leq\frac4{p(j+1)},                             \label{eq:block-j-ratio}\\
 \frac{b_{h,0}}{A_{N,p}}
 &=\frac{h}{p^{h-1}}\frac{\nu!}{(q+h)!}
 \leq\frac{h}{(p\nu)^{h-1}}.                    \label{eq:block-h-ratio}
\end{align}
Collecting equal powers in \cref{eq:trace-bessel-square} shows that the
part of $-(p-r)T_q$ of degrees $\nu,\ldots,2\nu-1$ is
\begin{equation}\label{eq:first-trace-polynomial}
 (-1)^qA_{N,p}
 \sum_{j=0}^q c_{q,j}p^{-j}\frac{a^{\nu+j}}{(\nu+j)!}.
\end{equation}
After multiplication by $\e^{-a}$, this polynomial is the Fourier
transform of the contribution from the blocks of dimension $q$
identified in \cref{lem:edge-laguerre-coefficients}; hence it gives
exactly $\lambda_{N,p}$. The original terms omitted from the $q$ trace satisfy
$(h-1)+j\geq\nu$.  Writing $u=h-1$ and iterating
\cref{eq:block-j-ratio,eq:block-h-ratio}, we obtain
\begin{align}
 \sum_{u+j\geq\nu}b_{u+1,j}
 &\leq A_{N,p}
 \sum_{u+j\geq\nu}
 \frac{u+1}{(p\nu)^u}\frac{(4/p)^j}{j!}
 \leq\frac{C_pA_{N,p}}{\nu}.                    \label{eq:fixed-omitted-mass}
\end{align}
For the last inequality split the sum into $u\geq\nu/2$ and
$j\geq\nu/2$: the first part is a geometric tail and the second an
exponential-series tail. For blocks of dimension $q+1$, the combined contribution
starts with a coefficient of magnitude $A_+$, as defined in
\eqref{eq:adjacent-block-coefficients}. Applying the ratio estimates
\cref{eq:block-j-ratio,eq:block-h-ratio} with $q$ replaced by $q+1$
bounds their complete coefficient mass by $C_pA_+$.
By \eqref{eq:second-reference}, this is $O_p(A_{N,p}/\nu)$.

For every $k\geq1$,
\begin{equation}\label{eq:gamma-L2-identity}
 \int_{\R^2}\e^{-2a}\left(\frac{a^k}{k!}\right)^2\dd t
 =2\pi\,4^{-k}\binom{2k}{k}\asymp k^{-1/2}.
\end{equation}
Define the omitted linear contribution by
\begin{equation}
 \mathcal E^{\mathrm{lin}}_{N,p}(t)
 :=-\e^{-a}L(\rho)-\lambda_{N,p}(t).
\end{equation}
Minkowski's inequality \cite{kn:folland99}, the above bounds
on coefficient mass, and \cref{eq:gamma-L2-identity} therefore give
\begin{equation}\label{eq:fixed-linear-tail-L2}
 \norm{\mathcal E^{\mathrm{lin}}_{N,p}}_2
 \leq C_pA_{N,p}\nu^{-5/4}.
\end{equation}
Next, we aim to control \cref{eq:fixed-product-linearization}.
For $d\in\{q,q+1\}$, define the positive majorant
\[
 M_d(a)
 :=\sum_{h\geq1}\sum_{j\geq0}
 B_{d,h,j}(p)a^{d+h+j}
 =\sum_{k\geq d+1}m_{d,k}\frac{a^k}{k!},
 \qquad m_{d,k}\geq0,
\]
where $B_{d,h,j}(p)$ is given by
\cref{eq:scaled-block-coefficient}. Thus
$|T_d(\rho)|\leq M_d(a)$ when $\rho=|t|/\sqrt p$ and
$a=|t|^2/4$. Define its coefficient mass by
\[
 \operatorname{mass}(M_d):=\sum_{k\geq d+1}m_{d,k};
\]
the coefficients here are taken with respect to the normalised
monomials $a^k/k!$.
The ratios \cref{eq:block-j-ratio,eq:block-h-ratio} give
\begin{align}
 \operatorname{mass}(M_q)
 &\leq \frac{A_{N,p}}{p-r}
 \left(\sum_{u\geq0}\frac{u+1}{(p\nu)^u}\right)
 \left(\sum_{j\geq0}\frac{(4/p)^j}{j!}\right)
 \leq C_pA_{N,p},                                  \label{eq:q-majorant-mass}\\
 \operatorname{mass}(M_{q+1})
 &\leq \frac{C_p}{p^{\nu+1}(\nu+1)!}
 \leq \frac{C_pA_{N,p}}\nu.                         \label{eq:qplus-majorant-mass}
\end{align}
Every mode in the first line has degree at least $\nu$, and every mode in
the second has degree at least $\nu+1$.  Again, Minkowski's inequality and
\cref{eq:gamma-L2-identity} therefore yield
\begin{align}
 \norm{\e^{-a}T_q}_2
 &\leq\norm{\e^{-a}M_q}_2
 \leq C_pA_{N,p}\nu^{-1/4},                       \label{eq:q-trace-L2}\\
 \norm{\e^{-a}T_{q+1}}_2
 &\leq\norm{\e^{-a}M_{q+1}}_2
 \leq C_pA_{N,p}\nu^{-5/4}.                       \label{eq:qplus-trace-L2}
\end{align}
On the central disc, \cref{eq:common-trace-small} also gives
\begin{equation}
 \sup_{\rho\leq2c_0\nu}\{|T_q(\rho)|+|T_{q+1}(\rho)|\}\leq C/\nu.
\end{equation}
Consequently, \cref{eq:fixed-product-linearization} gives the complete
norm calculation
\begin{align}
 \norm{\e^{-a}\mathcal R_p^{\mathrm{det}}}_{L^2(\rho\leq2c_0\nu)}
 &\leq C_p\left\{
  \norm{T_q}_{L^\infty(\rho\leq2c_0\nu)}\norm{\e^{-a}T_q}_2
  +\norm{T_{q+1}}_{L^\infty(\rho\leq2c_0\nu)}
    \norm{\e^{-a}T_{q+1}}_2
 \right\}\notag\\
 &\leq C_pA_{N,p}\nu^{-5/4}.                     \label{eq:fixed-nonlinear-L2}
\end{align}
This is the required factorial-scale non-linear estimate; it uses the
special signed-trace control in
\cref{lem:bessel-determinant-linearization}, not a generic trace-norm
remainder.  Combining \cref{eq:fixed-linear-tail-L2,eq:fixed-nonlinear-L2}
gives
\begin{equation}\label{eq:small-frequency-L2}
 \norm{\mathcal E_{N,p}}_{L^2(\rho\leq2c_0\nu)}
 \leq C_pA_{N,p}\nu^{-5/4}.                     
\end{equation}

We now treat the whole complementary plane in the single variable $\rho$.
Write $\Psi_d(\rho)=\mathsf H_d(-\rho^2/4)$.  Propositions~2.3 and~2.4
of~\cite{kn:courteautjohanssonlambert24}, in the exact form used here, state
\begin{align}
 |\Psi_d(\rho)|
 &\leq\sqrt2\exp\!\left\{-
 \frac{d^2\rho^2}{4(2.05\rho^2 +d^2)}\right\},
 &&\rho\geq0,                                    \label{eq:CJL-plateau-exact}\\
 |\Psi_d(\rho)|^2
 &\leq\left(\frac d\rho\right)^d,
 &&\rho\geq\frac{7d}{4}.                        \label{eq:CJL-power-exact}
\end{align}
In their notation the corresponding function is denoted by $F_d$; it is the
characteristic function of $\Tr(U_d)$,
their radial variable $|\zeta|$ is our $\rho$, and the scaling is
$|t|=\sqrt p\,\rho$ (or $|t|=\sqrt m\,\rho$ in the growing-power
regimes). Thus \cref{eq:CJL-plateau-exact,eq:CJL-power-exact}
reproduce both the hypotheses and the change of variables used below.
Fix $\eta>0$ and put $\Lambda_q=\e^{\eta q}$.  From
\cref{eq:CJL-plateau-exact}, for all large $q$,
\begin{equation}\label{eq:fixed-plateau-integral}
 \int_{2c_0\nu\leq\rho\leq\Lambda_q}
 |\varphi_{N,p}(t)|^2\dd t
 \leq C_p\Lambda_q^2\e^{-c_pq^2}
 =o(A_{N,p}^2\nu^{-5/2}).
\end{equation}
For $\rho\geq\Lambda_q$, \cref{eq:CJL-power-exact} and
$N=(p-r)q+r(q+1)$ give
\begin{align}
 |\varphi_{N,p}(t)|^2
 &\leq
 \left(\frac q\rho\right)^{q(p-r)}
 \left(\frac{q+1}\rho\right)^{(q+1)r}
 \leq\left(\frac{q+1}\rho\right)^N,             \label{eq:fixed-power-tail-pointwise}\\
 \int_{\rho\geq\Lambda_q}|\varphi_{N,p}(t)|^2\dd t
 &\leq\frac{2\pi p}{N-2}(q+1)^N\Lambda_q^{2-N}
 =o(A_{N,p}^2\nu^{-5/2}).                         \label{eq:fixed-power-tail-integral}
\end{align}
The terms $\e^{-a}$ and $\lambda_{N,p}(t)$ have exponentially
smaller complementary tails: here
$a=p\rho^2/4\geq pc_0^2\nu^2$, while the polynomial factor in
$\lambda_{N,p}$ has degree at most $2\nu-1$. The required bounds
therefore follow directly from incomplete Gamma integrals. Thus
\cref{eq:small-frequency-L2,eq:fixed-plateau-integral,eq:fixed-power-tail-integral}
give the claimed global $L^2$ estimate.

Finally, the pointwise tail estimate has an integrable square root for all
sufficiently large $N$; the bounded-frequency part is
continuous and bounded.  Hence $\varphi_{N,p}\in L^1(\R^2)$ and Fourier
inversion proves the existence of a continuous density $f_{N,p}$.
Plancherel then identifies the inverse transform of
\cref{eq:frequency-remainder-fixed} with
$f_{N,p}-g-\Lambda_{N,p}$ and proves
\cref{eq:determinant-L2-remainder}.
We must still convert the $L^2$ statement to total variation without losing
the leading scale.  The following truncation makes the
argument work for $p=1$ as well.

Write $h_{N,p}=f_{N,p}-g-\Lambda_{N,p}$ and choose
\begin{equation}\label{eq:spatial-cutoff}
 R_\nu=\bigl(8\log(A_{N,p}^{-1})\bigr)^{1/2}
 \asymp(\nu\log\nu)^{1/2}.
\end{equation}
On the disc $|x|\leq R_\nu$, Cauchy--Schwarz and
\cref{eq:determinant-L2-remainder} give
\begin{equation}\label{eq:inside-spatial-cutoff}
 \int_{|x|\leq R_\nu}|h_{N,p}(x)|\dd x
 \leq\sqrt\pi R_\nu\norm{h_{N,p}}_2
 \leq C_pA_{N,p}\nu^{-3/4}\sqrt{\log\nu}
 =o(A_{N,p}\nu^{-1/4}).
\end{equation}
Next, we control the three tails separately. Character orthogonality gives
\begin{equation}
 \mathbb E|\Tr(U_d)|^{2k}\leq k!,
 \qquad \mathsf H_d(z)\leq\e^z\quad(z\geq0).
\end{equation}
For every $v\in\R$, rotational invariance and
\cref{eq:H-product}, with $z=v^2/(4p)$, give
\begin{equation}
 \mathbb E\e^{v\operatorname{Re}X_{N,p}}
 \leq\e^{v^2/4}.
\end{equation}
The same bound holds for $\operatorname{Im}X_{N,p}$.
If $|X_{N,p}|>R$, then either
$|\operatorname{Re}X_{N,p}|>R/\sqrt2$ or
$|\operatorname{Im}X_{N,p}|>R/\sqrt2$.
Applying Chernoff's inequality to the upper and lower tails of
each component, followed by a union bound, gives
\begin{equation}\label{eq:fixed-subgaussian-tail}
 \mathbb P(|X_{N,p}|>R)\leq4\e^{-R^2/2}.
\end{equation}
Also $\mathbb P(|Z|>R)=\e^{-R^2}$.  Finally, Laguerre orthogonality and
Cauchy--Schwarz imply
\begin{equation}
 \int_{R^2}^\infty\e^{-u}|L_k(u)|\dd u\leq\e^{-R^2/2},
\end{equation}
while $\sum_{j=0}^qc_{q,j}p^{-j}\leq\e^{4/p}$.  Hence
\begin{equation}\label{eq:one-band-spatial-tail}
 \int_{|x|>R}|\Lambda_{N,p}(x)|\dd x
 \leq C_pA_{N,p}\e^{-R^2/2}.
\end{equation}
At $R=R_\nu$ all three tails are $O(A_{N,p}^4)$, which is
$o(A_{N,p}\nu^{-1/4})$.  Together with
\cref{eq:inside-spatial-cutoff}, this proves
\cref{eq:fixed-L1-remainder}.
\end{proof}

\subsection{Coherent summation of Laguerre modes}

\ifdefined\PedagogicalVersion\PedLaguerreSummation\fi

\paragraph{A uniform Laguerre envelope.}
There is an absolute constant $C$ such that, for every $n\geq1$,
\begin{equation}\label{eq:global-laguerre-bound}
 \e^{-u/2}|L_n(u)|\leq Cn^{-1/4}u^{-1/4},
 \qquad n^{-1}\leq u\leq n.
\end{equation}
\begin{proof}
Fix a sufficiently large constant $\omega$.  For
$n^{-1}\leq u\leq\omega$, Szeg\H{o}'s Hilb formula
\cite[Theorem~8.22.4 and Eq.~(8.22.5)]{kn:szego75}, with $\alpha=0$, gives
\begin{equation}
 \e^{-u/2}L_n(u)=J_0(2\sqrt{(n+1/2)u})+O(n^{-3/4}u^{5/4})
\end{equation}
uniformly.  The fixed-order Bessel estimate
$|J_0(v)|\leq C\min(1,v^{-1/2})$
\cite[Eq.~(10.17.3)]{kn:dlmf10} bounds the main term by
$Cn^{-1/4}u^{-1/4}$.  The remainder has the same bound because
\begin{equation}
 \frac{n^{-3/4}u^{5/4}}{n^{-1/4}u^{-1/4}}
 \leq C_\omega n^{-1/2}.
\end{equation}

For $\omega\leq u\leq n$, write
$u=(4n+2)\cos^2\phi$.  Then $\phi$ is bounded away from zero, while
$\pi/2-\phi\geq c\sqrt{\omega/n}$.  Taking $\omega$ large enough puts
the whole interval inside the uniform range
$\varepsilon\leq\phi\leq\pi/2-\varepsilon n^{-1/2}$ of Szeg\H{o}'s
Plancherel--Rotach formula \cite[Theorem~8.22.8(a)]{kn:szego75}. More
explicitly, when $\alpha=0$ that formula has, in magnitude, the form
\begin{equation}
 \e^{-u/2}|L_n(u)|
 \leq C(\sin\phi)^{-1/2}n^{-1/4}u^{-1/4}
       \{1+O((nu)^{-1/2})\}.
\end{equation}
Here $u\leq n$ gives
$\cos^2\phi\leq n/(4n+2)<1/4$, hence
$\sin\phi\geq\sqrt3/2$, while $u\geq\omega$ gives
$\pi/2-\phi\geq\cos\phi
=\sqrt{u/(4n+2)}\geq c_\omega n^{-1/2}$.  These are precisely the two
endpoint conditions in the cited uniform range.  Also $nu\geq n\omega$,
so its displayed remainder is uniform and no larger than the leading
amplitude.  Consequently that amplitude is
\begin{equation}
 O\!\left((\sin\phi)^{-1/2}n^{-1/4}u^{-1/4}\right)
 =O(n^{-1/4}u^{-1/4}),
\end{equation}
and the two ranges meet at $\omega$, proving
\cref{eq:global-laguerre-bound} with one constant.
\end{proof}

\paragraph{The averaging step.}
Let $0<\alpha<\beta<\infty$ and $W\in L^1([\alpha,\beta])$.  Then, for
every fixed $\theta\in\R$,
\begin{equation}\label{eq:periodic-averaging}
 \int_\alpha^\beta W(v)|\cos(\lambda v+\theta)|\dd v
 \longrightarrow \frac2\pi\int_\alpha^\beta W(v)\dd v
 \qquad(\lambda\to\infty).
\end{equation}
\begin{proof}
The Fourier series of $|\cos x|$ is absolutely convergent, with constant
term $2/\pi$ and nonconstant coefficients $O(k^{-2})$.  For each partial
sum, \cref{eq:periodic-averaging} follows from the Riemann--Lebesgue lemma.
Uniform convergence of the Fourier series then proves the assertion for
every $W\in L^1$.
\end{proof}

\begin{lemma}[Laguerre $L^1$ asymptotics]\label{lem:laguerre-L1}
As $n\to\infty$,
\begin{equation}\label{eq:laguerre-L1-constant}
 n^{1/4}\int_0^\infty\e^{-u}|L_n(u)|\dd u
 \longrightarrow
 C_L:=\frac{2^{7/4}\Gamma(3/4)}{\pi^{3/2}}.
\end{equation}
In addition, there is an absolute constant $C$ such that
\begin{equation}\label{eq:uniform-laguerre-L1}
 \int_0^\infty\e^{-u}|L_n(u)|\dd u
 \leq C(n+1)^{-1/4},\qquad n\geq0.
\end{equation}
For each fixed $j\geq0$,
\begin{equation}\label{eq:laguerre-shift}
 n^{1/4}\int_0^\infty\e^{-u}|L_{n+j}(u)-L_n(u)|\dd u
 \longrightarrow0.
\end{equation}
Consequently, for every fixed $p\geq1$,
\begin{equation}\label{eq:band-summation}
 \nu^{1/4}
 \norm{
   g\sum_{j=0}^q c_{q,j}p^{-j}L_{\nu+j}(|\cdot|^2)
 }_1
 \longrightarrow
 \e^{2/p}C_L.
\end{equation}
\end{lemma}

\begin{proof}
We use \cref{eq:global-laguerre-bound} for the estimates outside fixed compact
subintervals of $(0,\infty)$.

On every fixed compact subinterval of $(0,\infty)$, the same Hilb formula
and the Bessel asymptotic give, uniformly,
\begin{equation}
 \e^{-u}L_n(u)
 =
 \frac{n^{-1/4}}{\sqrt\pi}\,
 \e^{-u/2}u^{-1/4}
 \left[
  \cos\!\left(2\sqrt{nu}-\frac\pi4\right)+o(1)
 \right].
\end{equation}
We now split the half-line into five regions. On
$0\leq u\leq n^{-1}$, the defining series gives
$|L_n(u)|\leq\e^{nu}\leq\e$, so the contribution after multiplication by
$n^{1/4}$ is $O(n^{-3/4})$. Fix $0<\delta<M<\infty$; on $[n^{-1},\delta]$ and $[M,n]$ for sufficiently large $n$,
\cref{eq:global-laguerre-bound} gives
\begin{align}
 n^{1/4}\int_{n^{-1}}^\delta\e^{-u}|L_n(u)|\dd u
 &\leq C\delta^{3/4},                              \label{eq:laguerre-small-u}\\
 n^{1/4}\int_M^n\e^{-u}|L_n(u)|\dd u
 &\leq C\int_M^\infty\e^{-u/2}u^{-1/4}\dd u.       \label{eq:laguerre-middle-tail}
\end{align}
On $[n,\infty)$, Cauchy--Schwarz and Laguerre orthogonality give
\begin{align}
 \int_n^\infty\e^{-u}|L_n(u)|\dd u
 &\leq
 \left(\int_n^\infty\e^{-u}\dd u\right)^{1/2}
 \left(\int_n^\infty\e^{-u}L_n(u)^2\dd u\right)^{1/2}
 \leq\e^{-n/2}.                                   \label{eq:laguerre-far-tail}
\end{align}

It remains to analyse $[\delta,M]$.  Set $u=v^2$.  The leading weight
becomes $W(v)=2\e^{-v^2/2}v^{1/2}$ and the phase is
$2\sqrt n\,v-\pi/4$.  Thus the averaging identity
\eqref{eq:periodic-averaging} gives
\begin{align*}
 &\int_\delta^M\e^{-u/2}u^{-1/4}
 \left|\cos\!\left(2\sqrt{nu}-\frac\pi4\right)\right|\dd u\\
 &\hspace{5em}\longrightarrow
 \frac2\pi\int_\delta^M\e^{-u/2}u^{-1/4}\dd u.
\end{align*}
The compact-uniform error tends to zero.  First let $n\to\infty$, then
$\delta\downarrow0$, and finally $M\to\infty$.  The other four regions are
uniformly negligible by
\cref{eq:laguerre-small-u,eq:laguerre-middle-tail,eq:laguerre-far-tail}.
Consequently,
\begin{equation}
 C_L
 =
 \frac2{\pi^{3/2}}\int_0^\infty\e^{-u/2}u^{-1/4}\dd u
 =
 \frac{2^{7/4}\Gamma(3/4)}{\pi^{3/2}}.
\end{equation}
This proves \cref{eq:laguerre-L1-constant}.  Without taking limits, the
five-region estimates also prove \cref{eq:uniform-laguerre-L1} for $n\geq1$;
enlarging the constant covers $n=0$. For fixed $j$, on $[\delta,M]$ the amplitudes for $n+j$ and $n$ differ by
$o(1)$, while
\begin{equation}
 2\sqrt{(n+j)u}-2\sqrt{nu}=O_j(\sqrt{u/n})
\end{equation}
uniformly.  Hence the compact contribution to
$n^{1/4}\int\e^{-u}|L_{n+j}(u)-L_n(u)|\dd u$ tends to zero.  The same four
tail estimates apply to both indices.  The short interval between the two tail cutoffs is exponentially small: for fixed $j$, Cauchy--Schwarz and
Laguerre orthogonality give
\begin{equation}
 \int_n^{n+j}\e^{-u}|L_{n+j}(u)|\dd u
 \leq \left(\int_n^{n+j}\e^{-u}\dd u\right)^{1/2}
 \leq \e^{-n/2}.
\end{equation}
Taking first $n\to\infty$ and then $\delta\downarrow0$ and $M\to\infty$
proves \cref{eq:laguerre-shift}.

For \cref{eq:band-summation}, first truncate the $j$-sum.  For each fixed
truncation, use \cref{eq:cqj-limit,eq:laguerre-shift}.  The estimate
\begin{equation}
 c_{q,j}\leq\frac{4^j}{j!}
\end{equation}
together with \cref{eq:uniform-laguerre-L1} makes the remaining tail
uniformly summable even after multiplication by $\nu^{1/4}$.  Finally,
$\sum_{j\geq0}(2/p)^j/j!=\e^{2/p}$.
\end{proof}

\begin{proof}[Proof of \cref{thm:fixed-power}]
Combine \cref{prop:fixed-power-approximation,lem:laguerre-L1}.  For every fixed
$p\geq1$, the band summation and the shift estimate first give the sharper
Laguerre form
\begin{equation}
 \norm{\Lambda_{N,p}}_1
 =
 A_{N,p}\e^{2/p}
 \int_0^\infty\e^{-u}|L_\nu(u)|\dd u
+o(A_{N,p}\nu^{-1/4}).
\end{equation}
Here we used
$\int\e^{-u}|L_\nu(u)|\dd u=C_L\nu^{-1/4}+o(\nu^{-1/4})$
to replace the asymptotic constant in the band sum by the displayed
Laguerre integral.  Consequently, first
\begin{equation}
 \norm{f_{N,p}-g}_1
 =
 A_{N,p}\e^{2/p}
 \int_0^\infty\e^{-u}|L_\nu(u)|\dd u
 +o(A_{N,p}\nu^{-1/4}),
\end{equation}
and then, by \cref{eq:laguerre-L1-constant},
\begin{equation}
 \norm{f_{N,p}-g}_1
 =
 A_{N,p}\e^{2/p}C_L\nu^{-1/4}
 +o(A_{N,p}\nu^{-1/4}).
\end{equation}
As total variation is half the $L^1$ distance \eqref{eq:intro-TV-definition}, this proves
\cref{eq:fixed-laguerre-form,eq:fixed-explicit-form}.
\end{proof}

\begin{remark}[Check against the known $L^2$ constant]
When $p=1$, the same first-band calculation and
\begin{equation}
 \int_0^\infty\e^{-2u}L_n(u)^2\dd u
 =2^{-2n-1}\binom{2n}{n}
\end{equation}
give
\begin{equation}
 \norm{\F(f_{N,1}-g)}_2^2
 \sim
 \frac{2\e^4\sqrt\pi}{\Gamma(N+2)^2\sqrt N}.
\end{equation}
This is the exact Fourier-$L^2$ asymptotic of
Courteaut--Johansson--Lambert~\cite[Theorem~1.2]{kn:courteautjohanssonlambert24}.  Thus the
normalisation and the coherent factor $\e^2$ agree with their independently
derived constant; the additional step here is the signed comparison-function analysis
needed for total variation. Their Theorem~1.2 contains a normalisation typo: its right-hand side
requires the factor $(2\pi)^{-2}$, as follows from their identity~(1.5).
\end{remark}

\section{Growing powers with fixed block dimension}
\label{sec:fixed-quotient}

\ifdefined\PedagogicalVersion\PedFixedQuotient\fi

Our next aim is to prove \cref{thm:fixed-quotient}; this fills the regime
$m\asymp N$ below the high-power threshold.  The proof uses the
displacement determinant through \cref{cor:H-factorization}; it does not
invoke the power-map theorem.  The Charlier--Ewens formula determines the
first non-Gaussian coefficient of each fixed block, while the determinant
product controls all subsequent coefficients at once.

For $d\geq1$, define
\begin{equation}\label{eq:Gd-definition}
 G_d(z)=\e^{-z}\mathsf H_d(z).
\end{equation}
The Diaconis--Shahshahani moments give
$\mathsf H_d(z)-\e^z=O(z^{d+1})$.  At degree $d+1$, apply
\cref{thm:charlier-correction} with
$N=d$ and $\mu=\nu=(1^{d+1})$.  Here $D-N=1$, so only
$\mathbf c=\mathbf0$ occurs in the Charlier kernel and the correction is
$-1$.  Therefore
\begin{equation}
 \mathbb E|\Tr(U_d)|^{2(d+1)}=(d+1)!-1,
\end{equation}
and consequently
\begin{equation}\label{eq:Gd-edge-expansion}
 G_d(z)
 =
 1-\frac{z^{d+1}}{((d+1)!)^2}
 +z^{d+2}R_d(z),
\end{equation}
where $R_d$ is entire. This is the point at which the first
Charlier--Ewens correction band enters the analytic proof. Let $q,\nu,s,A_{m,r},B_{m,r}$ be as in
\cref{eq:quotient-remainder,eq:fixed-quotient-AB}, and put
$a=|t|^2/4$.  By \cref{eq:characteristic-H-product},
\begin{align}
 \varphi_{N,m}(t)
 &:=
 \mathbb E\exp\!\left(
  \ii t\cdot\frac{\Tr(U_N^m)}{\sqrt m}\right) \notag\\
 &=
 \e^{-a}G_q(-a/m)^sG_{q+1}(-a/m)^r.
                                                        \label{eq:fixed-quotient-product}
\end{align}
Introduce the two-mode signed comparison function
\begin{equation}\label{eq:fixed-quotient-leading-density}
 \Lambda^{(2)}_{m,r}(x)
 =
 g(x)\left\{
  (-1)^qA_{m,r}L_\nu(|x|^2)
  +(-1)^{q+1}B_{m,r}L_{\nu+1}(|x|^2)
 \right\}.
\end{equation}

\paragraph{A differentiated block-product estimate.}
Fix $\nu\geq2$ and $K<\infty$.  Let $s,r\geq0$ be integers with
$s+r=m$, and suppose that, on $0\leq a\leq\delta m$,
\begin{align*}
 u(a)&=\alpha(a/m)^\nu+u_1(a),
 &|\partial_a^ju_1(a)|&\leq
   \frac{K(1+a)^{\nu+1+j}}{m^{\nu+1}},\\
 v(a)&=\beta(a/m)^{\nu+1}+v_1(a),
 &|\partial_a^jv_1(a)|&\leq
   \frac{K(1+a)^{\nu+2+j}}{m^{\nu+2}},
 \qquad 0\leq j\leq2,
\end{align*}
where $|\alpha|+|\beta|\leq K$.  If $\delta=\delta(\nu,K)>0$ is
sufficiently small, then, for $0\leq j\leq2$,
\begin{align}
 &\left|\partial_a^j\left[
 \e^{-a}\left\{(1+u(a))^s(1+v(a))^r
 -1-s\alpha(a/m)^\nu-r\beta(a/m)^{\nu+1}\right\}
 \right]\right| \notag\\
 &\quad\leq
 C_{\nu,K}\e^{-a/2}(1+a)^{2\nu+2+2j}
 \left\{
  \frac{s}{m^{\nu+1}}+\frac{r}{m^{\nu+2}}
  +\left(\frac{s}{m^\nu}+\frac{r}{m^{\nu+1}}\right)^2
 \right\}.                                      \label{eq:abstract-product-bound}
\end{align}
\begin{proof}
The hypotheses and the two retained monomials imply, for $0\leq h\leq2$,
\begin{equation}
 |\partial_a^hu(a)|\leq
  \frac{C_{\nu,K}(1+a)^{\nu+h}}{m^\nu},
 \qquad
 |\partial_a^hv(a)|\leq
  \frac{C_{\nu,K}(1+a)^{\nu+1+h}}{m^{\nu+1}}.
\end{equation}
Moreover,
$s|u|+r|v|\leq C_{\nu,K}+C_{\nu,K}\delta^{\nu-1}a$ on the stated
interval.  To see the constant term explicitly, use
$(1+a)^h\leq2^{h-1}(1+a^h)$: then
\begin{equation}
 \frac{s(1+a)^{\nu+1}}{m^{\nu+1}}
 \leq C_\nu\{1+\delta^\nu a\},
 \qquad
 \frac{r(1+a)^{\nu+2}}{m^{\nu+2}}
 \leq C_\nu\{1+\delta^{\nu+1}a\},
\end{equation}
and the two retained monomials contribute at most
$C_{\nu,K}\delta^{\nu-1}a$.  Decrease $\delta$ so that the coefficient
of $a$ is at most $1/4$; hence
$s|u|+r|v|\leq C_{\nu,K}+a/4$. Expand the finite product by the binomial theorem.  Its constant and linear
parts are $1+su+rv$; every remaining term is
\begin{equation}
 \binom sk\binom r\ell u^kv^\ell,
 \qquad k+\ell\geq2.
\end{equation}
For $j=0,1,2$, differentiate these monomials by Leibniz's rule and use
$\binom sk\leq s^k/k!$, $\binom r\ell\leq r^\ell/\ell!$.  Summing the
resulting exponential series bounds all differentiated non-linear terms by
\begin{equation}
 C_{\nu,K}\e^{a/4}(1+a)^{2\nu+2+2j}
 \left(\frac{s}{m^\nu}+\frac{r}{m^{\nu+1}}\right)^2.
\end{equation}
For clarity, this summation can be organised by setting
$S_h=s|\partial_a^hu|+r|\partial_a^hv|$.  The derivative bounds give
$S_h\leq C_{\nu,K}(1+a)^{\nu+1+h}
(s/m^\nu+r/m^{\nu+1})$.  The terms of total binomial degree at least two
are bounded, for $j=0,1,2$, by $C\e^{S_0}$ times, respectively,
$S_0^2$, $S_0S_1$, and $S_1^2+S_0S_2$.  Since
$S_0\leq C_{\nu,K}+a/4$, these three bounds imply the above formula
after absorbing $\e^{C_{\nu,K}}$ into the constant.
The difference between $su+rv$ and the two displayed monomials in
\cref{eq:abstract-product-bound} is $su_1+rv_1$; its first two derivatives
are bounded by the first two terms in braces in
\cref{eq:abstract-product-bound}, with a smaller polynomial factor.
Finally, differentiating the common factor $\e^{-a}$ does not change its
size, and $\e^{-a}\e^{a/4}\leq\e^{-a/2}$.  This proves the claim.
\end{proof}

\begin{proposition}[Fixed-quotient approximation]
\label{prop:fixed-quotient-approximation}
For fixed $q\geq1$ and all sufficiently large $m$, the variable
$m^{-1/2}\Tr(U_N^m)$ has a density $f_{N,m}$ and, uniformly in
$0\leq r<m$,
\begin{align}
 \norm{f_{N,m}-g-\Lambda^{(2)}_{m,r}}_1
 \leq C_q\left\{
   \frac{s+1}{m^{\nu+1}}
  +\frac{(s+1)^2}{m^{2\nu}}
 \right\}
 +C_q\e^{-c_qm}.                                  \label{eq:fixed-quotient-L1-remainder}
\end{align}
\end{proposition}

\begin{proof}
We use the weighted Fourier transfer in
\cref{lem:laguerre-fourier}.  Expanding
\cref{eq:Gd-edge-expansion} at $-a/m$ shows that the Fourier transform of
$g+\Lambda^{(2)}_{m,r}$ is
\begin{align}
 \e^{-a}\left\{
 1+
 \frac{(-1)^q s}{m^\nu(\nu!)^2}a^\nu
 +\frac{(-1)^{q+1}r}{m^{\nu+1}((\nu+1)!)^2}
  a^{\nu+1}
 \right\}.                                      \label{eq:fixed-quotient-linear-fourier}
\end{align}
Equation \ref{eq:general-laguerre-fourier} supplies one factor
$1/k!$ and the coefficients in
\cref{eq:fixed-quotient-AB} supply the other. Thus it remains to estimate the product in
\cref{eq:fixed-quotient-product} after the two displayed leading
block terms have been removed. We first work in the central region $0\leq a\leq\delta_qm$, where
$\delta_q>0$ is fixed and sufficiently small.  Since the two functions $R_q,R_{q+1}$ in \cref{eq:Gd-edge-expansion} and their first two
derivatives are bounded on a fixed disc, put $x=a/m$ and write
\begin{align*}
 G_q(-x)&=1+\alpha_qx^\nu+x^{\nu+1}\tilde R_q(x),
 &\alpha_q&=\frac{(-1)^q}{(\nu!)^2},\\
 G_{q+1}(-x)&=1+\alpha_{q+1}x^{\nu+1}
                  +x^{\nu+2}\tilde R_{q+1}(x),
 &\alpha_{q+1}&=\frac{(-1)^{q+1}}{((\nu+1)!)^2}.
\end{align*}
For $0\leq j\leq2$, differentiation gives
\begin{align}
 \left|\partial_a^j\{G_q(-a/m)-1-\alpha_q(a/m)^\nu\}\right|
 &\leq \frac{C_q(1+a)^{\nu+1+j}}{m^{\nu+1}},
                                                        \label{eq:fixed-q-Taylor-derivative}\\
 \left|\partial_a^j\{G_{q+1}(-a/m)-1
              -\alpha_{q+1}(a/m)^{\nu+1}\}\right|
 &\leq \frac{C_q(1+a)^{\nu+2+j}}{m^{\nu+2}}.          \label{eq:fixed-qplus-Taylor-derivative}
\end{align}
Apply the block-product estimate \eqref{eq:abstract-product-bound} with
\[
 u(a)=G_q(-a/m)-1,
 \quad v(a)=G_{q+1}(-a/m)-1,
 \quad \alpha=\alpha_q,
 \quad \beta=\alpha_{q+1}.
\]
The hypotheses are precisely
\cref{eq:fixed-q-Taylor-derivative,eq:fixed-qplus-Taylor-derivative}.
It follows, for $0\leq j\leq2$, that
\begin{align}
 |\partial_a^j\mathcal E_{m,r}(a)|
 &\leq C_q\e^{-a/2}(1+a)^{2\nu+2+2j}
 \left\{
  \frac{s}{m^{\nu+1}}
 +\frac{r}{m^{\nu+2}}
 +\left(\frac{s}{m^\nu}
        +\frac{r}{m^{\nu+1}}\right)^2
 \right\},                             \label{eq:fixed-quotient-central-derivatives}
\end{align}
where $\mathcal E_{m,r}$ is the difference between
\cref{eq:fixed-quotient-product} and
\cref{eq:fixed-quotient-linear-fourier}.
Thus \cref{eq:fixed-quotient-central-derivatives} is the complete uniform
product estimate, including the endpoint cases $r=0$ and $s=1$.

Since polynomial multiples of $\e^{-c_qa}$ are square integrable,
\cref{eq:fixed-quotient-central-derivatives,eq:radial-laplacian-a} imply
\begin{align}
 \norm{\mathcal E_{m,r}}_{L^2(a\leq\delta_qm)}
 +\norm{\Delta_t\mathcal E_{m,r}}_{L^2(a\leq\delta_qm)}
 &\leq C_q\left\{
  \frac{s}{m^{\nu+1}}
 +\frac{r}{m^{\nu+2}}
 +\left(\frac{s}{m^\nu}
        +\frac{r}{m^{\nu+1}}\right)^2
 \right\}.                                      \label{eq:fixed-quotient-central-H2}
\end{align}
We next verify that extending this estimate to all frequencies costs only
$O_q(\e^{-c_qm})$. Put $u=t/\sqrt m=(u_1,u_2)\in\R^2$ and define
\begin{equation}
 \psi_d(u)=\mathsf H_d(-|u|^2/4),
 \qquad d\in\{q,q+1\}.
\end{equation}
These are the characteristic functions of $\Tr(U_d)$.  On every compact
annulus $\varepsilon\leq|u|\leq R$,
$|\psi_d(u)|<1$; equality in the triangle inequality would force
$u\cdot\Tr(U_d)$ to be constant modulo $2\pi$, which it is not.
Compactness therefore gives a common $\kappa_{q,\varepsilon,R}<1$.
For a multi-index $\alpha=(\alpha_1,\alpha_2)
\in\mathbb Z_{\geq0}^2$, write
$|\alpha|=\alpha_1+\alpha_2$ and
$\partial_u^\alpha
=\partial_{u_1}^{\alpha_1}\partial_{u_2}^{\alpha_2}$. The fixed-size Bessel determinant
\cref{eq:H-definition} and the standard fixed-order Bessel bounds give,
for $|\alpha|\leq4$,
\begin{equation}\label{eq:fixed-block-large-frequency}
 |\partial^{\alpha}_{u}\psi_d(u)|
 \leq C_q(1+|u|)^{-1/2}.
\end{equation}
At $z=-|u|^2/4$ the entries in \cref{eq:H-definition} satisfy
$I_k(\ii|u|)=\ii^kJ_k(|u|)$; the phases are removed by diagonal
conjugation.  Expanding the fixed $d\times d$ determinant, every derivative
of order at most four is a finite sum of products of fixed-order Bessel
functions and their derivatives.  The recurrence relations reduce these
derivatives to neighbouring fixed orders, and the standard fixed-order
estimate \cite[Eq.~(10.17.3)]{kn:dlmf10} makes each such factor
$O_q(|u|^{-1/2})$ for $|u|\geq1$.  The determinant expansion is therefore
$O_q(|u|^{-d/2})$, hence $O_q(|u|^{-1/2})$; boundedness on $|u|\leq1$
follows from analyticity.  This proves \cref{eq:fixed-block-large-frequency}.
The exponent $1/2$ is sufficient and is sharp for the one-dimensional
block. Fix the constant $C_q\geq1$ in
\cref{eq:fixed-block-large-frequency}, and choose
$\tau_q>2\sqrt{\delta_q}$ so large that
$C_q^2/(1+\tau_q)<1$. Put
\begin{equation}
\kappa_q:=\kappa_{q,\,2\sqrt{\delta_q},\,\tau_q}<1.
\end{equation}
In \cref{eq:fixed-quotient-product}, the two block counts add to $m$.
On $2\sqrt{\delta_q}\leq|u|\leq\tau_q$, apply the
compact-annulus bound to every factor. After at most two
derivatives, at most two factors lose their bound by $\kappa_q$,
so the resulting bound is $C_qm^2\kappa_q^{m-2}$.
For $|u|\geq\tau_q$, keeping the constant contributed by every
factor, Leibniz's rule gives, for $|\alpha|\leq2$,
\begin{equation}\label{eq:fixed-quotient-product-tail}
 \left|\partial_u^\alpha
   \{\psi_q(u)^s\psi_{q+1}(u)^r\}\right|
 \leq C_qm^2C_q^m(1+|u|)^{-m/2}.
\end{equation}
Consequently,
\begin{align*}
 C_q^{2m}\int_{\tau_q}^\infty u(1+u)^{-m}\dd u
 &\leq
 \frac{C_q^{2m}(1+\tau_q)^{2-m}}{m-2}
 \leq C_q\e^{-c_qm}.
\end{align*}
Changing variables $t=\sqrt m\,u$ introduces only polynomial factors in
$m$; derivatives in $t$ introduce compensating powers of $m^{-1/2}$.
Thus the $L^2$ norms of the exact characteristic function and its
Laplacian on $a>\delta_qm$ are $O_q(\e^{-c_qm})$.  The Gaussian and the two
polynomial comparison modes have the same bound by incomplete Gamma
estimates.  Hence
\begin{equation}\label{eq:fixed-quotient-global-H2}
 \norm{\mathcal E_{m,r}}_2+\norm{\Delta_t\mathcal E_{m,r}}_2
 \leq C_q\left\{
  \frac{s+1}{m^{\nu+1}}
 +\frac{(s+1)^2}{m^{2\nu}}
 \right\}
 +C_q\e^{-c_qm},
\end{equation}
where we used $r<m$.  The same large-frequency estimate makes the
characteristic function integrable for large $m$, and therefore proves
existence of $f_{N,m}$. Lastly apply \cref{eq:weighted-fourier-transfer} to
\cref{eq:fixed-quotient-global-H2}.  Fourier inversion of \cref{eq:fixed-quotient-linear-fourier}
gives $g+\Lambda^{(2)}_{m,r}$, proving
\cref{eq:fixed-quotient-L1-remainder}.
\end{proof}

\begin{proof}[Proof of \cref{thm:fixed-quotient}]
Let
\[
 \eta_{m,r}=\frac{s+1}{m^\nu}.
\]
Because $1\leq s\leq m$, the right side of
\cref{eq:fixed-quotient-L1-remainder} is $o_q(\eta_{m,r})$:
the two algebraic error ratios are at most $m^{-1}$ and
$2m^{-q}$, respectively.  By polar coordinates,
\begin{align*}
 \norm{\Lambda^{(2)}_{m,r}}_1
 &=
 \int_0^\infty\e^{-u}
 \left|
   (-1)^qA_{m,r}L_\nu(u)
  +(-1)^{q+1}B_{m,r}L_{\nu+1}(u)
 \right|\dd u.
\end{align*}
The reverse triangle inequality for norms, followed by
\cref{prop:fixed-quotient-approximation}, now gives
\cref{eq:fixed-quotient-two-mode} after division by two.

If $r/m\to\theta<1$, then
$A_{m,r}\sim(1-\theta)/(m^q\Gamma(q+2))$ and
$B_{m,r}=O(m^{-q-1})$, proving
\cref{eq:fixed-quotient-interior}.  If $s\to c$, then
\begin{equation}
 m^{q+1}A_{m,r}\to\frac{c}{\Gamma(q+2)},
 \qquad
 m^{q+1}B_{m,r}\to\frac1{\Gamma(q+3)},
\end{equation}
which gives \cref{eq:fixed-quotient-endpoint}.
\end{proof}

\begin{remark}[The first strip and the threshold]
\label{rem:fixed-quotient-first-strip}
For $q=1$, the theorem covers $N/2<m\leq N$.
If $r/m\to\theta<1$, the leading comparison function is
$-(2m-N)g(x)L_2(|x|^2)/(2m^2)$, reflecting the exact
fourth-moment defect
\begin{equation}
 \mathbb E|\Tr(U_N^m)|^4=2m^2-(2m-N).
\end{equation}
At $N=m$ one has $r=0$, and the leading term in
\cref{eq:fixed-quotient-two-mode} reduces to
$\frac1{4N}\int_0^\infty\e^{-u}|L_2(u)|\dd u$, exactly the constant in
\cref{thm:high-power}. Thus the below-threshold and above-threshold
answers meet without a normalisation gap.
\end{remark}

\section{Joint growth of power and block dimension}
\label{sec:joint-growth}

\ifdefined\PedagogicalVersion\PedJointGrowth\fi

We prove \cref{thm:joint-growth}.  Put $s=m-r$ and
\begin{equation}\label{eq:joint-A}
 A_{N,m}:=\frac{s}{m^\nu\Gamma(\nu+1)}.
\end{equation}
The displacement determinant gives, exactly as in
\cref{cor:H-factorization},
\begin{equation}\label{eq:joint-fredholm-product}
 \varphi_{N,m}(t)
 =\e^{-a}D_q(\rho)^sD_{q+1}(\rho)^r,
 \qquad a=\frac{|t|^2}{4},\quad \rho=\frac{|t|}{\sqrt m}.
\end{equation}

Let $T_d(\rho)=\Tr(Q_d\mathcal K(\rho)Q_d)$ and define the signed comparison function associated with the full linear Fredholm contribution by
\begin{equation}\label{eq:joint-linear-frequency}
 \lambda^{[1]}_{N,m}(t)
 :=
 -\e^{-a}\{sT_q(\rho)+rT_{q+1}(\rho)\},
 \qquad
 \Lambda^{[1]}_{N,m}:=\F^{-1}\lambda^{[1]}_{N,m}.
\end{equation}
We now show that the $\ell=1$ moment corrections in
\eqref{eq:charlier-correction}, summed with the weights
$(-a/m)^k/(k!)^2$ for the characteristic function, give
$\lambda^{[1]}_{N,m}$ in \eqref{eq:joint-linear-frequency}.

Put $\mu_k=(m^k)$ and
$K_k=K_{N,mk}^{\Ch}(\mu_k,\mu_k)$, with $K_0=0$.
The multiplicities in \eqref{eq:z-wick-partitions-def} are then
$a_m=k$ and $a_j=0$ for $j\ne m$.
For these latter indices, \eqref{eq:charlier-generating} gives
$P_{c_j}(0;1/j)=(-1/j)^{c_j}$.
Expanding the squared factor $P_{c_m}(k;1/m)^2$ using
\eqref{eq:charlier-generating}, and summing over the remaining
entries $c_j$, $j\ne m$, of the multi-index $\mathbf c$ in
\eqref{eq:charlier-kernel}, gives
\begin{equation}
 K_k=[u^{mk-N-1}]
 \frac{1}{(1-u)^2}
 \sum_{j=0}^{k}\binom{k}{j}^{2}j!m^j
 u^{mj}(1-u^m)^{2(k-j)}.
\end{equation}
Coefficients with negative index vanish.
Put $b=a/m$, so that $\rho=2\sqrt b$. The Bessel generating
function gives
\begin{equation}
 \sum_{n\geq0}\frac{(-b)^n}{(n!)^2}
       (v+v^{-1}-2)^n
 =
 \sum_{h\in\mathbb Z}(-1)^h I_h(\rho)^2v^h.
\end{equation}
Summing the preceding coefficient formula over $k=j+n$ and
using this identity with $v=u^m$, we obtain
\begin{align*}
 \sum_{k\geq0}K_k\frac{(-b)^k}{(k!)^2}
 &=
 \e^{-a}[u^{-N-1}]
 \frac{\sum_{h\in\mathbb Z}
       (-1)^h I_h(\rho)^2u^{mh}}
      {(1-u)^2}\\
 &=\e^{-a}\sum_{h>q}(mh-N)(-1)^hI_h(\rho)^2\\
 &=\e^{-a}\{sT_q(\rho)+rT_{q+1}(\rho)\}.
\end{align*}
Here the last equality uses
$mh-N=s(h-q)+r(h-q-1)$ for $h>q$ and
\eqref{eq:trace-bessel-square}.
These calculations are coefficient-wise in $b$.
Since the first Charlier--Ewens band contributes $-K_k$ to the
$2k$th moment, its generating function is exactly
$\lambda^{[1]}_{N,m}$. Subtracting it from the full expansion
identifies the sum of the higher bands with the non-linear
Fredholm remainder. Finally set the leading one-mode signed comparison function
\begin{equation}\label{eq:joint-leading-density}
 \Lambda^{(0)}_{N,m}(x)
 :=
 (-1)^qA_{N,m}g(x)L_\nu(|x|^2).
\end{equation}

\begin{proposition}[Joint-growth approximation]
\label{prop:joint-growth-approximation}
As $m,q\to\infty$, uniformly in $0\le r<m$,
\begin{equation}\label{eq:joint-first-band-bound}
 \norm{\Lambda^{[1]}_{N,m}-\Lambda^{(0)}_{N,m}}_1
 \leq
 CA_{N,m}\nu^{-1/4}\left(\frac1m+\frac1\nu\right).
\end{equation}
For all sufficiently large $m,q$, the law of
$m^{-1/2}\Tr(U_N^m)$ has a continuous density $f_{N,m}$ and
\begin{equation}\label{eq:joint-nonlinear-L1}
 \norm{f_{N,m}-g-\Lambda^{[1]}_{N,m}}_1
 =o(A_{N,m}\nu^{-1/4}),
\end{equation}
uniformly in $r$.
\end{proposition}

\begin{proof}
Use \cref{eq:trace-bessel-square,eq:bessel-square-series} with $p=m$.
After Fourier inversion by \cref{lem:laguerre-fourier}, the absolute
coefficient attached to the $q$ block, displacement $h\geq1$, and Bessel
degree $j\geq0$ is
\begin{equation}\label{eq:joint-bhj}
 b_{h,j}
 =
 \frac{sh}{m^{q+h+j}}\,
 \frac{(2(q+h)+j+1)_j}
 {j!(q+h+j)!}.
\end{equation}
In particular, $b_{1,0}=A_{N,m}$. We divide through as follows:
\begin{align}
 \frac{b_{h,j+1}}{b_{h,j}}
 &=
 \frac{(2q+2h+2j+2)(2q+2h+2j+1)}
 {m(j+1)(2q+2h+j+1)(q+h+j+1)}
 \leq\frac4{m(j+1)},                              \label{eq:joint-j-ratio}\\
 \frac{b_{h,0}}{A_{N,m}}
 &=\frac{h}{m^{h-1}}
   \frac{(q+1)!}{(q+h)!}
 \leq\frac{h}{(m\nu)^{h-1}}.                     \label{eq:joint-h-ratio}
\end{align}
The combined contribution from the $r$ blocks of dimension $q+1$
begins with a coefficient of magnitude
$r/[m^{\nu+1}\Gamma(\nu+2)]$, and
\begin{equation}\label{eq:joint-long-ratio}
 \frac{r}{m^{\nu+1}\Gamma(\nu+2)}
 \bigg/A_{N,m}
 =\frac{r}{sm(\nu+1)}
 \leq\frac1{\nu+1}.
\end{equation}
The remaining coefficients satisfy the same two ratio bounds.
Now apply \cref{eq:uniform-laguerre-L1}.  Summing
\cref{eq:joint-j-ratio} for $j\geq1$, then
\cref{eq:joint-h-ratio} for $h\geq2$, and lastly
\cref{eq:joint-long-ratio}, proves
\cref{eq:joint-first-band-bound}.  Notice that the estimate is uniform
even when $s=m-r=1$.

We next control all non-linear terms in the Fredholm determinant expansion uniformly. Write
$A=A_{N,m}$ and first work with characteristic functions.
Set
\begin{equation}\label{eq:joint-frequency-error}
 \mathcal E_{N,m}(t)
 =\varphi_{N,m}(t)-\e^{-a}-\lambda^{[1]}_{N,m}(t).
\end{equation}
For $\rho\leq2c_0\nu$, use the logarithmic form in
\cref{lem:bessel-determinant-linearization}:
\begin{align}
 D_q^sD_{q+1}^r
 &=\exp\{-L-R\},                                  \label{eq:joint-log-product}\\
 L&=sT_q+rT_{q+1},
 &R&=s\mathfrak r_q+r\mathfrak r_{q+1},            \label{eq:joint-LR}\\
 |R|&\leq C\{s|T_q|^2+r|T_{q+1}|^2\}.             \label{eq:joint-R-bound}
\end{align}
This avoids expanding a growing number of determinant factors.

We make the uniformity of the exponent explicit.  Put
$y=\rho/(2\nu)\leq c_0$.  From
\cref{eq:common-trace-factorial} and Stirling's lower bound,
\begin{equation}\label{eq:joint-y-trace}
 |T_q(\rho)|+|T_{q+1}(\rho)|
 \leq \frac C\nu\{y\e^{1+y^2}\}^{2\nu}.
\end{equation}
Since $y\e^{1+y^2}\leq y/c_0$ on $[0,c_0]$, the right side is at most
$Cy^2/\nu$.  As $a=m\rho^2/4=m\nu^2y^2$, this gives
\begin{equation}\label{eq:joint-exponent-small}
 |L|\leq\frac{Ca}{\nu^3},
 \qquad
 |R|\leq\frac C\nu\{s|T_q|+r|T_{q+1}|\}
       \leq\frac{Ca}{\nu^4}.
\end{equation}
In particular, $|L|+|R|\leq a/4$ for all sufficiently large $q$, uniformly
in $m$ and $r$.
The sharper factorial form of the same estimate is
\begin{align}
 s|T_q(\rho)|
 &\leq CA\frac{a^\nu}{\nu!}
       \exp\!\left(\frac{2a}{m\nu}\right),       \label{eq:joint-sharp-q-trace}\\
 r|T_{q+1}(\rho)|
 &\leq \frac{CA}{\nu}\frac{a^{\nu+1}}{(\nu+1)!}
       \exp\!\left(\frac{2a}{m(\nu+1)}\right).  \label{eq:joint-sharp-qplus-trace}
\end{align}
For the second estimate, the ratio of its original prefactor to the displayed
one is $r\nu/[sm(\nu+1)]\leq1$, ensuring uniformity even when $s=1$.
Since $s,r$ are nonnegative integers, \cref{eq:joint-R-bound} gives
\begin{equation}
 |R|\leq C\{s|T_q|^2+r|T_{q+1}|^2\}
 \leq C\{s|T_q|+r|T_{q+1}|\}^2.
\end{equation}
Writing $z=L+R$, the definition of the error becomes
\begin{equation}\label{eq:joint-quadratic-decomposition}
 \mathcal E_{N,m}(t)
 =\e^{-a}\{(\e^{-z}-1+z)-R\}.
\end{equation}
By \cref{eq:joint-exponent-small},
$|z|\leq C\{s|T_q|+r|T_{q+1}|\}$ and $|z|\leq a/4$
for all sufficiently large $q$.
We therefore apply
$|\e^{-z}-1+z|\leq\frac12|z|^2\e^{|z|}$
and then the trace estimates
\cref{eq:joint-sharp-q-trace,eq:joint-sharp-qplus-trace}. As
$m\nu\to\infty$, the harmless exponentials in those two trace estimates
are absorbed into $\e^{-a/2}$, and we obtain
\begin{equation}\label{eq:joint-central-pointwise}
 |\mathcal E_{N,m}(t)|
 \leq CA^2\e^{-a/2}
 \left\{
  \frac{a^\nu}{\nu!}
  +\frac1\nu\frac{a^{\nu+1}}{(\nu+1)!}
 \right\}^2.
\end{equation}
As $\dd t=4\pi\dd a$ for radial integrals, extending the integral to
$a=\infty$ and using Gamma integrals gives
\begin{align*}
 \norm{\mathcal E_{N,m}}_{L^2(\rho\leq2c_0\nu)}^2
 &\leq CA^4\int_0^\infty\e^{-a}
 \left\{\frac{a^\nu}{\nu!}
 +\frac1\nu\frac{a^{\nu+1}}{(\nu+1)!}\right\}^4\dd a\\
 &\leq C^{2\nu}A^4,
\end{align*}
where the last step follows term by term from
$\Gamma(4\nu+j+1)/(\nu!)^4\leq C^{2\nu}\nu^{O(1)}$ for $j \leq 4$ and Stirling's formula.
Absorbing the polynomial factor into the exponential constant yields
\begin{equation}\label{eq:joint-central-L2}
 \norm{\mathcal E_{N,m}}_{L^2(\rho\leq2c_0\nu)}
 \leq C^\nu A^2.
\end{equation}
We next use $\rho$ consistently on the complementary plane.  Let
$\Lambda_q=\e^{\eta q^2}$, where $\eta>0$ is a sufficiently small absolute
constant. The exact Courteaut--Johansson--Lambert estimates are
\cref{eq:CJL-plateau-exact,eq:CJL-power-exact}.  On
$2c_0\nu\leq\rho\leq\Lambda_q$ they give
\begin{equation}
 |\varphi_{N,m}(t)|
 \leq 2^{m/2}
 \exp\!\left\{-\frac{sq^2+r(q+1)^2}{18}\right\}.
\end{equation}
After including the $t$-area $O(m\Lambda_q^2)$,
\begin{equation}\label{eq:joint-plateau-integral}
 \int_{2c_0\nu\leq\rho\leq\Lambda_q}
 |\varphi_{N,m}(t)|^2\dd t
 \leq C\e^{-cmq^2}.
\end{equation}
For $\rho\geq\Lambda_q$, \cref{eq:CJL-power-exact} gives
\begin{align}
 |\varphi_{N,m}(t)|^2
 &\leq
 \left(\frac q\rho\right)^{sq}
 \left(\frac{q+1}\rho\right)^{r(q+1)}
 \leq\left(\frac{q+1}\rho\right)^N,\\
 \int_{\rho\geq\Lambda_q}|\varphi_{N,m}(t)|^2\dd t
 &\leq\frac{2\pi m}{N-2}(q+1)^N\Lambda_q^{2-N}
 \leq C\e^{-cmq^2}.                              \label{eq:joint-power-tail-integral}
\end{align}

We control the linear comparison on the same $\rho$-scale.  The Bessel
recurrence
$nI_n(\rho)=\frac\rho2(I_{n-1}(\rho)-I_{n+1}(\rho))$ and
$\sum_{n\in\mathbb Z}I_n(\rho)=\e^\rho$ imply
\begin{equation}\label{eq:joint-linear-global-bound}
 |T_d(\rho)|
 \leq\sum_{n\geq1}nI_n(\rho)^2
 \leq\rho\e^{2\rho}.
\end{equation}
Thus
\begin{equation}
 \e^{-a}|L|\leq m\rho\exp(-m\rho^2/4+2\rho).
\end{equation}
Its squared integral from $\rho=2c_0\nu$ onward is
$O(\e^{-cmq^2})$, and the Gaussian tail is smaller.  Together with
\cref{eq:joint-plateau-integral,eq:joint-power-tail-integral}, this proves
\begin{equation}\label{eq:joint-frequency-tail}
 \norm{\mathcal E_{N,m}}_{L^2(\rho>2c_0\nu)}
 \leq C\e^{-cmq^2}.
\end{equation}
The square root of the power-tail bound is integrable for large $N$, so
$\varphi_{N,m}\in L^1$ and Fourier inversion now proves that the law has a
continuous density $f_{N,m}$.  Plancherel and
\cref{eq:joint-central-L2,eq:joint-frequency-tail} give
\begin{equation}\label{eq:joint-global-L2-explicit}
 \norm{f_{N,m}-g-\Lambda^{[1]}_{N,m}}_2
 \leq C^\nu A^2+C\e^{-cmq^2}.
\end{equation}

It remains to compare both errors uniformly with the leading scale.  Put
\begin{equation}
 R_{N,m}=\sqrt{8\log(A^{-1})}.
\end{equation}
As $s\geq1$, Stirling's formula gives
\begin{equation}\label{eq:joint-logA-bound}
 -\log A\leq C\nu(\log m+\log\nu).
\end{equation}
Moreover,
\begin{align}
 R_{N,m}\nu^{1/4}C^\nu A&\longrightarrow0,
                                                        \label{eq:joint-central-scale-comparison}\\
 \frac{R_{N,m}\nu^{1/4}\e^{-cmq^2}}{A}&\longrightarrow0. 
                                                        \label{eq:joint-far-scale-comparison}
\end{align}
For the second limit, explicitly,
\begin{equation}
 \frac{mq^2}{\nu(\log m+\log\nu)}\longrightarrow\infty
\end{equation}
for arbitrary $m,q\to\infty$; the first follows from
$A\leq m^{1-\nu}/\Gamma(\nu+1)$.  Hence
\cref{eq:joint-global-L2-explicit} implies
\begin{equation}\label{eq:joint-density-L2}
 \norm{f_{N,m}-g-\Lambda^{[1]}_{N,m}}_2
 =o(A\nu^{-1/4}R_{N,m}^{-1}).
\end{equation}

On $|x|\leq R_{N,m}$ use Cauchy--Schwarz and
\cref{eq:joint-density-L2}.  Outside this disc, the moment-generating
argument in \cref{eq:fixed-subgaussian-tail} is uniform in $m$, so the law
and the Gaussian have mass $O(\e^{-R_{N,m}^2/2})$. To conclude,
\cref{eq:joint-j-ratio,eq:joint-h-ratio,eq:joint-long-ratio} show that the total absolute
coefficient mass of $\Lambda^{[1]}_{N,m}$ is $O(A)$; Laguerre
orthogonality bounds its mass outside this disc by
$O(A\e^{-R_{N,m}^2/2})$. As
$\e^{-R_{N,m}^2/2}=A^4$, these tails are negligible and
\cref{eq:joint-nonlinear-L1} follows.
\end{proof}

\begin{proof}[Proof of \cref{thm:joint-growth}]
By \cref{prop:joint-growth-approximation},
\begin{equation}
 \norm{f_{N,m}-g}_1
 =
 A_{N,m}\int_0^\infty\e^{-u}|L_\nu(u)|\dd u
+o(A_{N,m}\nu^{-1/4}).
\end{equation}
Divide by two and use \cref{lem:laguerre-L1} to obtain
\cref{eq:joint-laguerre-form,eq:joint-explicit-form}.
\end{proof}

\section{The independent-phase regime \texorpdfstring{$m_N\geq N$}{mN at least N}}
\label{sec:high-power}

\ifdefined\PedagogicalVersion\PedHighPower\fi

We now prove \cref{thm:high-power}. This regime has a different
probabilistic mechanism, so only the Fourier--Laguerre transfer from
\cref{lem:laguerre-fourier} is shared with the fixed-power proof. By Rains's power map theorem \eqref{eq:rains-decomposition}, for every $m\geq N$
the eigenvalues of $U_N^m$ are independent and uniform on the unit circle.
Therefore
\begin{equation}\label{eq:rains-high-power}
 \frac{\Tr(U_N^m)}{\sqrt N}
 \mathrel{\overset{\mathrm d}=}
 S_N:=\frac1{\sqrt N}\sum_{j=1}^N\xi_j,
 \qquad \xi_j=\e^{\ii\Theta_j},
\end{equation}
where the $\Theta_j$ are independent and uniform on $[0,2\pi]$.
Let $J=J_0$.  Then
\begin{equation}
 \Phi_N(t):=\mathbb E\e^{\ii t\cdot S_N}
 =J(|t|/\sqrt N)^N.
\end{equation}
Define
\begin{equation}\label{eq:high-remainder}
 R_N(t)=
 J(|t|/\sqrt N)^N
 -\e^{-|t|^2/4}\left(1-\frac{|t|^4}{64N}\right).
\end{equation}

\paragraph{A Fourier remainder estimate.}
\begin{equation}\label{eq:high-H2}
 \norm{R_N}_2+\norm{\Delta R_N}_2=O(N^{-2}).
\end{equation}

\begin{proof}
Choose $\delta>0$ below the first positive zero of $J$.  The expansion
\begin{equation}
 \log J(s)=-\frac{s^2}{4}-\frac{s^4}{64}+s^6\eta(s^2)
\end{equation}
holds for $|s|\leq\delta$, with $\eta$ analytic.  Put $y=|t|^2$,
$T_N=N^{1/8}$,
\begin{equation}
 a_N(y)=\frac{y^2}{64N},\qquad
 b_N(y)=\frac{y^3}{N^2}\eta(y/N),\qquad
 u_N=-a_N+b_N.
\end{equation}
Using the exact Taylor remainder
\begin{equation}
 \e^u=1+u+u^2\int_0^1(1-\theta)\e^{\theta u}\dd\theta
\end{equation}
and differentiating this identity gives, for $j=0,1,2$ and
$0\leq y\leq N^{1/4}$,
\begin{equation}
 \left|
  \frac{\dd^j}{\dd y^j}
  \left\{
   J(\sqrt{y/N})^N-\e^{-y/4}(1-a_N(y))
  \right\}
 \right|
 \leq \frac{C}{N^2}P_j(y)\e^{-y/8}.
\end{equation}
Here, $P_0,P_1,P_2$ are fixed polynomials with non-negative
coefficients, independent of $N$. The radial identity \cref{eq:radial-laplacian-a}, with $a=y/4$, controls
$|t|\leq T_N$. For $T_N\leq|t|\leq\delta\sqrt N$, analyticity gives
$\log J(s)\leq-cs^2$, $|(\log J)'(s)|\leq Cs$, and
$|(\log J)''(s)|\leq C$.  Thus the function and its Laplacian are bounded by
a fixed polynomial times $\e^{-c|t|^2}$, giving
$O(\e^{-cN^{1/4}})$.

For $s=|t|/\sqrt N$, the Bessel equation gives exactly
\begin{equation}
 \Delta_t(J(s)^N)
 =(N-1)J(s)^{N-2}J'(s)^2-J(s)^N.
\end{equation}
Choose $C_0,s_0$ so that
$|J(s)|+|J'(s)|\leq C_0s^{-1/2}$ for $s\geq s_0$, and then choose
$M>\max\{\delta,s_0,C_0^2\}$.  On $\delta\leq s\leq M$,
$\sup|J(s)|=\rho<1$, so the contribution is polynomial in $N$ times
$\rho^N$.  On $s\geq M$,
\begin{equation}
 |J(s)^N|+|\Delta_t(J(s)^N)|
 \leq CN C_0^Ns^{-N/2},
\end{equation}
whose squared integral is bounded by
\begin{equation}
 \frac{CN^3M^2}{N-2}\left(\frac{C_0^2}{M}\right)^N.
\end{equation}
The Gaussian comparison has exponentially small tails.  Combining these estimates proves \cref{eq:high-H2}.
\end{proof}

\begin{proof}[Proof of \cref{thm:high-power}]
The large-frequency estimate used to prove \eqref{eq:high-H2} also shows,
for all sufficiently large $N$, that
\begin{equation}
 \Phi_N\in L^1(\R^2).
\end{equation}
Fourier
inversion therefore gives a continuous density for $S_N$; denote it by
$f_{S_N}$.  Let $r_N=\F^{-1}R_N$.  By
\cref{eq:weighted-fourier-transfer,eq:high-H2},
$\norm{r_N}_1=O(N^{-2})$.  The case $k=2$ of
\cref{eq:general-laguerre-fourier} gives
\begin{equation}
 f_{S_N}(x)-g(x)
 =
 -\frac1{2N}L_2(|x|^2)g(x)+r_N(x).
\end{equation}
Hence
\begin{equation}
 d_{\TV}(S_N,Z)
 =
 \frac1{4N}\int_0^\infty\e^{-u}|L_2(u)|\dd u+O(N^{-2}).
\end{equation}
The roots of $L_2(u)=1-2u+u^2/2$ are $2\pm\sqrt2$, and
\begin{equation}
 \int_0^\infty\e^{-u}|L_2(u)|\dd u
 =
 4\e^{-2}\bigl(\sqrt2\cosh(\sqrt2)-\sinh(\sqrt2)\bigr).
\end{equation}
Now use \cref{eq:rains-high-power}.
\end{proof}

\section{Local Poisson approximation for spectra of \texorpdfstring{$U^m$}{U to the m}}
\label{sec:local-poisson}

\ifdefined\PedagogicalVersion\PedLocalPoisson\fi
We now prove \cref{thm:local-poisson}. The local mechanism
is different to those in the previous sections: Rains's decomposition \eqref{eq:rains-decomposition} produces independent CUE blocks, and,
when their largest intensity tends to zero, a fixed microscopic interval
sees every block, up to a smaller-order error, as a rare zero-or-one point
source. The leading correction to the local count distribution involves
the second Charlier polynomial.

Let $\mathcal N(B)$ be the space of finite simple counting measures on $B$,
equipped with its counting-measure sigma-algebra. We use the total-variation
distance defined in \eqref{eq:intro-TV-definition}. For a finite signed
measure $\sigma$, $\norm{\sigma}_{\mathrm{var}}$ denotes its variation norm;
we do not write $d_{\TV}$ when one of the arguments is signed.

\subsection{Spectral block decomposition at the local scale}

Recall the unfolded process $\Xi_{N,m}$ from
\eqref{eq:intro-unfolded-process}.
Fix $L>0$ and an interval $B\subset\mathbb T_N$ of length $L$.  When
$N>L$, rotation invariance makes its location immaterial.

With $q,r$ as in \eqref{eq:euclidiean-division}, define
\begin{equation}\label{eq:local-block-sizes}
 d_a=
 \begin{cases}
  q+1,&1\leq a\leq r,\\
  q,&r<a\leq m.
 \end{cases}
\end{equation}
Blocks of size zero are empty. Let $X_a$ denote the restriction to $B$
of the $a$th independent unfolded CUE block of dimension $d_a$. Thus, for every Borel set $A\subseteq B$, $X_a(A)$ counts the eigenangles of block $a$ whose unfolded positions $N\theta/(2\pi)$ lie in $A$.
Rains's power-map theorem \eqref{eq:rains-decomposition} identifies
$\Xi_{N,m}\restriction_B$ in law with $\sum_{a=1}^mX_a$.
We work with this representation throughout the section.
Each $X_a$ belongs to $\mathcal N(B)$ almost surely.
Before restriction, a block of dimension $d$ has kernel \cite[Eq.~(1.8)]{kn:soshnikov00}
\begin{equation}\label{eq:local-block-kernel}
 K^{\mathrm{CUE}}_{N,d}(x,y)
 =
 \frac1N\sum_{\ell=0}^{d-1}
 \e^{2\pi\ii\ell(x-y)/N},
\end{equation}
with respect to Lebesgue measure on $\mathbb T_N$; the $a$th block therefore has kernel $K^{\mathrm{CUE}}_{N,d_a}$.
Its intensity is constant:
\begin{equation}\label{eq:local-block-weight}
 w_a:=K^{\mathrm{CUE}}_{N,d_a}(x,x)=\frac{d_a}{N},
 \qquad
 \sum_{a=1}^m w_a=1.
\end{equation}
Consequently, the quantities in \cref{eq:main-local-gamma} can equivalently
be written as
\begin{equation}\label{eq:local-power-sums}
 \gamma_{j;N,m}=\sum_{a=1}^m w_a^j,\qquad
 \beta_{N,m}=\sum_{a=1}^m w_a^2,\qquad
 \varepsilon_{N,m}=\max_a w_a.
\end{equation}
In particular,
\begin{align}
 \beta_{N,m}
 &=
 \frac{r(q+1)^2+(m-r)q^2}{N^2} \notag\\
 &=
 \frac1m+\frac{r(m-r)}{mN^2}.                    \label{eq:local-beta-exact}
\end{align}
The second identity also holds when $m>N$: then $q=0$, $r=N$, and
$\beta_{N,m}=1/N$.

The power sums in \eqref{eq:local-power-sums} have a direct probabilistic
meaning; a window of length $L$ receives mean $Lw_a$ from block $a$.
Replacing the blocks by independent Bernoulli sources gives a
Poisson--binomial total count with parameter power sums
$L^j\gamma_{j;N,m}$. The sum of squared block intensities $\beta_{N,m}$
determines the leading Poisson approximation error. The cubic sum controls the next term in the
convolution logarithm, while $\gamma_{4;N,m}$ controls the total error in replacing
the CUE blocks by independent Bernoulli uniform processes.  Finally,
$\varepsilon_{N,m}$ is the largest individual block intensity and is the
small parameter needed for uniform Taylor expansion.

This hierarchy also explains the unrestricted condition in the theorem.  If
$m\leq N$, then $\varepsilon_{N,m}\asymp1/m$; if $m>N$, it is $1/N$.
Thus $N,m\to\infty$ is exactly what makes every source rare, regardless of
the relative rate of the two parameters.  Moreover,
\begin{equation}
\label{eq:local-power-sum-bounds}
 \gamma_{3;N,m}\leq\varepsilon_{N,m}\beta_{N,m},
 \qquad
 \gamma_{4;N,m}\leq\varepsilon_{N,m}^2\beta_{N,m},
 \qquad
 \beta_{N,m}\leq\varepsilon_{N,m}.
\end{equation}
Hence every remainder on the right of
\eqref{eq:main-local-poisson} is $o(\beta_{N,m})$.  The theorem's leading
term is therefore uniform across the transition $m\asymp N$, not merely on
fixed-ratio subsequences.

For later use, let $P_2(\,\cdot\,;L)$ be the monic Charlier polynomial
from \cref{eq:charlier-generating}.  Explicitly,
\begin{equation}\label{eq:local-P2}
 P_2(k;L)=(k-L)^2-k,
 \qquad
 A(L)=\frac14\mathbb E|P_2(Z_L;L)|.
\end{equation}

\subsection{Poisson--Charlier expansion for sparse block counts}

Let $\pi_L$ denote the law of $Z_L\sim\operatorname{Pois}(L)$ on
$\mathbb Z_{\geq0}$, put
\begin{equation}
 \mathfrak d=\delta_1-\delta_0,
\end{equation}
and write $*$ for convolution of signed measures on $\mathbb Z$.

\begin{proposition}[Poisson--Charlier transfer.]
For all $j,k\geq0$,
\begin{equation}\label{eq:poisson-charlier-transfer}
 (\pi_L*\mathfrak d^{*j})(k)
 =
 \pi_L(k)\frac{P_j(k;L)}{L^j}.
\end{equation} 
\end{proposition}
\begin{proof}
The generating function of the left side of
\cref{eq:poisson-charlier-transfer} is
\begin{equation}
 \e^{L(z-1)}(z-1)^j.
\end{equation}
On the other hand, multiply \cref{eq:charlier-generating}, with
$\alpha=L$, by $\pi_L(k)z^k$, sum over $k$, and compare coefficients of
$t^j/j!$.  This gives
\begin{equation}
 \sum_{k\geq0}\pi_L(k)P_j(k;L)z^k
 =
 L^j(z-1)^j\e^{L(z-1)},
\end{equation}
which proves the identity.
\end{proof}

For a sum of independent Bernoulli variables, the variance is smaller
than the mean by $\sum_a p_a^2$, where $p_a$ are the success probabilities. The next lemma shows how this difference appears in
the leading correction to the Poisson law and bounds the remainder.

\begin{lemma}[First Poisson--Charlier correction]
\label{lem:first-poisson-charlier}
Let $I_1,\ldots,I_s$ be independent Bernoulli variables with parameters
$p_1,\ldots,p_s$.  Suppose
\begin{equation}
 \sum_{a=1}^s p_a=L,
 \qquad
 p_*:=\max_a p_a\leq\frac14,
\end{equation}
put $S=\sum_a I_a$, and define
\begin{equation}
 \lambda_j=\sum_{a=1}^s p_a^j,\qquad j\geq2.
\end{equation}
Then
\begin{equation}\label{eq:first-poisson-charlier}
 \left|
 d_{\TV}\!\left(S,Z_L\right)
 -
 \frac{\lambda_2}{4L^2}\,
 \mathbb E|P_2(Z_L;L)|
 \right|
 \leq C_L(\lambda_3+\lambda_2^2).
\end{equation}
\end{lemma}

\begin{proof}
Let $\mathcal M=\ell^1(\mathbb Z)$ be the commutative Banach algebra of
finite signed measures, with convolution as multiplication and norm
$\norm{\mu}_1=\sum_k|\mu(k)|$. We regard $\mathcal L(S)$, the probability law of $S$, as an
element of $\mathcal M$. Write
\begin{equation}
 \exp_*(\mu)=\sum_{n\geq0}\frac{\mu^{*n}}{n!}.
\end{equation}
Here $\mu^{*n}$ denotes the $n$-fold convolution of $\mu$,
with $\mu^{*0}=\delta_0$; $\exp_*$ is the convolution exponential.
The Bernoulli law with parameter $p_a$ is
$\delta_0+p_a\mathfrak d$. As
$\norm{p_a\mathfrak d}_1=2p_a\leq\frac12$, its convolution logarithm is
absolutely convergent. We write $\mathop{*}_{a=1}^s\mu_a$ for the repeated convolution
$\mu_1*\cdots*\mu_s$. Independence therefore gives
\begin{align}
 \mathcal L(S)
 &=
 \mathop{*}_{a=1}^s(\delta_0+p_a\mathfrak d) \notag\\
 &=
 \pi_L*\exp_*(\eta),
 \qquad
 \eta=
 \sum_{j\geq2}\frac{(-1)^{j+1}\lambda_j}{j}
 \mathfrak d^{*j},                               \label{eq:bernoulli-convolution-log}
\end{align}
because $\pi_L=\exp_*(L\mathfrak d)$.

Split
\begin{equation}
 \eta_2=-\frac{\lambda_2}{2}\mathfrak d^{*2},
 \qquad
 \eta_{\geq3}=\eta-\eta_2.
\end{equation}
For $j\geq3$,
$\lambda_j\leq p_*^{\,j-3}\lambda_3$, while
$\norm{\mathfrak d^{*j}}_1=2^j$.  Hence
\begin{equation}\label{eq:eta-ge3}
 \norm{\eta_{\geq3}}_1
 \leq
 \lambda_3\sum_{j\geq3}\frac{2^jp_*^{\,j-3}}j
 \leq C\lambda_3.
\end{equation}
Also $\norm{\eta_2}_1=2\lambda_2$.  The exact exponential remainder in a
Banach algebra gives
\begin{equation}
 \norm{\exp_*(\eta)-\delta_0-\eta}_1
 \leq\frac12\norm{\eta}_1^2\e^{\norm{\eta}_1}.
\end{equation}
As $\lambda_3\leq p_*\lambda_2$ and
$\lambda_2\leq p_*L$, the exponential factor is bounded by a constant
depending only on $L$. Since $\pi_L$ is a probability measure, the triangle inequality gives
$\norm{\pi_L*\sigma}_1\leq\norm{\pi_L}_1\norm{\sigma}_1
=\norm{\sigma}_1$ for every finite signed measure $\sigma$ on $\mathbb Z$.
Combining this with \cref{eq:eta-ge3} and the preceding exponential
remainder bound, we obtain
\begin{equation}\label{eq:poisson-charlier-signed-remainder}
 \norm{
  \mathcal L(S)-\pi_L*(\delta_0+\eta_2)
 }_1
 \leq C_L(\lambda_3+\lambda_2^2).
\end{equation}
By the transfer identity \eqref{eq:poisson-charlier-transfer},
\begin{equation}\label{eq:poisson-second-mode}
 (\pi_L*\eta_2)(k)
 =
 -\frac{\lambda_2}{2L^2}\pi_L(k)P_2(k;L).
\end{equation}
Thus, for the signed measure of total mass one
$\sigma_2=\pi_L*(\delta_0+\eta_2)$,
\begin{equation}\label{eq:poisson-second-mode-variation}
 \frac12\norm{\sigma_2-\pi_L}_{\mathrm{var}}
 =
 \frac{\lambda_2}{4L^2}\mathbb E|P_2(Z_L;L)|.
\end{equation}
The reverse triangle inequality for the variation norm, together with
\cref{eq:poisson-charlier-signed-remainder,%
eq:poisson-second-mode-variation}, proves
\cref{eq:first-poisson-charlier}.
\end{proof}

\subsection{Configuration-level approximation of CUE blocks}

We next extend the Bernoulli count approximation to the full point
configuration in $B$. Replacing the CUE blocks in \eqref{eq:rains-decomposition} by independent processes
that are empty or contain one uniform point introduces a total
error of $o(\beta_{N,m})$, so it does not affect the leading
total-variation term.

\begin{proposition}[Sparse-block approximation]
\label{prop:sparse-block-approximation}
Let $X_d$ be a CUE block of dimension $d$, unfolded onto
$\mathbb T_N$ and restricted to $B$, and put $w=d/N$.  Assume $Lw\leq1$.
Let $Y_w$ be empty with probability $1-Lw$ and otherwise consist of one
uniform point of $B$.  Then
\begin{equation}\label{eq:sparse-cue-block}
 d_{\TV}\!\left(X_d, Y_w\right)
 \leq\frac{\pi^2L^4}{18}w^4.
\end{equation}
\end{proposition}

\begin{proof}
We first consider a general point-process observation.  Let $X$ be a finite
simple point process on $B$ with intensity measure $\mu$, suppose
$p=\mu(B)\leq1$, and put $M=X(B)$.  Let $Y$ be empty with probability
$1-p$ and otherwise contain one point with law $\mu/p$, with the statement
understood trivially if $p=0$.

Let $Q_n$ be the restriction of the law of $X$ to the cardinality sector
$\{\xi:\xi(B)=n\}$, and let $Q_1^\sharp$ be the corresponding
subprobability measure on $B$ in the one-point sector. The intensity
measure decomposes as
\begin{equation}
 \mu=Q_1^\sharp+\zeta,
 \qquad
 \zeta(A)=\mathbb E[X(A)\mathbf1_{\{M\geq2\}}].
\end{equation}
The singleton-sector measure of $Y$ is $\mu$, whereas that of $X$ is
$Q_1^\sharp$.  Moreover,
\begin{equation}
 Q_0(\mathcal N(B))-(1-p)
 =
 \mathbb E[(M-1)\mathbf1_{\{M\geq2\}}].
\end{equation}
Because different cardinality sectors are disjoint, their variation
norms add, and hence
\begin{align}
 2d_{\TV}(X,Y)
 &=
 \mathbb E[(M-1)\mathbf1_{\{M\geq2\}}]
 +\zeta(B)+\mathbb P(M\geq2) \notag\\
 &=
 2\mathbb E[M\mathbf1_{\{M\geq2\}}].
\end{align}
It follows that
\begin{equation}\label{eq:general-zero-one-bound}
 d_{\TV}(X,Y)
 \leq\mathbb E[M(M-1)].
\end{equation}
For the CUE block, the intensity is $w\,\dd x$, so the zero-or-one process
in this construction is exactly $Y_w$.  The determinantal two-point
formula and \cref{eq:local-block-kernel} give
\begin{equation}\label{eq:cue-factorial-moment}
 \mathbb E[X_d(B)(X_d(B)-1)]
 =
 \int_{B^2}\left(w^2-|K^{\mathrm{CUE}}_{N,d}(x,y)|^2\right)\dd x\,\dd y.
\end{equation}
Writing $t=x-y$ and expanding the squared Dirichlet kernel,
\begin{align}
 w^2-|K^{\mathrm{CUE}}_{N,d}(x,y)|^2
 &=
 \frac2{N^2}\sum_{h=1}^{d-1}
 (d-h)\left\{1-\cos\!\left(\frac{2\pi ht}{N}\right)\right\}
 \notag\\
 &\leq
 \frac{4\pi^2t^2}{N^4}
 \sum_{h=1}^{d-1}(d-h)h^2 \notag\\
 &=
 \frac{\pi^2t^2d^2(d^2-1)}{3N^4}
 \leq\frac{\pi^2t^2}{3}w^4.                     \label{eq:cue-kernel-cancellation}
\end{align}
By rotation invariance and $L<N$, we may represent $B$ by the
interval $[0,L]\subset[0,N)$. All differences $x-y$ in the above formula are then ordinary real differences with $x,y\in[0,L]$.
Finally,
\begin{equation}
 \int_{B^2}(x-y)^2\dd x\,\dd y=\frac{L^4}{6}.
\end{equation}
Substitution into \cref{eq:general-zero-one-bound,%
eq:cue-factorial-moment,eq:cue-kernel-cancellation} proves the result.
\end{proof}
Assume $N>L$ and $L\varepsilon_{N,m}\leq1/4$, as in
\cref{thm:local-poisson}. Then $Lw_a\leq1/4$ for every $a$.
For each block $a$, let $Y_a$ be empty with probability $1-Lw_a$
and otherwise contain one point chosen uniformly from $B$, as in
\cref{prop:sparse-block-approximation}. Choose these processes
independently across blocks, and set
\begin{equation}\label{eq:local-superpositions}
 X=\sum_{a=1}^mX_a=\Xi_{N,m}\!\restriction_B,
 \qquad
 Y=\sum_{a=1}^mY_a,
 \qquad
 S=Y(B).
\end{equation}
The total-variation distance between the sums $X$ and $Y$ is bounded
by that between the corresponding collections of block processes.
Independence and \eqref{eq:sparse-cue-block} therefore give
\begin{align}
 d_{\TV}(X,Y)
 &\leq \sum_{a=1}^m d_{\TV}(X_a,Y_a) \notag\\
 &\leq \frac{\pi^2L^4}{18}\sum_{a=1}^m w_a^4
 =\frac{\pi^2L^4}{18}\gamma_{4;N,m}.
 \label{eq:local-configuration-approximation}
\end{align}

\paragraph{Reduction to the count.}
With $Y$ and $S$ as in \cref{eq:local-superpositions},
\begin{equation}\label{eq:common-conditional-locations}
 d_{\TV}\!\left(Y,\Pi_B\right)
 =
 d_{\TV}\!\left(S,Z_L \right).
\end{equation}
\begin{proof}
Conditional on $S=k$, the $k$ points of $Y$ are independent and uniform
on $B$.  This holds for each fixed set of successful Bernoulli sources and
therefore remains true after mixing over that set.  Conditional on
$Z_L=k$, the Poisson process $\Pi_B$ has the same law.  As distinct
cardinality sectors of $\mathcal N(B)$ are disjoint, the variation
distance between the two mixtures is exactly the variation distance
between their mixing weights.
\end{proof}

\subsection{Uniform conclusion across power scales}

\begin{proof}[Proof of \cref{thm:local-poisson}]
The count $S=Y(B)$ in \cref{eq:local-superpositions} is a sum of
independent Bernoulli variables with success probabilities $p_a=Lw_a$.
Consequently,
\begin{equation}
 \sum_a p_a=L,\qquad
 \lambda_2=L^2\beta_{N,m},\qquad
 \lambda_3=L^3\gamma_{3;N,m}.
\end{equation}
The assumption $L\varepsilon_{N,m}\leq\frac14$ permits the application of
\cref{lem:first-poisson-charlier}.  Together with
the reduction \eqref{eq:common-conditional-locations} and
\cref{eq:local-P2}, it gives
\begin{equation}\label{eq:bernoulli-source-leading}
 d_{\TV}(Y,\Pi_B)
 =
 A(L)\beta_{N,m}
 +O_L(\gamma_{3;N,m}+\beta_{N,m}^2).
\end{equation}
The reverse triangle inequality for a metric, followed by
\cref{eq:local-configuration-approximation}, now yields
\begin{align}
 \left|
 d_{\TV}(X,\Pi_B)
 -A(L)\beta_{N,m}
 \right|
 \leq
 C_L\left(
 \gamma_{3;N,m}+\beta_{N,m}^2+\gamma_{4;N,m}
 \right).                                      \label{eq:local-process-leading}
\end{align}
Let $W=X(B)=\Xi_{N,m}(B)$ be the actual local count.
The total-variation distance between the counts $W$ and $S$ is
bounded by that between the full configurations $X$ and $Y$:
\begin{equation}
 d_{\TV}(W,S)\leq d_{\TV}(X,Y).
\end{equation}
Another application of the reverse triangle inequality, using
\cref{lem:first-poisson-charlier,%
eq:local-configuration-approximation}, gives the corresponding count
bound.  Adding it to \cref{eq:local-process-leading} proves
\cref{eq:main-local-poisson}.

It remains to verify the uniform asymptotic assertion; recall \eqref{eq:local-power-sum-bounds} and the fact that
\begin{equation}\label{eq:local-epsilon-bound}
 \varepsilon_{N,m}
 \leq\frac1m+\frac1N.
\end{equation}
This means that the largest non-empty block has size at most $N/m+1$.
Thus the right side of \cref{eq:main-local-poisson} is at most
$C_L\varepsilon_{N,m}\beta_{N,m}$.  If $N,m\to\infty$, then
$\varepsilon_{N,m}\to0$, with no condition on $m/N$.  Since $A(L)>0$,
division by $\beta_{N,m}$ proves both asserted asymptotics.
\end{proof}

The identity \eqref{eq:local-beta-exact} describes
$\beta_{N,m}$ across all power regimes. If $m=o(N)$ and $m\to\infty$, then
$\beta_{N,m}\sim1/m$.  If $m\asymp N$, the same formula covers every fixed
quotient and transitions between block dimensions.  If $m\geq N$, then
$\beta_{N,m}=1/N$ and the spectrum of $U^m$ is already a binomial point
process of $N$ independent uniform points; its remaining local discrepancy
from Poisson is the fixed-total-number correction.  Fixed $m$ is excluded:
then the local limit is a superposition of finitely many repulsive CUE bulk
processes rather than a Poisson process. Independence of the blocks and \eqref{eq:cue-factorial-moment} give
\begin{align}
 \mathbb E\Xi_{N,m}(B)&=L, \qquad  
 \operatorname{Var}\bigl(\Xi_{N,m}(B)\bigr)
 =L-L^2\beta_{N,m}+O_L(\gamma_{4;N,m}).
 \label{eq:local-count-variance}
\end{align}
The remainder follows from \eqref{eq:cue-kernel-cancellation}
and \eqref{eq:local-power-sums}. As $\gamma_{4;N,m}=o(\beta_{N,m})$ by
\eqref{eq:local-power-sum-bounds},
\begin{equation}
 L-\operatorname{Var}\bigl(\Xi_{N,m}(B)\bigr)
 \sim L^2\beta_{N,m}
 \qquad\text{as }N,m\to\infty.
\end{equation}
In particular, the count variance is smaller than the Poisson
variance $L$ for all sufficiently large $N,m$.
\vspace{-4mm}
\paragraph{Numerical evidence.}
We compute the distribution of the local eigenvalue count $W=\Xi_{N,m}(B)$ from the finite CUE kernels and Rains's decomposition, and examine its deviation from the Poisson probabilities $\pi_L(k)=\mathbb P(Z_L=k)$. For $L=4$ and $N=4m$, the left plot compares the rescaled differences $\beta_{N,m}^{-1}[\mathbb P(W=k)-\pi_L(k)]$ with the predicted leading correction $-\frac12\pi_L(k)[(k-L)^2-k]$. The values for $m=128$ lie closer to this profile than those for $m=16$. The right plot shows $d_{\TV}(W,Z_L)/(A(L)\beta_{N,m})$ approaching one for $L=0.25,1,4$, with both fixed and growing block dimensions. These comparisons support the predicted probability correction and leading total-variation constant for the parameter sequences tested.

\begin{figure}[H]
    \centering
\includegraphics[width=16cm]{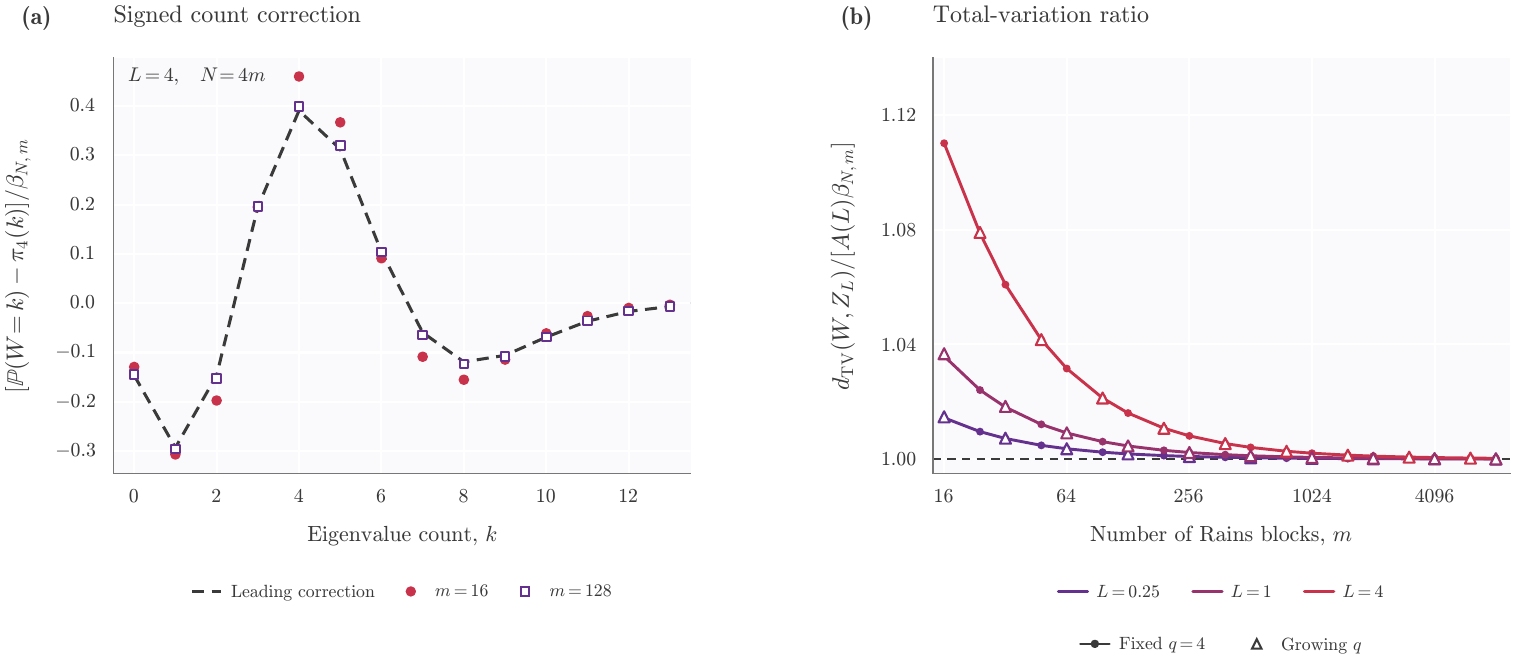}
    \caption{Local Poisson approximation for eigenvalue counts. (a) Rescaled probability differences compared with the predicted first correction (dashed curve). (b) Actual total-variation error divided by $A(L)\beta_{N,m}$, approaching the predicted limit one.}
    \label{fig:numerical-data}
\end{figure}

\section*{Acknowledgments}
This research was supported by the Heilbronn Institute for Mathematical Research. I would like to thank Nina Snaith and Brian Conrey for introducing me to this problem and for the numerous helpful discussions throughout my PhD relating to this topic.

\ifdefined\PedagogicalVersion\PedConclusionGuide\fi

\begingroup
\small
\bibliographystyle{plain}
\bibliography{ref}

@ARTICLE(kn:adlervanmoerbeke01,
    AUTHOR = "M. Adler and P. van Moerbeke",
    TITLE = "Integrals over classical groups, random permutations, {Toda} and {Toeplitz} lattices",
    JOURNAL = "Comm. Pure Appl. Math.",
    VOLUME = "54",
    YEAR = "2001",
    PAGES = "153--205",
    DOI = "10.1002/1097-0312(200102)54:2<153::AID-CPA2>3.0.CO;2-5")

@ARTICLE(kn:arratiabarbourtavare92,
    AUTHOR = "R. Arratia and A. D. Barbour and S. Tavar{\'e}",
    TITLE = "Poisson process approximations for the {Ewens} sampling formula",
    JOURNAL = "Ann. Appl. Probab.",
    VOLUME = "2",
    NUMBER = "3",
    YEAR = "1992",
    PAGES = "519--535",
    DOI = "10.1214/aoap/1177005647"
)

@ARTICLE(kn:basorwidom00,
    AUTHOR = "E. L. Basor and H. Widom",
    TITLE = "On a {Toeplitz} determinant identity of {Borodin} and {Okounkov}",
    JOURNAL = "Integral Equations Operator Theory",
    VOLUME = "37",
    YEAR = "2000",
    PAGES = "397--401",
    DOI = "10.1007/BF01192828")

@ARTICLE(kn:bornemann10,
    AUTHOR = "F. Bornemann",
    TITLE = "On the numerical evaluation of distributions in random matrix theory: A review",
    JOURNAL = "Markov Processes Relat. Fields",
    VOLUME = "16",
    YEAR = "2010",
    PAGES = "803--866")

@ARTICLE(kn:borodinokounkov00,
    AUTHOR = "A. Borodin and A. Okounkov",
    TITLE = "A {Fredholm} determinant formula for {Toeplitz} determinants",
    JOURNAL = "Integral Equations Operator Theory",
    VOLUME = "37",
    YEAR = "2000",
    PAGES = "386--396",
    DOI = "10.1007/BF01192827")

@ARTICLE(kn:bumdia02,
    AUTHOR = "D. Bump and P. Diaconis",
    TITLE = "Toeplitz Minors",
    JOURNAL = "Journal of Combinatorial Theory, Series A",
    VOLUME = "97",
    YEAR = "2002",
    PAGES = "252--271")

@BOOK(kn:chihara1978,
    AUTHOR = "T. S. Chihara",
    TITLE = "An Introduction to Orthogonal Polynomials",
    SERIES = "Mathematics and its Applications",
    VOLUME = "13",
    PUBLISHER = "Gordon and Breach",
    ADDRESS = "New York",
    YEAR = "1978"
)

@ARTICLE(kn:collinssniady06,
    AUTHOR = "B. Collins and P. {\'S}niady",
    TITLE = "Integration with respect to the {H}aar measure on unitary, orthogonal and symplectic group",
    JOURNAL = "Commun. Math. Phys.",
    VOLUME = "264",
    YEAR = "2006",
    PAGES = "773--795",
    DOI = "10.1007/s00220-006-1554-3")

@ARTICLE(kn:courteautjohanssonlambert24,
    AUTHOR = "K. Courteaut and K. Johansson and G. Lambert",
    TITLE = "From {Berry--Esseen} to super-exponential",
    JOURNAL = "Electron. J. Probab.",
    VOLUME = "29",
    YEAR = "2024",
    PAGES = "1--48",
    NOTE = "Paper No. 11",
    DOI = "10.1214/23-EJP1068")

@MISC(kn:decreusefondvasseur18,
    AUTHOR = "L. Decreusefond and A. Vasseur",
    TITLE = "{Stein}'s method and {Papangelou} intensity for {Poisson} or {Cox} process approximation",
    HOWPUBLISHED = "arXiv preprint arXiv:1807.02453",
    YEAR = "2018",
    DOI = "10.48550/arXiv.1807.02453"
)

@ARTICLE(kn:diaconisevans01,
    AUTHOR = "P. Diaconis and S. N. Evans",
    TITLE = "Linear functionals of eigenvalues of random matrices",
    JOURNAL = "Trans. Amer. Math. Soc.",
    VOLUME = "353",
    YEAR = "2001",
    PAGES = "2615--2633",
    DOI = "10.1090/S0002-9947-01-02800-8")

@ARTICLE(kn:diasha,
    AUTHOR = "P. Diaconis and M. Shahshahani",
    TITLE = "On the eigenvalues of random matrices",
    JOURNAL = "Journal of Applied Probability",
    VOLUME = "31A",
    YEAR = "1994",
    PAGES = "49--62",
    DOI = "10.1017/S0021900200106989")

@ARTICLE(kn:doblerstolz11,
    AUTHOR = "C. D{\"o}bler and M. Stolz",
    TITLE = "{Stein}'s method and the multivariate {CLT} for traces of powers on the classical compact groups",
    JOURNAL = "Electron. J. Probab.",
    VOLUME = "16",
    YEAR = "2011",
    PAGES = "2375--2405",
    NOTE = "Paper No. 86",
    DOI = "10.1214/EJP.v16-960")

@ARTICLE(kn:duitsjohansson10,
    AUTHOR = "M. Duits and K. Johansson",
    TITLE = "Powers of large random unitary matrices and {Toeplitz} determinants",
    JOURNAL = "Trans. Amer. Math. Soc.",
    VOLUME = "362",
    YEAR = "2010",
    PAGES = "1169--1187")

@ARTICLE(kn:dyson62iii,
    AUTHOR = "F. J. Dyson",
    TITLE = "Statistical theory of the energy levels of complex systems. {III}",
    JOURNAL = "J. Math. Phys.",
    VOLUME = "3",
    YEAR = "1962",
    PAGES = "166--175",
    DOI = "10.1063/1.1703775")

@BOOK(kn:folland99,
    AUTHOR = "G. B. Folland",
    TITLE = "Real Analysis: Modern Techniques and Their Applications",
    EDITION = "Second",
    PUBLISHER = "John Wiley and Sons, Inc.",
    ADDRESS = "New York",
    YEAR = "1999")

@ARTICLE(kn:forwit04,
    AUTHOR = "P. J. Forrester and N. S. Witte",
    TITLE = "Application of the $\tau$-function theory of {P}ainlev{\'e} equations to random matrices: ${P}_{VI}$, the {JUE}, {C}y{UE}, c{JUE} and scaled limits",
    JOURNAL = "Nagoya Math. J.",
    VOLUME = "174",
    YEAR = "2004",
    PAGES = "29--114")

@ARTICLE(kn:fulman12,
    AUTHOR = "J. Fulman",
    TITLE = "{Stein}'s method, heat kernel, and traces of powers of elements of compact {Lie} groups",
    JOURNAL = "Electron. J. Probab.",
    VOLUME = "17",
    YEAR = "2012",
    PAGES = "1--16",
    NOTE = "Paper No. 66",
    DOI = "10.1214/EJP.v17-2251")

@ARTICLE(kn:gessel90,
    AUTHOR = "I. M. Gessel",
    TITLE = "Symmetric functions and {$P$}-recursiveness",
    JOURNAL = "J. Combin. Theory Ser. A",
    VOLUME = "53",
    YEAR = "1990",
    PAGES = "257--285")

@ARTICLE(kn:jiangmatsumoto15,
    AUTHOR = "T. Jiang and S. Matsumoto",
    TITLE = "Moments of traces of circular beta-ensembles",
    JOURNAL = "Ann. Probab.",
    VOLUME = "43",
    NUMBER = "6",
    YEAR = "2015",
    PAGES = "3279--3336",
    DOI = "10.1214/14-AOP960")

@ARTICLE(kn:johansson97,
    AUTHOR = "K. Johansson",
    TITLE = "On random matrices from the compact classical groups",
    JOURNAL = "Ann. of Math.",
    VOLUME = "145",
    YEAR = "1997",
    PAGES = "519--545",
    DOI = "10.2307/2951843")

@ARTICLE(kn:johanssonlambert21,
    AUTHOR = "K. Johansson and G. Lambert",
    TITLE = "Multivariate normal approximation for traces of random unitary matrices",
    JOURNAL = "Ann. Probab.",
    VOLUME = "49",
    YEAR = "2021",
    PAGES = "2961--3010",
    DOI = "10.1214/21-AOP1520")

@BOOK(kn:macdonald1995,
    AUTHOR = "I. G. Macdonald",
    TITLE = "Symmetric Functions and Hall Polynomials",
    EDITION = "Second",
    PUBLISHER = "Oxford University Press",
    ADDRESS = "Oxford",
    YEAR = "1995")

@BOOK(kn:meckes19,
    AUTHOR = "E. S. Meckes",
    TITLE = "The Random Matrix Theory of the Classical Compact Groups",
    SERIES = "Cambridge Tracts in Mathematics",
    VOLUME = "218",
    PUBLISHER = "Cambridge University Press",
    YEAR = "2019")

@MISC(kn:dlmf10,
    AUTHOR = "{NIST Digital Library of Mathematical Functions}",
    TITLE = "Chapter 10: {Bessel} Functions",
    URL = "https://dlmf.nist.gov/10",
    HOWPUBLISHED = "\url{https://dlmf.nist.gov/10}")

@ARTICLE(kn:pasturvasilchuk04,
    AUTHOR = "L. Pastur and V. Vasilchuk",
    TITLE = "On the moments of traces of matrices of classical groups",
    JOURNAL = "Commun. Math. Phys.",
    VOLUME = "252",
    YEAR = "2004",
    PAGES = "149--166",
    DOI = "10.1007/s00220-004-1231-3")

@ARTICLE(kn:rains97,
    AUTHOR = "E. M. Rains",
    TITLE = "High powers of random elements of compact {L}ie groups",
    JOURNAL = "Prob. Theory Relat. Fields",
    VOLUME = "107",
    YEAR = "1997",
    PAGES = "219--241",
    DOI = "10.1007/s004400050084")

@ARTICLE(kn:rains03,
    AUTHOR = "E. M. Rains",
    TITLE = "Images of eigenvalue distributions under power maps",
    JOURNAL = "Probab. Theory Related Fields",
    VOLUME = "125",
    YEAR = "2003",
    PAGES = "522--538",
    DOI = "10.1007/s00440-002-0245-5")

@ARTICLE(kn:soshnikov00,
    AUTHOR = "A. Soshnikov",
    TITLE = "The central limit theorem for local linear statistics in classical compact groups and related combinatorial identities",
    JOURNAL = "Ann. Probab.",
    VOLUME = "28",
    YEAR = "2000",
    PAGES = "1353--1370",
    DOI = "10.1214/aop/1019160338")

@ARTICLE(kn:soshnikovwu23,
    AUTHOR = "A. Soshnikov and C. Wu",
    TITLE = "A note on cumulant technique in random matrix theory",
    JOURNAL = "Entropy",
    VOLUME = "25",
    NUMBER = "5",
    YEAR = "2023",
    PAGES = "725",
    DOI = "10.3390/e25050725")

@BOOK(kn:szego75,
    AUTHOR = "G. Szeg{\H{o}}",
    TITLE = "Orthogonal Polynomials",
    SERIES = "American Mathematical Society Colloquium Publications",
    VOLUME = "23",
    EDITION = "Fourth",
    PUBLISHER = "American Mathematical Society",
    YEAR = "1975")

@ARTICLE(kn:weingarten78,
    AUTHOR = "D. Weingarten",
    TITLE = "Asymptotic behavior of group integrals in the limit of infinite rank",
    JOURNAL = "J. Math. Phys.",
    VOLUME = "19",
    YEAR = "1978",
    PAGES = "999--1001",
    DOI = "10.1063/1.523807")
\endgroup
\end{document}